\documentclass[a4paper]{article}
\usepackage[english]{babel}
\usepackage{german}
\usepackage{amssymb}
\usepackage{amsfonts}
\usepackage{amsmath}
\usepackage[latin1]{inputenc}
\usepackage{fullpage}
\usepackage{graphicx}
\usepackage{subfigure}
\usepackage{algorithm}
\usepackage[noend]{algorithmic}
\usepackage{amsmath,amssymb,cite}
\usepackage{lineno,footmisc,marvosym}
\usepackage{wasysym,vmargin,stackrel}
\usepackage{color,hyperref}
\usepackage{boxedminipage}
\usepackage{xspace}
\usepackage{todonotes}
\usepackage{fancyhdr}

\newtheorem{theorem}{Theorem}

\newtheorem{corollary}{Corollary}

\newtheorem{definition}{Definition}
\newtheorem{observation}{Observation}

\newtheorem{lemma}{Lemma}
\newtheorem{notation}{Notation}

\newenvironment{proof}[1][Proof]{\noindent\textbf{#1.} }{\ \rule{0.5em}{0.5em}}

\newtheorem{redrule}{Reduction Rule}
\setmarginsrb{3.5cm}{2.25cm}{3.5cm}{3.05cm}{0.3cm}{0.3cm}{-0.3cm}{1.0cm}
\begin{document}

\title{\vspace{-0.5cm}Polynomial Fixed-Parameter Algorithms: \\
A Case Study for Longest Path on Interval Graphs\thanks{%
A preliminary conference version of this work appears in the \emph{%
Proceedings of the 10th International Symposium on Parameterized and Exact
Computation (IPEC)}, Patras, Greece, September 2015, pages~102--113~\cite%
{Gian-Mertz-Niederm-IPEC15}.} \ \thanks{%
Partially supported by the EPSRC Grant EP/K022660/1.}}
\author{Archontia C.~Giannopoulou\thanks{%
Institut f\"ur Softwaretechnik und Theoretische Informatik, TU Berlin, Germany. 
Email: \texttt{archontia.giannopoulou@gmail.com}} \ \thanks{%
The main part of this paper was prepared while the author was affiliated at
the School of Engineering and Computing Sciences, Durham University, UK.}
\and George B.~Mertzios\thanks{%
School of Engineering and Computing Sciences, Durham University, UK. Email: 
\texttt{george.mertzios@durham.ac.uk}} 
\and Rolf Niedermeier\thanks{%
Institut f\"ur Softwaretechnik und Theoretische Informatik, TU Berlin,
Germany. Email: \texttt{rolf.niedermeier@tu-berlin.de}}}
\date{\vspace{-1.0cm}}
\maketitle

\begin{abstract}
We study the design of fixed-parameter algorithms for problems already known
to be solvable in polynomial time. The main motivation is to get more
efficient algorithms for problems with unattractive polynomial running
times. Here, we focus on a fundamental graph problem: \textsc{Longest Path},
that is, given an undirected graph, find a maximum-length path in~$G$. 
\textsc{Longest Path} is NP-hard in general but known to be solvable in $%
O(n^{4})$~time on $n$-vertex interval graphs. We show how to solve \textsc{%
Longest Path on Interval Graphs}, parameterized by vertex deletion number~$k$
to proper interval graphs, in $O(k^{9}n)$ time. Notably, \textsc{Longest Path%
} is trivially solvable in linear time on proper interval graphs, and the
parameter value~$k$ can be approximated up to a factor of~4 in linear time.
From a more general perspective, we believe that using parameterized
complexity analysis 
may enable a refined understanding of efficiency aspects for polynomial-time
solvable problems similarly to what classical parameterized complexity
analysis does for NP-hard problems.\newline

\end{abstract}

\section{Introduction\label{introduction-sec}}

Parameterized complexity analysis~\cite%
{Cyganetal15,DowneyF13,FG06,Niedermeierbook06} is a flourishing field
dealing with the exact solvability of NP-hard problems. The key idea is to
lift classical complexity analysis, rooted in the P~versus~NP phenomenon,
from a one-dimensional to a two- (or even multi-)dimensional perspective,
the key concept being ``fixed-parameter tractability (FPT)''. But why should
this natural and successful approach be limited to intractable
(i.e.,~NP-hard) problems? We are convinced that appropriately parameterizing
polynomially solvable problems sheds new light on what makes a problem far
from being solvable in \emph{linear} time, in the same way as classical FPT
algorithms help in illuminating what makes an NP-hard problem far from being
solvable in \emph{polynomial} time. 
In a nutshell, the credo and leitmotif of this paper is that ``FPT
inside~P'' is a very interesting, but still too~little explored, line of
research.

The known results fitting under this leitmotif are somewhat scattered around
in the literature and do not systematically refer to or exploit the toolbox
of parameterized algorithm design. This should change and ``FPT inside~P'' 
should be placed on a much wider footing, using parameterized algorithm
design techniques such as data reduction and kernelization. 
As a simple illustrative example, consider the \textsc{Maximum Matching}
problem. By following a ``Buss-like'' kernelization (as is standard
knowledge in parameterized algorithmics~\cite{DowneyF13,Niedermeierbook06})
and then applying a known polynomial-time matching algorithm, it is not
difficult to derive an efficient algorithm that, given a graph $G$ with $n$
vertices, computes a matching of size at least~$k$ in $O(kn+k^{3})$ time.
For the sake of completeness we present the details of this algorithm in
Section~\ref{matching-kernelization-sec}.

More formally, and somewhat more generally, we propose the following
scenario. Given a problem with instance size~$n$ for which there exists an $%
O(n^{c})$-time algorithm, our aim is to identify appropriate parameters $k$
and to derive algorithms with time complexity ${f(k)\cdot n^{c^{\prime}}}$
such that ${c^{\prime}<c}$, where $f(k)$ depends only on~$k$. First we
refine the class FPT by defining, for every polynomially-bounded function $%
p(n)$, the class FPT$(p(n))$ containing the problems solvable in ${f(k)\cdot
p(n)}$ time, where $f(k)$ is an arbitrary (possibly exponential) function of~%
$k$. It is important to note that, in strong contrast to FPT algorithms for
NP-hard problems, here the function $f(k)$ may also become \emph{polynomial}
in~$k$. %as the initial time complexity $O(n^{c})$ is polynomial.
Motivated by this, we refine the complexity class~P by introducing, for
every polynomial function~$p(n)$, the class \emph{P-FPT}$(p(n))$ \emph{%
(Polynomial~Fixed-Parameter Tractable)}, containing the problems solvable in 
${O(k^{t}\cdot p(n))}$ time for some constant ${t\geq 1}$, i.e., the
dependency of the complexity on the parameter~$k$ is at most polynomial. In
this paper we focus our attention on the (practically perhaps most
attractive)~subclass~\emph{PL-FPT} \emph{(Polynomial-Linear Fixed-Parameter
Tractable)}, where PL-FPT = P-FPT$(n)$. 
For example, the algorithm we sketched above for \textsc{Maximum Matching},
parameterized by solution size~$k$, yields containment in the class PL-FPT.

In an attempt to systematically follow the leitmotif ``FPT~inside~P'', we
put forward three desirable algorithmic properties: 

\begin{enumerate}
\item The running time should have a polynomial dependency on the parameter.

\item The running time should be as close to linear as possible if the
parameter value is constant, improving upon an existing ``high-degree''
polynomial-time (unparameterized) algorithm.

\item The parameter value, or a good approximation thereof, 
should be computable efficiently (preferably in linear time) for arbitrary
parameter values.
\end{enumerate}

In addition, as this research direction is still only little explored, we
suggest to focus first on problems for which the best known upper bounds of
the time complexity are polynomials of high degree, e.g., $O(n^4)$ or higher.

\medskip \noindent\textbf{Related work.} Here we discuss previous work on 
\emph{graph problems} that fits under the leitmotif ``FPT inside~P'';
however, there exists further related work also in other topics such as
string matching~\cite{AmirLP04}, XPath query evaluation in XML databases~%
\cite{Bojanczyk-XPath-2011}, and Linear Program solving~\cite{Megiddo84}.

The complexity of some known polynomial-time algorithms can be easily
``tuned'' with respect to specific parameters, thus immediately reducing the
complexity whenever these parameters are bounded. For instance, in $n$%
-vertex and $m$-edge graphs with nonnegative edge weights, Dijkstra's $%
O(m+n\log n)$-time algorithm for computing shortest paths can be adapted to
an $O(m + n\log k)$-time algorithm, where $k$ is the number of distinct edge
weights~\cite{KoutisMP11} (also refer to~\cite{OrlinMSW10}). In addition,
motivated by the quest for explaining the efficiency of several shortest
path heuristics for road networks (where Dijkstra's algorithm is too slow
for routing applications), the ``highway dimension'' was introduced~\cite%
{AbrahamFGW10} as a parameterization helping to do rigorous proofs about the
quality of the heuristics. Altogether, the work on shortest path
computations shows that, despite of known quasi-linear-time algorithms,
adopting a parameterized view may be of significant (practical) interest.

Maximum flow computations constitute another important application area for
``FPT inside~P''. An $O(k^3 n\log n)$-time maximum flow algorithm was
presented~\cite{HochsteinW07} for graphs that can be made planar by deleting~%
$k$~``crossing edges''; notably, here it is assumed that the embedding and
the $k$~crossing edges are given along with the input. An $O(g^8n\log^2
n\log^2 C)$-time maximum flow algorithm was developed~\cite{ChambersEN12},
where $g$ is the genus of the graph and $C$~is the sum of all edge
capacities; here it is also assumed that the embedding and the parameter~$g$
are given in the input. Finally, we remark that multiterminal flow~\cite%
{HagerupKNR98} and Wiener index computations~\cite{CabelloK09} have
exploited the treewidth parameter, assuming that the corresponding tree
decomposition of the graph is given. However, in both publications~\cite%
{HagerupKNR98,CabelloK09} the dependency on the parameter~$k$ is \emph{%
exponential}.

We finally mention that, very recently, two further works delved deeper into ``FPT inside~P'' 
algorithms for \textsc{Maximum Matching}~\cite{FominLPSW_Matching-Treewidth_SODA17,MertziosNN16}.

\medskip \noindent\textbf{Our contribution.} In this paper, to illustrate
the potential algorithmic challenges posed by the ``FPT inside~P'' framework
(which seem to go clearly beyond the known ``FPT inside~P'' examples), we
focus on \textsc{Longest Path on Interval Graphs}, which is known to be
solvable in $O(n^{4})$ time~\cite{longest-int-algo}, and we derive a
PL-FPT-algorithm (with the appropriate parameterization) that satisfies all
three desirable algorithmic properties described above.

The \textsc{Longest Path} problem asks, given an undirected graph $G$, to
compute a maximum-length path in~$G$. On general graphs, the decision
variant of \textsc{Longest Path} is NP-complete and many FPT algorithms have
been designed for it, e.g.,~\cite{AlonYZ95,ChenLSZ07,KneisMRR06,Williams09,FominLS14,BjorklundHKK-JCSS17},
contributing to the parameterized algorithm design toolkit techniques such
as color-coding~\cite{AlonYZ95} (and further randomized techniques~\cite%
{ChenLSZ07,KneisMRR06}) as well as algebraic approaches~\cite{Williams09}. 
The currently best known deterministic FPT algorithm runs in $O(2.851^k n\log^2 n \log W)$ time on weighted graphs with maximum edge weight $W$, 
where $k$ is the number of vertices in the path~\cite{FominLS14}, 
while the currently best known randomized FPT algorithm runs in $O(1.66^k n^{O(1)})$ time with constant, one-sided error~\cite{BjorklundHKK-JCSS17}. 
\textsc{Longest Path} is known to be solvable in polynomial time only on
very few non-trivial graph classes~\cite{longest-int-algo,MertziosCorneil12}
(see also~\cite{Uehara07} for much smaller graph classes). 
This problem has also been studied on directed graphs; 
a polynomial-time algorithm was given by Gutin~\cite{GutinSidma93} for the class of orientations of
multipartite tournaments, which was later extended
by Bang-Jensen and Gutin~\cite{Bang-JensenG97a}.
With respect to undirected graphs, a few years ago it was shown that \textsc{Longest Path on
Interval Graphs} can be solved in polynomial time, providing an algorithm
that runs in $O(n^{4})$ time~\cite{longest-int-algo}; this algorithm has
been extended with the same running time to the larger class of
cocomparability graphs~\cite{MertziosCorneil12} using a lexicographic depth
first search (LDFS) approach. Notably, a longest path in a \emph{proper}
interval graph can be computed by a \emph{trivial} linear-time algorithm
since every connected proper interval graph has a Hamiltonian path~\cite%
{Bertossi83}. Consequently, as the classes of interval graphs and of proper
interval graphs seem to behave quite differently, it is natural to
parameterize \textsc{Longest Path on Interval Graphs} by the size~$k$ of a 
\emph{minimum proper interval (vertex) deletion set}, i.e., by the minimum
number of vertices that need to be deleted to obtain a proper interval
graph. That is, this parameterization exploits what is also known as
``distance from triviality''~\cite{Cai03,GHN04,FJR13,Nie10} in the sense
that the parameter $k$ measures how far a given input instance is from a
trivially solvable special case. As it turns out, one can compute a $4$%
-approximation of~$k$ in~${O(n+m)}$ time for an interval graph with $n$%
~vertices and $m$~edges. Using this constant-factor approximation of $k$, we
provide a polynomial fixed-parameter algorithm that runs in $O(k^{9}n)$
time, thus proving that \textsc{Longest Path on Interval Graphs} is in the
class PL-FPT when parameterized by the size of a minimum proper interval
deletion set.

To develop our algorithm, we first introduce in Section~\ref%
{data-reductions-sec} two data reduction rules on interval graphs. Each of
these reductions shrinks the size of specific vertex subsets, called \emph{%
reducible} and \emph{weakly reducible} sets, respectively. Then, given any
proper interval deletion set $D$ of an interval graph $G$, in Section~\ref%
{special-interval-sec} we appropriately decompose the graph $G\setminus D$
into two collections $\mathcal{S}_{1}$ and $\mathcal{S}_{2}$ of reducible
and weakly reducible sets, respectively, on which we apply the reduction
rules of Section~\ref{data-reductions-sec}. The resulting interval graph $%
\widehat{G}$ is \emph{weighted} (with weights on its vertices) and has some
special properties; we call $\widehat{G}$ a \emph{special weighted interval
graph} with \emph{parameter} $\kappa $, where in this case $\kappa =O(k^{3})$%
. Notably, although~$\widehat{G}$ has reduced size, it still has $O(n)$
vertices. Then, in Section~\ref{algorithm-sec} we present a fixed-parameter
algorithm (with parameter $\kappa $) computing in $O(\kappa ^{3}n)$ time the
maximum weight of a path in a special weighted interval graph. We note here
that such a maximum-weight path in a special weighted interval graph can be
directly mapped back to a longest path in the original interval graph. Thus,
our parameterized algorithm computes a longest path in the initial interval
graph~$G$ in $O(\kappa ^{3}n)=O(k^{9}n)$ time.

Turning our attention away from \textsc{Longest Path on Interval Graphs} we
present for the sake of completeness our ``Buss-like'' kernelization of the 
\textsc{Maximum Matching} problem in Section~\ref{matching-kernelization-sec}%
. Using this kernelization an efficient algorithm can be easily deduced
which, given an arbitrary graph~$G$ with $n$~vertices, computes a matching
of size at least~$k$ in~$G$ in $O(kn+k^{3})$ time. Finally, in the
concluding Section~\ref{outlook-sec} we discuss our contribution and provide
a brief outlook for future research directions.

\medskip \noindent \textbf{Notation.} 
We consider finite, simple, and undirected graphs. Given a graph $G$, we
denote by $V(G)$ and $E(G)$ the sets of its vertices and edges,
respectively. A graph $G$ is \emph{weighted} if it is given along with a
weight function $w:V(G)\rightarrow \mathbb{N}$ on its vertices. An edge
between two vertices $u$ and $v$ of a graph $G=(V,E)$ is denoted by $uv$,
and in this case $u$ and $v$ are said to be \emph{adjacent}. The \emph{%
neighborhood} of a vertex $u\in V$ is the set $N(u)=\{v\in V\ |\ uv\in E\}$
of its adjacent vertices. The cardinality of~$N(u)$ is the degree~$\deg (u)$
of~$u$. For every subset $S\subseteq V$ we denote by $G[S]$ the subgraph of~$%
G$ induced by the vertex set $S$ and we define $G\setminus S=G[V\setminus S]$%
. A set $S\subseteq V$ induces an \emph{independent set} (resp.~a \emph{%
clique}) in $G$ if $uv\notin E$ (resp.~if $uv\in E$) for every pair of
vertices $u,v\in S$. Furthermore, $S$ is a \emph{vertex cover} if and only
if $V\setminus S$ is an independent set. For any two graphs $G_{1},G_{2}$,
we write $G_{1}\subseteq G_{2}$ if $G_{1}$ is an induced subgraph of $G_{2}$%
. A \emph{matching} $M$ in a graph $G$ is a set of edges of $G$ without
common vertices. 
All paths considered in this paper are simple. Whenever a path $P$ visits
the vertices ${v_{1},v_{2},\ldots ,v_{k}}$ in this order, we write $P=({%
v_{1},v_{2},\ldots ,v_{k})}$. Furthermore, for two vertex-disjoint paths $P={%
(a,\ldots ,b)}$ and $Q={(c,\ldots ,d)}$ where $bc\in E$, we denote by $(P,Q)$
the path $({a,\ldots ,b,c,\ldots ,d})$.

A graph $G=(V,E)$ is an \emph{interval} graph if each vertex $v\in V$ can be
bijectively assigned to a closed interval $I_{v}$ on the real line, such
that $uv\in E$ if and only if $I_{u}\cap I_{v}\neq \emptyset $, and then the
collection of intervals $\mathcal{I}=\{I_{v}:v\in V\}$ is called an \emph{%
interval representation} of $G$. The interval graph $G$ is a \emph{proper
interval graph} if it admits an interval representation $\mathcal{I}$ such
that $I_{u}\nsubseteq I_{v}$ for every $u,v\in V$, and then $\mathcal{I}$ is
called a \emph{proper interval representation}. Given an interval graph $%
G=(V,E)$, a subset $D\subseteq V$ is a \emph{proper interval deletion set}
of~$G$ if $G\setminus D$ is a proper interval graph. The \emph{proper
interval deletion number} of~$G$ is the size of the smallest proper interval
deletion set. Finally, for any positive integer $t$, we denote $%
[t]=\{1,2,\ldots ,t\}$.

\section{Data reductions on interval graphs\label{data-reductions-sec}}

In this section we present two data reductions on interval graphs. The first
reduction (cf.~Section~\ref{first-data-reduction-subsec}) shrinks the size
of a collection of vertex subsets of a certain kind, called \emph{reducible}
sets, and it produces a \emph{weighted} interval graph. The second reduction
(cf.~Section~\ref{second-data-reduction-subsec}) is applied to an arbitrary
weighted interval graph; it shrinks the size of a collection of another kind
of vertex subsets, called \emph{weakly reducible} sets, and it produces a
smaller weighted interval graph. Both reductions retain as an invariant the
maximum path weight. The proof of this invariant is based on the crucial
notion of a \emph{normal} path in an interval graph (cf.~Section~\ref%
{normal-paths-subsec}). The following vertex ordering characterizes interval
graphs~\cite{Ramalingam88}. Moreover, given an interval graph $G$ with~$n$
vertices and $m$ edges, this vertex ordering of $G$ can be computed in $%
O(n+m)$ time~\cite{Ramalingam88}.

\begin{lemma}[\hspace{-0.001cm}\protect\cite{Ramalingam88}]
\label{lem:ordering.edges} A graph $G$ is an interval graph if and only if
there is an ordering $\sigma $ (called \emph{right-endpoint ordering}) of~$%
V(G)$ such that for all $u<_{\sigma }v<_{\sigma }z$, if $uz\in E(G)$ then
also $vz\in E(G)$.
\end{lemma}

In the remainder of the paper we assume that we are given an interval graph~$%
G$ with $n$~vertices and $m$~edges as input, together with an interval
representation~$\mathcal{I}$ of~$G$, where the endpoints of the intervals
are given sorted increasingly. Without loss of generality, we assume that
the endpoints of all intervals are distinct. For every vertex $v\in V(G)$ we
denote by $I_{v}=[l_{v},r_{v}]$ the interval of $\mathcal{I}$ that
corresponds to $v$, i.e.,~$l_{v}$ and~$r_{v}$ are the left and the right
endpoint of~$I_{v}$, respectively. In particular, $G$ is assumed to be given
along with the \emph{right-endpoint ordering} $\sigma $ of its vertices~$%
V(G) $, i.e.,~${u<_{\sigma }v}$ if and only if ${r_{u}<r_{v}}$ in the
interval representation~$\mathcal{I}$ (see also Lemma~\ref%
{lem:ordering.edges}). Given a set $S\subseteq V(G)$, we denote by $\mathcal{%
I}[S]$ the interval representation induced from $\mathcal{I}$ on the
intervals of the vertices of~$S$. We say that two vertices $u_{1},u_{2}\in S$
are \emph{consecutive in~$S$} (with respect to the vertex ordering~$\sigma $%
) if $u_{1}<_{\sigma }u_{2}$ and for every vertex $u\in S\setminus
\{u_{1},u_{2}\}$ either $u<_{\sigma }u_{1}$ or $u_{2}<_{\sigma }u$.
Furthermore, for two sets $S_{1},S_{2}\subseteq V(G)$, we write $%
S_{1}<_{\sigma }S_{2}$ whenever $u<_{\sigma }v$ for every $u\in S_{1}$ and $%
v\in S_{2}$. Finally, we denote by $\mathbf{span}(S)$ the interval $[\min
\{l_{v}:~v\in S\},\max \{r_{v}:~v\in S\}]$.

It is well known that an interval graph $G$ is a proper interval graph if
and only if~$G$ is $K_{1,3}$-free, i.e.,~if~$G$ does not include the claw 
$K_{1,3}$ with four vertices (cf.~Figure~\ref{claw-K-1-3-fig}) as an induced
subgraph~\cite{Roberts69}. It is worth noting here that, to the best of our
knowledge, it is unknown whether a minimum proper interval deletion set of
an interval graph $G$ can be computed in polynomial time. 
However, since there is a unique forbidden induced subgraph $K_{1,3}$ on four vertices, 
we can apply Cai's generic algorithm~\cite{Cai96} on an arbitrary given
interval graph $G$ with $n$ vertices to compute a proper interval deletion
set of~$G$ of size at most $k$ in FPT time $4^{k}\cdot \text{poly}(n)$. 
The main idea of Cai's bounded search tree algorithm in our case is that we repeat the following two steps until we either get 
a $K_{1,3}$-free graph or have used up the deletion of $k$ vertices from $G$: 
(i)~detect an induced $K_{1,3}$, (ii)~branch on deleting one of the four vertices of the detected $K_{1,3}$. 
At every iteration of this process we have four possibilities on the next vertex to delete. 
Thus, since we can delete up to $k$ vertices in total, 
the whole process finishes after at most $4^{k}\cdot \text{poly}(n)$ steps.

As we prove in the next theorem, a $4$-approximation of the minimum proper
interval deletion number of an interval graph can be computed much more
efficiently.

\begin{figure}[t]
\centering%
\subfigure[]{ \label{claw-K-1-3-fig-1}
\includegraphics[scale=0.6]{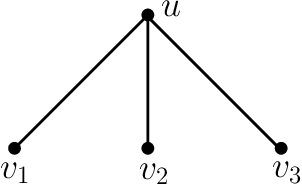}} \hspace{2,0cm} 
\subfigure[]{ \label{claw-K-1-3-fig-2}
\includegraphics[scale=0.6]{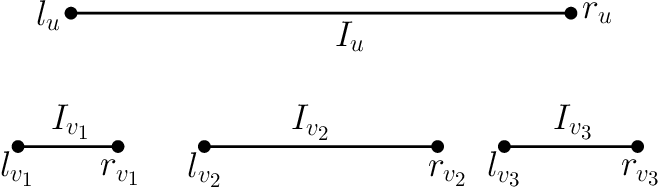}}
\caption{(a) The forbidden induced subgraph (claw $K_{1,3}$) for an interval
graph to be a proper interval graph and (b)~an interval representation of
the $K_{1,3}$.}
\label{claw-K-1-3-fig}
\end{figure}

\begin{theorem}
\label{proper-deletion-set-approximation-thm}Let $G=(V,E)$ be an interval
graph, where $|V|=n$ and $|E|=m$. Let $k$ be the size of the minimum proper
interval deletion set of~$G$. Then a proper interval deletion set $D$ of
size at most $4k$ can be computed in $O(n+m)$ time.
\end{theorem}

\begin{proof}
Let $\{u,v_{1},v_{2},v_{3}\}$ be a set of four vertices that induces a $%
K_{1,3}$ in~$G$ such that $v_{1},v_{2},v_{3}\in N(u)$, $v_{1}<_{\sigma
}v_{2}<_{\sigma }v_{3}$, and $I_{v_{2}}\subseteq I_{u}$ in the interval
representation $\mathcal{I}$ (cf.~Figure~\ref{claw-K-1-3-fig}). Let $%
v_{1}^{\prime }$ (resp.~$v_{3}^{\prime }$) be the neighbor of vertex $u$
with the leftmost right endpoint $r_{v_{1}^{\prime }}$ (resp.~with the
rightmost left endpoint $l_{v_{3}^{\prime }}$) in the representation $%
\mathcal{I}$, i.e.,~$r_{v_{1}^{\prime }}=\min \{r_{v}:v\in N(u)\}$ and $%
l_{v_{3}^{\prime }}=\max \{l_{v}:v\in N(u)\}$. Then note that the set $%
\{u,v_{1}^{\prime },v_{2},v_{3}^{\prime }\}$ also induces a $K_{1,3}$ in~$G$.

We now describe an $O(n+m)$-time algorithm that iteratively detects an
induced $K_{1,3}$ and removes its vertices from the current graph. During
its execution the algorithm maintains a set $D$ of ``marked'' vertices; a
vertex is marked if it has been removed from the graph at a previous
iteration. Initially, $D=\emptyset $, i.e., all vertices are unmarked. The
algorithm processes once every vertex $u\in V$ in an arbitrary order. If $%
u\in D$ (i.e., if $u$ has been marked at a previous iteration), then the
algorithm ignores $u$ and proceeds with the next vertex in $V$. If $u\notin
D $ (i.e., if $u$ is unmarked), then the algorithm iterates for every vertex 
$v\in N(u)\setminus D$ and it computes the vertices $z_{1}(u),z_{2}(u)\in
N(u)\setminus D$ such that $r_{z_{1}(u)}=\min \{r_{v}:v\in N(u)\setminus D\}$
and $l_{z_{2}(u)}=\max \{l_{v}:v\in N(u)\setminus D\}$. In the case where $%
N(u)\setminus D=\emptyset $, the algorithm defines $z_{1}(u)=z_{2}(u)=u$.
Then the algorithm iterates once again for every vertex $v\in N(u)\setminus
D $ and it checks whether the set $\{u,v,z_{1}(u),z_{2}(u)\}$ induces a $%
K_{1,3}$ in~$G$. If it detects at least one vertex $v\in N(u)\setminus D$
such that $\{u,v,z_{1}(u),z_{2}(u)\}$ induces a $K_{1,3}$, then it marks all
four vertices $\{u,v,z_{1}(u),z_{2}(u)\}$, i.e.,~it adds these vertices to
the set $D$. Otherwise the algorithm proceeds with processing the next
vertex of $V$. It is easy to check that every vertex $u\in V$ is processed
by this algorithm in $O(\deg (u))$ time, and thus all vertices of $V$ are
processed in $O(n+m)$ time in total.

The algorithm terminates after it has processed all vertices of $V$ and it
returns the computed set $D$ of all quadruples of marked vertices. Note that
there are $\frac{|D|}{4}$ such quadruples. This set $D$ is clearly a proper
interval deletion set of~$G$, since $G\setminus D$ does not contain an
induced $K_{1,3}$, i.e.,~$k\leq |D|$. In addition, each of the detected
quadruples of the set $D$ induces a $K_{1,3}$ in the initial interval graph $%
G$, and thus any minimum proper interval deletion set must contain at least
one vertex from each of these quadruples, i.e.,~$k\geq \frac{|D|}{4}$.
Summarizing $k\leq |D|\leq 4k$.
\end{proof}

\medskip

Note that, whenever four vertices induce a claw $K_{1,3}$ in an interval
graph~$G$, then in the interval representation $\mathcal{I}$ of $G$ at least
one of these intervals is necessarily properly included in another one
(e.g.,~ $I_{v_{2}}\subseteq I_{u}$ in Figure~\ref{claw-K-1-3-fig-2}).
However the converse is not always true, as there may exist two vertices $%
u,v $ in $G$ such that $I_{v}\subseteq I_{u}$, although $u$ and $v$ do not
belong to any induced claw $K_{1,3}$ in $G$. 

\begin{definition}
\label{semi-proper-representation-def}Let $G=(V,E)$ be an interval graph. An
interval representation $\mathcal{I}$ of $G$ is \emph{semi-proper} when, for
any $u,v\in V$:

\begin{itemize}
\item if $I_{v}\subseteq I_{u}$ in $\mathcal{I}$ then the vertices $u$ and $%
v $ belong to an induced claw $K_{1,3}$ in $G$, i.e.~$\{u,v,a,b\}$ induces a
claw $K_{1,3}$ in $G$ for some vertices $a,b$.
\end{itemize}
\end{definition}

Every interval representation $\mathcal{I}$ of a graph $G$ can be
efficiently transformed into a semi-proper representation $\mathcal{I}%
^{\prime }$ of $G$, as we prove in the next theorem. In the remainder of the
paper we always assume that this preprocessing step has been already applied
to $\mathcal{I}$.

\begin{theorem}[preprocessing]
\label{interval-representation-preprocessing-thm}Given an interval
representation $\mathcal{I}$, a semi-proper interval representation $%
\mathcal{I}^{\prime }$ can be computed in~$O(n+m)$ time. 
\end{theorem}

\begin{proof}
Similarly to the proof of Theorem~\ref{proper-deletion-set-approximation-thm}%
, the algorithm first computes for every vertex $u\in V$ the vertices $%
z_{1}(u),z_{2}(u)\in N(u)$, such that $r_{z_{1}(u)}=\min \{r_{v}:v\in N(u)\}$
and $l_{z_{2}(u)}=\max \{l_{v}:v\in N(u)\}$. If $N(u)=\emptyset$, then the
algorithm defines $z_{1}(u)=z_{2}(u)=u$.

The algorithm iterates over all $u\in V$. For each $u\in V$, the algorithm
iterates over all $v\in N(u)$ such that $I_{v}\subseteq I_{u}$ in the
current interval representation. Let these vertices be $\{v_{1},v_{2},\ldots
,v_{t}\}$, where $l_{v_{1}}<l_{v_{2}}<\ldots <l_{v_{t}}$. The algorithm
processes the vertices $\{v_{1},v_{2},\ldots ,v_{t}\}$ in this order. For
every $i\in \{1,2,\ldots ,t\}$, if $z_{2}(u)\in N(v_{i})$, then the
algorithm increases the right endpoint of~$I_{v_{i}}$ to the point $%
r_{u}+\varepsilon _{i}$, for an appropriately small $\varepsilon _{i}>0$.
The algorithm chooses the values of~$\varepsilon _{i}$ such that $%
\varepsilon _{1}<\varepsilon _{2}<\ldots <\varepsilon _{t}$. By performing
these operations no new adjacencies are introduced, and thus the resulting
interval representation remains a representation of the same interval graph $%
G$.

We note here that the algorithm can be efficiently implemented (i.e., in $%
O(n+m)$ time in total) without explicitly computing the values of these~$%
\varepsilon _{i}$, as follows. Since the endpoints of the $n$ intervals of $%
\mathcal{I}$ are assumed to be given increasingly sorted, we initially scan
them from left to right and map them bijectively to the integers $%
\{1,2,\ldots ,2n\}$. Then, instead of increasing the right endpoint of~$%
I_{v_{i}}$ to the point $r_{u}+\varepsilon _{i}$ as described above, where $%
i\in \{1,2,\ldots ,t\}$, we just store the vertices $\{v_{1},v_{2},\ldots
,v_{t}\}$ (in this order) in a linked list after the endpoint $r_{u}$. At
the end of the whole process (i.e., after dealing with all pairs of vertices 
$u,v$ such that $v\in N(u)$ and $I_{v}\subseteq I_{u}$ in the interval
representation $\mathcal{I}$), we scan again all interval endpoints from
left to right and re-map them bijectively to the integers $\{1,2,\ldots
,2n\} $, where in this new mapping we place the endpoints $%
\{r_{v_{1}},r_{v_{2}},\ldots ,r_{v_{t}}\}$ (in this order) immediately after 
$r_{u}$. This can be clearly done in $O(n+m)$ time.

Then the algorithm iterates (again) over all $v\in N(u)$ such that $%
I_{v}\subseteq I_{u}$ in the current interval representation. Let these
vertices be $\{v_{1},v_{2},\ldots ,v_{t^{\prime }}\}$, where $%
r_{v_{1}}>r_{v_{2}}>\ldots >r_{v_{t^{\prime }}}$. The algorithm processes
the vertices $\{v_{1},v_{2},\ldots ,v_{t^{\prime }}\}$ in this order. For
every $i\in \{1,2,\ldots ,t^{\prime }\}$, if $z_{2}(u)\notin N(v_{i})$ and $%
z_{1}(u)\in N(v_{i})$, then the algorithm decreases the left endpoint of~$%
I_{v_{i}}$ to the point $l_{u}-\varepsilon _{i}$, for an appropriately small 
$\varepsilon _{i}>0$. The algorithm chooses the values of~$\varepsilon _{i}$
such that $\varepsilon _{1}<\varepsilon _{2}<\ldots <\varepsilon _{t}$.
Similarly to the above, no new adjacencies are introduced by performing
these operations, and thus the resulting interval representation remains a
representation of the same interval graph $G$. Furthermore, the algorithm
can be efficiently implemented (i.e., in $O(n+m)$ time in total) without
explicitly computing these values of $\varepsilon _{i}$, similarly to the
description in the previous paragraph.

Denote by $\mathcal{I}^{\prime }$ the resulting interval representation of~$%
G $, which is obtained after performing all the above operations.
Furthermore denote by $\sigma ^{\prime }$ the right-endpoint ordering of the
intervals in~$\mathcal{I}^{\prime }$. Let $u,v\in V$. It can be easily
checked that, if $I_{v}\subseteq I_{u}$ in~$\mathcal{I}^{\prime }$, then
also $I_{v}\subseteq I_{u}$ in the initial representation $\mathcal{I}$.
Furthermore, it follows directly by the above construction that, if $%
I_{v}\subseteq I_{u}$ in~$\mathcal{I}^{\prime }$, then $z_{1}(u),z_{2}(u)%
\notin N(v)$, where $z_{1}<_{\sigma ^{\prime }}v<_{\sigma ^{\prime }}z_{2}$,
and thus the vertices $\{u,v,z_{1}(u),z_{2}(u)\}$ induce a $K_{1,3}$ in~$G$.

The computation of the vertices $z_{1}(u),z_{2}(u)$ for all vertices $u\in V$
can be done in~$O(n+m)$ time. Furthermore, for every $u\in V$ we can visit
all vertices $v\in N(u)$ in~$O(\deg (u))$ time in the above algorithm, since
the endpoints of the intervals are assumed to be given sorted in increasing
order. For every such edge $uv\in E$, where $I_{v}\subseteq I_{u}$, we can
check in~$O(1)$ time whether $z_{1}(u)\in N(v)$ (resp.~whether $z_{2}(u)\in
N(v)$) by checking whether $r_{v}<l_{z_{2}(u)}$ (resp.~by checking whether $%
r_{z_{1}(u)}<l_{v}$). Therefore, the total running time of the algorithm is $%
O(n+m)$.
\end{proof}

\subsection{Normal paths\label{normal-paths-subsec}}

All our results on interval graphs rely on the notion of a \emph{normal}
path~\cite{longest-int-algo} (also referred to as a \emph{straight} path in~%
\cite{Damasc93,Keil85}). This notion has also been extended to the greater
class of cocomparability graphs~\cite{MertziosCorneil12}. Normal paths are
useful in the analysis of our data reductions in this section, as well as in
our algorithm in Section~\ref{algorithm-sec}, as they impose certain \emph{%
monotonicity} properties of the paths. Informally, the vertices in a normal
path appear in a ``left-to-right fashion'' in the right-endpoint ordering $%
\sigma $. In the following, given a graph $G$ and a path $%
P=(v_{1},v_{2},\dots ,v_{l})$ of~$G$, we write $v_{i}<_{P}v_{j}$ if and only
if $i<j$, i.e.,~whenever $v_{i}$ precedes $v_{j}$ in~$P$.

\begin{definition}
\label{normal-path-def}Let $G=(V,E)$ be an interval graph and $\sigma $ be a
right-endpoint ordering of~$V$. The path ${P=(v_{1},v_{2},\ldots ,v_{k})}$
of~$G$ is \emph{normal} if:

\begin{itemize}
\item $v_{1}$ is the leftmost vertex among $\{v_{1},v_{2},\ldots ,v_{k}\}$
in~$\sigma $ and 

\item $v_{i}$ is the leftmost vertex of~$N(v_{i-1})\cap
\{v_{i},v_{i+1},\ldots ,v_{k}\}$ in~$\sigma $, for every $i=2,\ldots ,k$.
\end{itemize}
\end{definition}

\begin{lemma}[\hspace{-0.001cm}\protect\cite{longest-int-algo}]
\label{lem:normal.path.existence} Let $G$ be an interval graph and $\mathcal{%
I}$ be an interval representation of~$G$. For every path $P$ of~$G$, there
exists a normal path $P^{\prime }$ of~$G$ such that $V(P)=V(P^{\prime })$.
\end{lemma}

We now provide a few properties of normal paths on interval graphs that we
will need later on.

\begin{observation}
\label{obs:appending.a.vertex} Let $G$ be an interval graph and $P$ be a
normal path of~$G$. Let $y$ be the last vertex of~$P$ and $z\in
V(G)\setminus V(P)$ such that $yz\in E(G)$ and $v<_{\sigma }z$ for every
vertex $v\in V(P)$. Then $(P,z)$ is a normal path of~$G$.
\end{observation}

\begin{observation}
\label{obs:first.vertex} Let $G$ be an interval graph, $P$ be a normal path
of~$G$, and $u,w\in V(P)$. If $u<_{P} w$ and $w<_{\sigma} u$, then $u$ is
not the first vertex of~$P$.
\end{observation}

\begin{lemma}
\label{lem:oxi.anapoda} Let $G$ be an interval graph, $P$ be a normal path
of~$G$, and $u,w\in V(P)$. If $u<_{P}w$ and $w<_{\sigma }u$, then $wu\in
E(G) $.
\end{lemma}

\begin{proof}
The proof is done by contradiction. Let $u,w\in V(P)$, where $u<_{P}w$ and $%
w<_{\sigma }u$. Assume that $wu\notin E(G)$. Among all such pairs of
vertices, we can assume without loss of generality that ${\mathbf{dist}}%
_{P}(u,w)$ is maximum, where ${\mathbf{dist}}_{P}(u,w)$ denotes the distance
between the vertices $u$ and $w$ on the path $P$. From Observation~\ref%
{obs:first.vertex}, $u$ is not the first vertex of~$P$, and therefore $u$
has a predecessor, say $z$, in~$P$. Note that ${\mathbf{dist}}_{P}(z,w)={%
\mathbf{dist}}_{P}(u,w)+1$. Suppose that $wz\in E(G)$. Then, since $%
w<_{\sigma }u$ and $u<_{P}w$, it follows by the normality of~$P$ that $u$ is
not the next vertex of~$z$ in the path $P$, which is a contradiction.
Therefore $wz\notin E(G)$. Suppose now that $z<_{\sigma }w$. Then, since in
this case $z<_{\sigma }w<_{\sigma }u$ and $zu\in E(G)$, it follows by Lemma~%
\ref{lem:ordering.edges} that $wu\in E(G)$, which is a contradiction to our
assumption. Therefore $w<_{\sigma }z$. Recall that $z<_{P}u<_{P}w$, $%
zw\notin E(G)$, and ${\mathbf{dist}}_{P}(z,w)={\mathbf{dist}}_{P}(u,w)+1$.
This is a contradiction to our assumption that ${\mathbf{dist}}_{P}(u,w)$ is
maximum. Therefore $wu\in E(G)$. This completes the proof of the lemma.
\end{proof}

\begin{lemma}
\label{lem:turning.point}
Let $G$ be an interval graph and $P=(P_{1},u,w,P_{2})$ be a normal path of~$%
G $, $u,w\in V(G)$. If $w<_{\sigma} u$ then $I_{w}\subseteq I_{u}$.
\end{lemma}

\begin{proof}
Let us assume, to the contrary, that $I_{w}\nsubseteq I_{u}$, that is, $%
l_{w}<l_{u}$. From Observation~\ref{obs:first.vertex}, $u$ is not the first
vertex of~$P$ and therefore $u$ has a predecessor, say $z$, in~$P$. Since $P$
is normal, $w<_{\sigma }u$, and $u<_{P}w$, it follows that $w$ is not a
neighbor of~$z$. Notice then that $z<_{P}u<_{P}w$. Furthermore, as $z$ is a
neighbor of~$u$, $w<_{\sigma }z$. Summarizing, $z<_{P}w$, $w<_{\sigma }z$,
and $wz\notin G$. From Lemma~\ref{lem:oxi.anapoda}, this is a contradiction
to the assumption that $P$ is normal. Therefore, $I_{w}\subseteq I_{u}$.
\end{proof}

\begin{lemma}
\label{lem:propers.consecutive.in.path}
Let $G$ be an interval graph and $\mathcal{I}$ be an interval representation
of~$G$. Let $S\subseteq V(G)$ such that $\mathcal{I}[S]$ is a proper
interval representation of~$G[S]$. Let $P$ be a normal path of~$G$ and $%
u,v\in S\cap V(P)$. If $u<_{\sigma }v$ then $u<_{P}v$. 
\end{lemma}

\begin{proof}
Let $P$ be a normal path of~$G$ where $u,v\in V(P)$, $u<_{\sigma }v$. If $%
uv\notin E(G)$ then from Lemma~\ref{lem:oxi.anapoda}, we obtain that $u<_{P}v
$. Thus, from now on, we assume that $uv\in E(G)$. Towards a contradiction
we further assume that $v<_{P}u$, that is, $P=(P_{1},v,P_{2},u,P_{3})$. From
Observation~\ref{obs:first.vertex}, it follows that $v$ is not the first
vertex of~$P$ and thus $P_{1}\neq \emptyset $, that is, $P=(P_{1}^{\prime
},y,v,P_{2},u,P_{3})$, for some $y\in V(G)$. Notice also that if $%
P_{2}=\emptyset $, then $P=(P_{1}^{\prime },y,v,u,P_{3})$ and from Lemma~\ref%
{lem:turning.point}, $I_{u}\subseteq I_{v}$, a contradiction to the
assumption that $u,v\in S$ and $\mathcal{I}[S]$ is a proper interval
representation of~$G[S]$. Thus, $P_{2}\neq \emptyset $. Therefore, $%
P=(P_{1}^{\prime },y,v,z,P_{2}^{\prime },u,P_{3})$, for some $z\in V(G)$.
Since $P$ is normal and $u<_{\sigma }v$ then $y\notin N(u)$. Notice that if
we prove that $u<_{\sigma }y$, then we obtain a contradiction from Lemma~\ref%
{lem:oxi.anapoda} and the lemma follows. Thus, it is enough to prove that $%
u<_{\sigma }y$.

To prove that $u<_{\sigma} y$ we claim towards a contradiction that $%
y<_{\sigma} u$. Notice then, that as $y\notin N(u)$, it also holds that $%
r_{y}<l_{u}$. However, since $v$ and $u$ are proper intervals and $%
u<_{\sigma} v$, $l_{u}<l_{v}$. Thus, $r_{y}<l_{v}$, a contradiction to the
assumption that $yv\in E(G)$ as $y$ is the predecessor of~$v$ in~$P$.
Therefore, $u<_{\sigma} y$ and this completes the proof of the lemma.
\end{proof}

\begin{lemma}
\label{lem:tobenamed1} Let $G$ be an interval graph and $P=(P_{1},u,u^{%
\prime },P_{2})$ be a normal path of~$G$. For every vertex $v\in V(P_{2})$,
it holds that $u<_{\sigma} v$ or $u^{\prime }<_{\sigma} v$.
\end{lemma}

\begin{proof}
Towards a contradiction we assume that $v<_{\sigma }u$ and $v<_{\sigma
}u^{\prime }$. Then from Lemma~\ref{lem:oxi.anapoda}, we obtain that $uv\in
E(G)$. Thus, since $u^{\prime },v\in N(u)$ and $v<_{\sigma }u^{\prime }$, it
follows by the normality of~$P$ that $u^{\prime }$ is not the next vertex of 
$u$ in~$P$, which is a contradiction. Therefore $u<_{\sigma }v$ or $%
u^{\prime }<_{\sigma }v$.
\end{proof}

\begin{lemma}
\label{lem:not.turning.point} Let $G$ be an interval graph and $%
P=(P_{1},u,v,w,P_{2})$ be a normal path of~$G$. If $v<_{\sigma} u$ then $%
v<_{\sigma} w$.
\end{lemma}

\begin{proof}
Let us assume that $v<_{\sigma} u$ and $w<_{\sigma} v$. From Lemma~\ref%
{lem:turning.point} it follows that $I_{v}\subseteq I_{u}$ and that $%
I_{w}\subseteq I_{v}$. Therefore, $I_{w}\subseteq I_{u}$ and thus $w\in N(u) 
$. This is a contradiction to the assumption that $P$ is normal as $v,w\in
N(u)$, $w<_{\sigma} v$, $v<_{P} w$, and $v$ is the vertex that follows $u$
in~$P$.
\end{proof}

\begin{lemma}
\label{lem:endiamesa} Let $G$ be an interval graph and $%
P=(P_{1},u,w,P_{2},v,P_{3})$ be a normal path of~$G$ where $v\in N(u)$. If $%
u<_{\sigma} w$, then $u<_{\sigma} x<_{\sigma} v$, for every vertex $x\in
V(P_{2})\cup \{w\}$.
\end{lemma}

\begin{proof}
As $w,v\in N(u)$, $w<_{P}v$, $w$ follows $u$ in~$P$, and $P$ is normal we
obtain that $u<_{\sigma} w<_{\sigma} v$. Thus, it remains to prove that $%
x<_{\sigma} v$ for every $x\in V(P_{2})$. We prove first that $u<_{\sigma} x$
for every $x\in V(P_{2})$. Assume to the contrary that $x<_{\sigma}
u<_{\sigma} w$. If $x\in N(u)$ then we obtain a contradiction to the
assumption that $P$ is normal. Thus, $x<_{\sigma} u$, $u<_{P} x$, and $%
ux\notin E(G)$. This is again a contradiction to the assumption that $P$ is
normal (from Lemma~\ref{lem:oxi.anapoda}). Therefore, $u<_{\sigma} x$ for
every $x\in V(P_{2})$ and, since $u<_{\sigma} w$, $u<_{\sigma} x$ for every $%
x\in V(P_{2})\cup \{w\}$.

We assume towards a contradiction that there exists a vertex $x^{\prime }\in
V(P_{2})$ such that $v<_{\sigma }x^{\prime }$. Without loss of generality we
also assume that $x^{\prime }$ is the first such vertex of~$P_{2}$, that is, 
$x<_{\sigma }v$ for every vertex $x\in V(P_{2})$ with $x<_{P}x^{\prime }$.
We denote by $z$ the predecessor of~$x^{\prime }$ in~$P_{2}$ or $w$ if $%
x^{\prime }$ is the first vertex of~$P_{2}$. Then, as the path is normal, $%
x^{\prime }<_{P}v$, $v<_{\sigma }x^{\prime }$, and $zv\notin E(G)$. Then,
since $z<_{\sigma }v$, $r_{z}<l_{v}$. However, as $uv\in E(G)$ and $%
u<_{\sigma }v$, we obtain that $l_{v}<r_{u}$. Thus, $r_{z}<l_{v}<r_{u}$ and,
therefore, $z<_{\sigma }u$, a contradiction. Thus, $x<_{\sigma }v$ for every
vertex $x\in V(P_{2})\cup \{w\}$.
\end{proof}

\subsection{The first data reduction\label{first-data-reduction-subsec}}

Here we present our first data reduction on interval graphs (see~Reduction
Rule~\ref{data-reduction-1-redrule}). By applying this data reduction to a
given interval graph $G$, we obtain a weighted interval graph $G^{\#}$ with
weights on its vertices, such that the maximum weight of a path in~$G^{\#}$
equals the greatest number of vertices of a path in~$G$ (cf.~Theorem~\ref%
{first-data-reduction-thm}). We first introduce the notion of a \emph{%
reducible} set of vertices and some related properties, which are essential
for our Reduction Rule~\ref{data-reduction-1-redrule}.

\begin{definition}
\label{def:reducible.set} Let $G$ be a (weighted) interval graph and $%
\mathcal{I}$ be an interval representation of~$G$. A set $S\subseteq V(G)$
is \emph{reducible} if it satisfies the following conditions:

\begin{enumerate}
\item $\mathcal{I}[S]$ induces a connected proper interval representation of 
$G[S]$ and

\item for every $v\in V(G)$ such that $I_{v}\subseteq \mathbf{span}(S)$ it
holds $v\in S$.
\end{enumerate}
\end{definition}

The intuition behind reducible sets is as follows. For every reducible set~$%
S $, a longest path $P$ contains either all vertices of $S$ or none of them
(cf.~Lemma~\ref{full.component.or.none}). Furthermore, in a certain longest
path $P$ which contains the whole set $S$, the vertices of $S$ appear \emph{%
consecutively} in $P$ (cf.~Lemma~\ref{lem:consecutive.vertices}). Thus we
can reduce the number of vertices in a longest path $P$ (without changing
its total weight) by replacing all vertices of $S$ with a single vertex
having weight $|S|$, see Reduction Rule~\ref{data-reduction-1-redrule}.

The next two observations will be useful for various technical lemmas in the
remainder of the paper. Observation~\ref{extra-condition-reducible-set-obs}
follows by the two conditions of Definition~\ref{def:reducible.set} for the
reducible sets $S$ in a weighted interval graph $G$. Furthermore,
Observation \ref{obs:order.in.uirepresentation} can be easily verified by
considering any proper interval representation.

\begin{observation}
\label{extra-condition-reducible-set-obs} Let $G$ be a (weighted) interval
graph, $\mathcal{I}$ be an interval representation of~$G$, and $S\subseteq
V(G)$ be a reducible set. Then, for every $u\in S$ and every $v\in
V(G)\setminus \{u\}$, it holds $I_{v}\nsubseteq I_{u}$.
\end{observation}

\begin{observation}
\label{obs:order.in.uirepresentation}\label{obs:consecutive.joined.by.edge}
Let $G$ be a proper interval graph and $\mathcal{I}$ be a proper interval
representation of~$G$. For every $u,v\in V(G)$:

\begin{itemize}
\item If $u<_{\sigma }v$, then $l_{u}<l_{v}$.

\item If $u$ and $v$ are consecutive vertices in the ordering~$\sigma $ and $%
G$ is connected, then $uv\in E(G)$.
\end{itemize}
\end{observation}

\begin{lemma}
\label{full.component.or.none} Let $G$ be a weighted interval graph with
weight function $w:V(G)\rightarrow \mathbb{N}$ and let $S$ be a reducible
set in~$G$. Let also $P$ be a path of maximum weight in~$G$. Then either $%
S\subseteq V(P)$ or $S\cap V(P)=\emptyset $.
\end{lemma}

\begin{proof}
Let $P$ be a path of~$G$ of maximum weight and $S$ be a reducible set of~$G$%
. Without loss of generality we may assume by Lemma~\ref%
{lem:normal.path.existence} that $P$ is a normal path. Assume that $S\cap
V(P)\neq \emptyset $ and $S\nsubseteq V(P)$. Then there exist two
consecutive vertices $u_{1},u_{2}\in S$ in the vertex ordering $\sigma $
(where $u_{1}<_{\sigma }u_{2}$) such that either $u_{1}\in V(P)$ and $%
u_{2}\notin V(P)$, or $u_{1}\notin V(P)$ and $u_{2}\in V(P)$. In both cases
we will show that we can augment the path $P$ by adding vertex $u_{2}$ or $%
u_{1}$, respectively, which contradicts our maximality assumption on $P$.
Since, by Definition~\ref{def:reducible.set}, $\mathcal{I}[S]$ induces a
connected proper interval representation of~$G[S]$, it follows by
Observation~\ref{obs:consecutive.joined.by.edge} that $u_{1}u_{2}\in E(G)$.

First suppose that $u_{1}\in V(P)$ and $u_{2}\notin V(P)$. Let $%
P=(P_{1},u_{1},P_{2})$. Notice first that, if $P_{2}=\emptyset $, then the
path $P^{\prime }=(P_{1},u_{1},u_{2})=(P,u_{2})$ is a path of~$G$ with
greater weight than $P$, which is a contradiction to the maximality
assumption on $P$. Thus, $P_{2}\neq \emptyset $. Let $w\in V(P_{2})$ be the
first vertex of~$P_{2}$, i.e.,~$P=(P_{1},u_{1},w,P_{2}^{\prime })$. We show
that $u_{1}<_{\sigma }w$. Assume to the contrary that $w<_{\sigma }u_{1}$.
Then $I_{w}\subseteq I_{u_{1}}$ by Lemma~\ref{lem:turning.point}. This is a
contradiction, since $u_{1}\in S$ and $S$ is a reducible set (cf.~Definition~%
\ref{def:reducible.set}). Therefore $u_{1}<_{\sigma }w$. Then either $%
u_{1}<_{\sigma }u_{2}<_{\sigma }w$ or $u_{1}<_{\sigma }w<_{\sigma }u_{2}$.
Now we show that $u_{2}w\in E(G)$. If $u_{1}<_{\sigma }u_{2}<_{\sigma }w$,
then Lemma~\ref{lem:ordering.edges} implies that $u_{2}w\in E(G)$, since $%
u_{1}w\in E(G)$. If $u_{1}<_{\sigma }w<_{\sigma }u_{2}$ then again Lemma~\ref%
{lem:ordering.edges} implies that $u_{2}w\in E(G)$ since $u_{1}u_{2}\in E(G) 
$. Thus, since $u_{2}w\in E(G)$, it follows that there exists the path $%
P^{\prime }=(P_{1},u_{1},u_{2},w,P_{2^{\prime }})$ which has greater weight
than $P$, which is a contradiction to the maximality assumption on $P$.

Now suppose that $u_{1}\notin V(P)$ and $u_{2}\in V(P)$. Let then $%
P=(P_{1},u_{2},P_{2})$. Notice that, if $P_{1}=\emptyset $, then the path $%
P^{\prime }=(u_{1},u_{2},P_{2})=(u_{1},P)$ is a path of~$G$ with greater
weight than $P$, which is a contradiction. Thus $P_{1}\neq \emptyset $. Let $%
z\in V(P_{1})$ be the last vertex of~$P_{1}$, i.e.,~$P=(P_{1}^{\prime
},z,u_{2},P_{2})$. We show that $u_{1}z\in E(G)$. First let $u_{2}<_{\sigma
}z$. Then $I_{u_{2}}\subseteq I_{z}$ by Lemma~\ref{lem:turning.point}, and
thus $N(u_{2})\subseteq N(z)$. Therefore, since $u_{1}u_{2}\in E(G)$, it
follows that $u_{1}z\in E(G)$ in the case where $u_{2}<_{\sigma }z$. Let now 
$z<_{\sigma }u_{2}$. Suppose that $u_{1}z\notin E(G)$. Note that $%
l_{u_{2}}<r_{u_{1}}$, since $u_{1}<_{\sigma }u_{2}$ and $u_{1}u_{2}\in E(G)$%
. Furthermore, since $S$ is a reducible set, $\mathcal{I}[S]$ induces a
proper interval representation of~$G[S]$ by Definition~\ref%
{def:reducible.set}. Then, since $u_{1}<_{\sigma }u_{2}$ and $u_{1},u_{2}\in
S$ are consecutive in~$\sigma $, Observation~\ref%
{obs:order.in.uirepresentation} implies that $l_{u_{1}}<l_{u_{2}}$. That is, 
$l_{u_{1}}<l_{u_{2}}<r_{u_{1}}$. Hence, since $u_{2}z\in E(G)$ and $%
u_{1}z\notin E(G)$, it follows that $l_{u_{1}}<l_{u_{2}}<r_{u_{1}}<l_{z}$.
Finally $r_{z}<r_{u_{2}}$, since $z<_{\sigma }u_{2}$ by assumption, and thus 
$I_{z}\subseteq I_{u_{2}}$. This is a contradiction since $u_{1}\in S$ and $%
S $ is a reducible set (cf.~Definition~\ref{def:reducible.set}). Therefore, $%
zu_{1}\in E(G)$ in the case where $z<_{\sigma }u_{2}$. That is, we always
have $zu_{1}\in E(G)$. Therefore there exists the path $P^{\prime
}=(P_{1},u_{1},u_{2},w,P_{2^{\prime }})$ which has greater weight than $P$,
which is a contradiction to the maximality assumption on $P$.
\end{proof}

\begin{lemma}
\label{lem:consecutive.vertices} Let $G$ be a weighted interval graph with
weight function $w:V(G)\rightarrow \mathbb{N}$, and $S$ be a reducible set
in~$G$. Let also $P$ be a path of maximum weight in~$G$ and let $S\subseteq
V(P)$. Then there exists a path $P^{\prime }$ of~$G$ such that $V(P^{\prime
})=V(P)$ and the vertices of~$S$ appear consecutively in~$P^{\prime }$.
\end{lemma}

\begin{proof}
Let $P$ be a path of maximum weight of~$G$ such that $S\subseteq V(P)$.
Denote $S=\{u_{1},u_{2},\dots ,u_{|S|}\}$, where $u_{1}<_{\sigma
}u_{2}<_{\sigma }\dots <_{\sigma }u_{|S|}$. Without loss of generality we
may assume by Lemma~\ref{lem:normal.path.existence} that $P$ is normal.
Since $\mathcal{I}[S]$ induces a proper interval representation of~$G[S]$
(cf. Definition~\ref{def:reducible.set}) Lemma~\ref%
{lem:propers.consecutive.in.path} implies that also $u_{1}<_{P}u_{2}<_{P}%
\dots <_{P}u_{|S|}$, i.e.,~the vertices of~$S$ appear in the same order both
in the vertex ordering $\sigma $ and in the path $P$. Furthermore,
Observation~\ref{obs:consecutive.joined.by.edge} implies that $%
u_{i}u_{i+1}\in E(G)$ for every $i\in \lbrack |S|-1]$. Let $%
P=(P_{0},u_{1},P_{1},u_{2},\dots ,P_{|S|-1},u_{|S|},P_{|S|})$, where 
\begin{equation}
V(P_{i})\cap S=\emptyset ,\text{ for every }i\in \lbrack |S|-1].
\label{eq:eq.empty.intersection}
\end{equation}%
For every $i\in \lbrack |S|-1]$, we denote 
\begin{equation*}
z_{i}=%
\begin{cases}
\text{the first vertex of }P_{i} & \text{if }P_{i}\neq \emptyset  \\ 
u_{i+1} & \text{otherwise}%
\end{cases}%
.
\end{equation*}

Let $i\in \lbrack |S|-1]$ and suppose that $z_{i}<_{\sigma }u_{i}$. If $%
z_{i}=u_{i+1}$ then $u_{i}<_{\sigma }z_{i}=u_{i+1}$, which is a
contradiction. Therefore $z_{i}\neq u_{i+1}$. That is, $z_{i}$ is the first
vertex of~$P_{i}$, i.e.,~the successor of~$u_{i}$ in~$P$. Then, since we
assumed that $z_{i}<_{\sigma }u_{i}$, it follows by Lemma~\ref%
{lem:turning.point} that $I_{z_{i}}\subseteq I_{u_{i}}$. Thus $z_{i}\in S$,
since $u_{i}\in S$ and $S$ is a reducible set (cf.~Definition~\ref%
{def:reducible.set}). This is a contradiction to Eq.~(\ref%
{eq:eq.empty.intersection}). Therefore $u_{i}<_{\sigma }z_{i}$, for every $%
i\in \lbrack |S|-1]$.

Recall that $\mathcal{I}[S]$ induces a proper interval representation of~$%
G[S]$ and that $u_{1}<_{\sigma }u_{2}<_{\sigma }\dots <_{\sigma }u_{|S|}$.
Thus $l_{u_{1}}=\min \{l_{u}:~u\in S\}$ and $r_{u_{|S|}}=\max \{r_{u}~u\in
S\}$, i.e.,~$\mathbf{span}(S)=[l_{u_{1}},r_{u_{|S|}}]$. Now let $i\in
\lbrack |S|-1]$. Since $u_{i}u_{i+1}\in E(G)$ and $u_{i}<_{\sigma }z_{i}$,
Lemma~\ref{lem:endiamesa} implies that $u_{i}<_{\sigma }z<_{\sigma }u_{i+1}$
for every $z\in V(P_{i})$. Therefore $u_{1}<_{\sigma }z<_{\sigma }u_{|S|}$,
for every $z\in V(P_{i})$, where $i\in \lbrack |S|-1]$. Suppose that there
exists a vertex $z\in V(P_{i})$ such that $l_{u_{1}}<l_{z}$. Then, since $%
z<_{\sigma }u_{|S|}$, it follows that $I_{z}\subseteq \mathbf{span}(S)$.
Thus $z\in S$, since $S$ is a reducible set by assumption (cf.~Definition~%
\ref{def:reducible.set}). This is a contradiction to Eq.~(\ref%
{eq:eq.empty.intersection}). Thus $l_{z}<l_{u_{1}}$ for every $z\in V(P_{i})$%
, where $i\in \lbrack |S|-1]$. Since also $u_{1}<_{\sigma }z$ as we proved
above, it follows that $I_{u_{1}}\subseteq I_{z}$ for every $z\in V(P_{i})$,
where $i\in \lbrack |S|-1]$. Thus $N(u_{1})\subseteq N(z)$ for every $z\in
V(P_{i})$, where $i\in \lbrack |S|-1]$. Therefore $P^{\prime
}=(P_{0},P_{1},\dots ,P_{|S|-1},u_{1},u_{2},\dots ,u_{|S|},P_{|S|})$ is a
path of~$G$, where $V(P^{\prime })=V(P)$ and the vertices of~$S$ appear
consecutively in~$P^{\prime }$. This completes the proof of the lemma.
\end{proof}

\medskip

We now present two auxiliary technical lemmas that will be used to prove the
correctness of Reduction Rule~\ref{data-reduction-1-redrule} in Theorem~\ref%
{first-data-reduction-thm}.

\begin{lemma}
\label{lem-invariant-reducible}Let $G$ be an interval graph and $\mathcal{I}$
be an interval representation of~$G$. Let also $S,S^{\prime }\subseteq V(G)$
be two reducible sets of~$G$ such that $S\cap S^{\prime }=\emptyset $. Let $%
G^{\prime }$ be the graph obtained from $G$ by replacing $\mathcal{I}[S]$ by 
$\mathbf{span}(S)$. Then $S^{\prime }$ remains a reducible set of~$G^{\prime
}$.
\end{lemma}

\begin{proof}
Let $u$ be the vertex of~$V(G^{\prime })\setminus V(G)$, i.e.,~$I_{u}=%
\mathbf{span}(S)$. Denote by $\mathcal{I}^{\prime }$ the interval
representation obtained from $\mathcal{I}$ after replacing $\mathcal{I}[S]$
by $\mathbf{span}(S)$. First note that $\mathcal{I}^{\prime }[S^{\prime }]=%
\mathcal{I}[S^{\prime }]$, and thus $\mathcal{I}^{\prime }[S^{\prime }]$
induces a connected proper interval representation as $S^{\prime }$ is a
reducible set by assumption. This proves Condition~1 of Definition~\ref%
{def:reducible.set}.

Let now $v\in V(G^{\prime })$ such that $I_{v}\subseteq \mathbf{span}%
(S^{\prime })$. Assume that $v\notin S^{\prime }$. If $v\neq u$, then $v\in
V(G)$ and thus $v\in S^{\prime }$, since $S^{\prime }$ is a reducible set of 
$G$. If $v=u$, then $\mathbf{span}(S)=I_{u}=I_{v}\subseteq \mathbf{span}%
(S^{\prime })$. That is, for every $u_{0}\in S$ we have $I_{u_{0}}\subseteq 
\mathbf{span}(S)\subseteq \mathbf{span}(S^{\prime })$, and thus also $%
u_{0}\in S^{\prime }$, since $S^{\prime }$ is a reducible set. This is a
contradiction, since $S\cap S^{\prime }=\emptyset $. Therefore, for every $%
v\in V(G^{\prime })$ such that $I_{v}\subseteq \mathbf{span}(S^{\prime })$,
we have that $v\in S^{\prime }$. This proves Condition~2 of Definition~\ref%
{def:reducible.set} and completes the proof of the lemma.
\end{proof}

\begin{lemma}
\label{lem:equal-weight-first-replacement}Let $\ell $ be a positive integer.
Let $G$ be a weighted interval graph, $w:V(G)\rightarrow \mathbb{N}$, $%
\mathcal{I}$ be an interval representation of~$G$, and $S$ be a reducible
set in~$G$. Let also $G^{\prime }$ be the graph obtained from $G$ by
replacing $\mathcal{I}[S]$ by an interval $I_{u}=\mathbf{span}(S)$ where $%
w(u)=\sum_{v\in S}w(v)$. Then the maximum weight of a path in~$G$ is $\ell $
if and only if the maximum weight of a path in~$G^{\prime }$ is $\ell $.
\end{lemma}

\begin{proof}
First assume that the maximum weight of a path $P$ in~$G$ is $\ell $.
Without loss of generality, from Lemma~\ref{lem:normal.path.existence}, we
may also assume that $P$ is normal. Furthermore, either $S\subseteq V(P)$ or 
$S\cap V(P)=\emptyset $ by Lemma~\ref{full.component.or.none}. Notice that
if $S\cap V(P)=\emptyset $, then $P$ is also a path of~$G^{\prime }$.
Suppose that $S\subseteq V(P)$. Then from Lemma~\ref%
{lem:consecutive.vertices}, we can obtain a path $\widehat{P}$ such that $V(%
\widehat{P})=V(P)$ and the vertices of~$S$ appear consecutively in~$\widehat{%
P}$. Notice then that by replacing the subpath of~$\widehat{P}$ consisting
of the vertices of~$S$ by the single vertex $u$, we obtain a path $P^{\prime
}$ of~$G^{\prime }$ such that $V(P^{\prime })\setminus V(P)=\{u\}$ and $%
V(P^{\prime })\cap V(P)=V(P)\setminus \{S\}$. Since $w(u)=\sum_{v\in S}w(v)$%
, we obtain that $\sum_{v\in V(P^{\prime })}w(v)=\sum_{v\in V(P)}w(v)$.
Thus, $G^{\prime }$ has a path $P^{\prime }$ of weight at least $\ell $.

Now assume that the maximum weight of a path $P^{\prime }$ in~$G^{\prime }$
is $\ell $. If $u\notin V(P^{\prime })$, then $V(P^{\prime })$ is also a
path of~$G$. Suppose that $u\in V(P^{\prime })$. Let $P^{\prime
}=(P_{1},v,u,v^{\prime },P_{2})$. Our aim is to show that $u_{|S|}v^{\prime
}\in E(G)$ and $vu_{1}\in E(G)$. Suppose that $v^{\prime }<_{\sigma }u$.
Then Lemma~\ref{lem:turning.point} implies that $I_{v^{\prime }}\subseteq
I_{u}=\mathbf{span}(S)$, and thus $v^{\prime }\in S=V(G)\setminus
V(G^{\prime })$. This is a contradiction, since $v^{\prime }\in V(G^{\prime
})$. Thus $u<_{\sigma }v^{\prime }$, i.e.,~$r_{u}=r_{u_{|S|}}<r_{v^{\prime
}} $. Since $uv^{\prime }\in E(G)$, it follows that $l_{v^{\prime
}}<r_{u}=r_{u_{|S|}}$. Therefore $r_{u_{|S|}}\in I_{v^{\prime }}$, and thus $%
v^{\prime }u_{|S|}\in E(G)$. It remains to show that $vu_{1}\in E(G)$. First
let $u<_{\sigma }v$. Then $I_{u}\subseteq I_{v}$ by Lemma~\ref%
{lem:turning.point}. Furthermore, since $I_{u_{1}}\subseteq I_{u}=\mathbf{%
span}(S)$, it follows that $I_{u_{1}}\subseteq I_{v}$, and thus $vu_{1}\in
E(G)$. Let now $v<_{\sigma }u$, i.e.,~$r_{v}<r_{u}=r_{u_{|S|}}$. Then, since 
$vu\in E(G)$, it follows that $l_{u_{1}}=l_{u}<r_{v}<r_{u}$. If $%
I_{v}\subseteq I_{u}=\mathbf{span}(S)$, then $v^{\prime }\in S=V(G)\setminus
V(G^{\prime })$, which is a contradiction since $v\in V(G^{\prime })$. Thus $%
I_{v}\nsubseteq I_{u}$. Therefore, since $l_{u}<r_{v}<r_{u}$, it follows
that $l_{v}<l_{u}=l_{u_{1}}<r_{v}$, i.e.,~$l_{u_{1}}\in I_{v}$. Thus, $%
u_{1}v\in E(G)$. This implies that we may obtain a path $P$ of~$G$ by
replacing the vertex $u$ in~$P^{\prime }$ by the path $(u_{1},u_{2},\dots
,u_{|S|})$. As before, since $w(u)=\sum_{v\in S}w(v)$, the weight of~$P$ is
equal to the weight of~$P^{\prime }$. Hence, $G$ has a path of weight at
least $\ell $.
\end{proof}

\medskip

In the next definition we reduce the interval graph $G$ to the \emph{weighted%
} interval graph $G^{\#}$ which has fewer vertices than $G$. Then, as we
prove in Theorem~\ref{first-data-reduction-thm}, the longest paths of~$G$
correspond to the maximum-weight paths of $G^{\#}$.

\begin{redrule}[first data reduction]
\label{data-reduction-1-redrule}Let $G=(V,E)$ be an interval graph, $%
\mathcal{I}$ be an interval representation of~$G$, and $D$ be a proper
interval deletion set of~$G$. Let $\mathcal{S}$ be a set of vertex disjoint
reducible sets of~$G$, where $S\cap D=\emptyset $, for every $S\in \mathcal{S%
}$. The \emph{weighted interval graph} $G^{\#}=(V^{\#},E^{\#})$ is induced
by the \emph{weighted interval representation} $\mathcal{I}^{\#}$, which is
derived from $\mathcal{I}$ as follows:\vspace{-0,1cm}

\begin{itemize}
\item for every $S\in \mathcal{S}$, replace in~$\mathcal{I}$ the intervals $%
\{I_{v}:v\in S\}$ with the single interval $I_{S}=\mathbf{span}(S)$ which
has weight $|S|$; all other intervals receive weight $1$.
\end{itemize}
\end{redrule}

In the next theorem we prove the correctness of Reduction Rule~\ref%
{data-reduction-1-redrule}.

\begin{theorem}
\label{first-data-reduction-thm}Let $\ell $ be a positive integer. Let $G$
be an interval graph and $G^{\#}$ be the weighted interval graph derived by
Reduction Rule~\ref{data-reduction-1-redrule}. Then the longest path in $G$
has $\ell $ vertices if and only if the maximum weight of a path in~$G^{\#}$
is $\ell $.
\end{theorem}

\begin{proof}
The construction of the graph $G^{\#}$ from $G$ in Reduction Rule~\ref%
{data-reduction-1-redrule} can be done sequentially, replacing each time the
intervals of a set $S\in \mathcal{S}$ with one interval of the appropriate
weight. After making this replacement for such a set $S\in \mathcal{S}$, the
maximum weight of a path in the resulting graph is by Lemma~\ref%
{lem:equal-weight-first-replacement} equal to the maximum weight of a path
in the graph before the replacement of~$S$. Furthermore, the vertex set of
every other set $S^{\prime }\in \mathcal{S}\setminus S$ remains reducible in
the resulting graph by Lemma~\ref{lem-invariant-reducible}. Therefore we can
iteratively replace all sets of~$\mathcal{S}$, resulting eventually in the
weighted graph~$G^{\#}$, in which the maximum weight of a path is equal to
the maximum number of vertices in a path in the original graph $G$.
\end{proof}

\medskip

In the next lemma we prove that the weighted interval graph $G^{\#}$, which
is obtained from an interval graph $G$ by applying Reduction Rule~\ref%
{data-reduction-1-redrule} to it, has some useful properties. These
properties will be exploited in Section~\ref{special-interval-sec}
(cf.~Corollary~\ref{first-reduction-G-diesi-properties-cor} in Section~\ref%
{special-interval-sec}) as they are crucial for deriving a \emph{special
weighted interval graph} (cf.~Definition~\ref{special-interval-def} in
Section~\ref{special-interval-sec}).

\begin{lemma}
\label{lem:first:data:reduction} Let $G$ be an interval graph, $D$ be a
proper interval deletion set of~$G$, and $\mathcal{S}$ be a set of
vertex-disjoint reducible sets of~$G$, where $S\cap D=\emptyset $, for every 
$S\in \mathcal{S}$. Suppose that $\mathbf{span}(S)\cap \mathbf{span}%
(S^{\prime })=\emptyset $ for every two distinct sets $S,S^{\prime }\in 
\mathcal{S}$. Furthermore, let $G^{\#}$ be the weighted interval graph
obtained from Reduction Rule~\ref{data-reduction-1-redrule}. Then $D$ is a
proper interval deletion set of~$G^{\#}$. Furthermore $V(G^{\#})\setminus D$
can be partitioned into two sets $A$ and $U^{\#}$, where $%
A=V(G^{\#})\setminus V(G)$ and:

\begin{enumerate}
\item $A$ is an independent set of~$G^{\#}$, and

\item for every $v\in A$ and every $u\in V(G^{\#})\setminus \{v\}$, we have $%
I_{u}\nsubseteq I_{v}$.
\end{enumerate}
\end{lemma}

\begin{proof}
Define $A=V(G^{\#})\setminus V(G)$ and $U^{\#}=V(G^{\#})\setminus (D\cup A)$%
. Note that the sets $A$ and $U^{\#}$ form a partition of the set $%
V(G^{\#})\setminus D$. First we prove that $A$ is an independent set. Note
by Reduction Rule~\ref{data-reduction-1-redrule} that every vertex $a\in A$
corresponds to a different set $S_{a}\in \mathcal{S}$, which corresponds to
the interval $I_{a}=\mathbf{span}(S)$ in the interval representation $%
\mathcal{I}^{\#}$ of the graph $G^{\#}$. Therefore, since $\mathbf{span}%
(S)\cap \mathbf{span}(S^{\prime })=\emptyset $ for every two distinct sets $%
S,S^{\prime }\in \mathcal{S}$, it follows that $I_{a}\cap I_{a^{\prime
}}=\emptyset $ for every two distinct vertices $a,a^{\prime }\in A$. That
is, $A$ induces an independent set in~$G^{\#}$.

Second we prove that for every $v\in A$ and every $u\in V(G^{\#})\setminus
\{v\}$, we have $I_{u}\nsubseteq I_{v}$. Suppose otherwise that there exist $%
v\in A$ and $u\in V(G^{\#})\setminus \{v\}$ such that $I_{u}\subseteq I_{v}$%
. Then, since $A$ is an independent set, it follows that $u\notin A$, i.e.,~$%
u\in V(G)$. Recall by the construction that $I_{v}=\mathbf{span}(S)$, for
some reducible set $S$ in~$G$. Thus, since $I_{u}\subseteq I_{v}=\mathbf{span%
}(S)$, it follows by Definition~\ref{def:reducible.set} that $u\in S$. This
is a contradiction since by construction $V(G^{\#})\cap S=\emptyset $.
Therefore, for every $v\in A$ and every $u\in V(G^{\#})\setminus \{v\}$, we
have $I_{u}\nsubseteq I_{v}$.

Now we prove that $D$ is a proper interval deletion set of~$G^{\#}$. Suppose
otherwise that $G^{\#}\setminus D$ is not a proper interval graph. Then $%
G^{\#}\setminus D$ contains an induced $K_{1,3}$~\cite{Roberts69}. If this $%
K_{1,3}$ contains no vertex of~$A$, then this $K_{1,3}$ is also contained as
an induced subgraph of~$G\setminus D$. This is a contradiction, since $D$ is
a proper interval deletion set of~$G$ by assumption. Thus this $K_{1,3}$
contains at least one vertex of~$A$. Let $v$ denote the vertex of degree $3$
in this $K_{1,3}$. Suppose that $v\in A$. Then, for at least one vertex $u$
of the other three vertices of this $K_{1,3}$ we have that $I_{u}\subseteq
I_{v}$, which is a contradiction as we proved above. Suppose that $v\notin A$%
, i.e.,~$v\in U^{\#}$. Denote the leaves of the $K_{1,3}$ by $%
v_{1},v_{2},v_{3}$. Let $i\in \lbrack 3]$. If $v_{i}\in U^{\#}$, then $%
v_{i}\in V(G)$. Otherwise, if $v_{i}\in A$, then there exists a reducible
set $S_{i}$ in~$G$ such that $I_{v_{i}}=\mathbf{span}(S_{i})$. Furthermore,
since $v\in U^{\#}$ and $I_{v}\cap \mathbf{span}(S_{i})\neq \emptyset $,
there exists at least one vertex $u_{i}\in S_{i}$ where $u_{i}v\in E(G)$.
For every $i\in \lbrack 3]$, define%
\begin{equation*}
z_{i}=%
\begin{cases}
u_{i} & \text{if }v_{i}\in A, \\ 
v_{i} & \text{otherwise.}%
\end{cases}%
.
\end{equation*}%
Observe that $v,z_{1},z_{2},z_{3}$ belong to the original graph $G$, i.e.,~$%
v,z_{1},z_{2},z_{3}\in V(G)\setminus D$. Since $I_{v_{i}}\cap
I_{v_{j}}=\emptyset $ for every $1\leq i<j\leq 3$ and $I_{z_{i}}\subseteq
I_{v_{i}}$, $i\in \lbrack 3]$ it follows that $I_{z_{i}}\cap
I_{z_{j}}=\emptyset $, for every $1\leq i<j\leq 3$. Therefore, $%
\{z_{1},z_{2},z_{3}\}$ is an independent set in~$G\setminus D$. Moreover $%
vz_{i}\in E(G)$, for every $i\in \lbrack 3]$. Thus, $\{v,z_{1},z_{2},z_{3}\}$
induces a $K_{1,3}$ in~$G\setminus D$. This is a contradiction, since $D$ is
a proper interval deletion set of~$G$ by assumption. Therefore $D$ is a
proper interval deletion set of~$G^{\#}$.
\end{proof}

\subsection{The second data reduction\label{second-data-reduction-subsec}}

Here we present our second data reduction, which is applied to an arbitrary
weighted interval graph $G$ with weights on its vertices (cf.~Reduction Rule~%
\ref{data-reduction-2-redrule}). As we prove in Theorem~\ref%
{second-data-reduction-thm}, the maximum weight of a path in the resulting
weighted interval graph $\widehat{G}$ is the same as the maximum weight of a
path in~$G$. We first introduce the notion of a \emph{weakly reducible} set
of vertices and some related properties, which are needed for our Reduction
Rule~\ref{data-reduction-2-redrule}.

\begin{definition}
\label{def:weakly.reducible.set} Let $G$ be a (weighted) interval graph and $%
\mathcal{I}$ be an interval representation of~$G$. A set $S\subseteq V(G)$
is \emph{weakly reducible} if it satisfies the following conditions:

\begin{enumerate}
\item $\mathcal{I}[S]$ induces a connected proper interval representation of 
$G[S]$ and

\item for every $v\in V(G)$ and every $u\in S$, if $I_{v}\subseteq I_{u}$
then $S\subseteq N(v)$.
\end{enumerate}
\end{definition}

Note here that Condition~2 of Definition~\ref{def:weakly.reducible.set} also
applies to the case where $v=u$. Therefore $S\subseteq N(u)$ for every $u\in
S$, i.e.,~$S$ induces a clique, as the next observation states.

\begin{observation}
\label{obs:weakly.reducible.clique}\label%
{obs:common.intersection.weakly.reducible} Let $G$ be a (weighted) interval
graph, $\mathcal{I}$ be an interval representation of~$G$ and $S$ be a
weakly reducible set in~$G$. Then $G[S]$ is a clique. Furthermore ${[\max
\{l_{v}:v\in S\},\min\{r_{v}:v\in S\}]\subseteq I_{u}}$, for every~${u\in S}$%
.
\end{observation}

The intuition behind weakly reducible sets is as follows. For every weakly
reducible set~$S$, a longest path~$P$ contains either all vertices of $S$ or
none of them (cf.~Lemma~\ref{lem:full.weakly.reducible.or.none}).
Furthermore let $D$ be a given proper interval deletion set of $G$. Then, in
a certain path $P$ of maximum weight which contains the whole set $S$, the
appearance of the vertices of $S$ in $P$ is \emph{interrupted at most} $%
|D|+3 $ times by vertices outside $S$. That is, such a path~$P$ has at most $%
\min \{|S|,|D|+4\}$ vertex-maximal subpaths with vertices from $S$
(cf.~Lemma~\ref{lem:weakly.consecutive.vertices}). Thus we can reduce the
number of vertices in a maximum-weight normal path~$P$ (without changing its
total weight) by replacing all vertices of~$S$ with~${\min \{|S|,|D|+4\}}$
vertices, cf.~the Reduction Rule~\ref{data-reduction-2-redrule}; each of
these new vertices has the same weight and their total weight sums up to~$%
|S| $.

\begin{lemma}
\label{lem:full.weakly.reducible.or.none} Let $G$ be a weighted interval
graph with weight function $w:V(G)\rightarrow \mathbb{N}$ and let $S$ be a
weakly reducible set in~$G$. Let also $P$ be a path of maximum weight in~$G$%
. Then either $S\subseteq V(P)$ or $S\cap V(P)=\emptyset $.
\end{lemma}

\begin{proof}
Let $P$ be a path of~$G$ of maximum weight and $S$ be a weakly reducible set
in~$G$. Without loss of generality we also assume that $P$ is a normal path
(Lemma~\ref{lem:normal.path.existence}). Suppose towards a contradiction
that $S\cap V(P)\neq \emptyset $ and $S\nsubseteq V(P)$. Then there exist
two consecutive vertices $u_{1},u_{2}\in S$ in the vertex ordering $\sigma $
(where $u_{1}<_{\sigma }u_{2}$) such that either $u_{1}\in V(P)$ and $%
u_{2}\notin V(P)$, or $u_{1}\notin V(P)$ and $u_{2}\in V(P)$. In both cases
we will show that we can augment the path $P$ by adding vertex $u_{2}$ or $%
u_{1}$, respectively, which contradicts our maximality assumption on $P$.
From Observation~\ref{obs:weakly.reducible.clique}, $S$ induces a clique in~$%
G$. Hence, $u_{1}u_{2}\in E(G)$.

First suppose that $u_{1}\in V(P)$ and $u_{2}\notin V(P)$. Let $%
P=(P_{1},u_{1},P_{2})$. Notice first that, if $P_{2}=\emptyset $, then the
path $P^{\prime }=(P_{1},u_{1},u_{2})=(P,u_{2})$ is a path of~$G$ with
greater weight than $P$, which is a contradiction to the maximality
assumption on $P$. Thus, $P_{2}\neq \emptyset $. Let $w\in V(P_{2})$ be the
first vertex of~$P_{2}$, i.e.,~$P=(P_{1},u_{1},w,P_{2}^{\prime })$. We prove
that $u_{2}w\in E(G)$. For this, notice first that either $w<_{\sigma }u_{1}$
or $u_{1}<_{\sigma }w$. If $w<_{\sigma }u_{1}$, then $I_{w}\subseteq
I_{u_{1}} $ by Lemma~\ref{lem:turning.point}. Since $u_{1},u_{2}$ are
vertices of the weakly reducible set $S$ and $I_{w}\subseteq I_{u_{1}}$ then 
$u_{2}w\in E(G)$ (Definition~\ref{def:weakly.reducible.set}). If $%
u_{1}<_{\sigma }w$, then either $u_{1}<_{\sigma }u_{2}<_{\sigma }w$ or $%
u_{1}<_{\sigma }w<_{\sigma }u_{2}$. If $u_{1}<_{\sigma }u_{2}<_{\sigma }w$,
then Lemma~\ref{lem:ordering.edges} implies that $u_{2}w\in E(G)$ since $%
u_{1}w\in E(G)$. If $u_{1}<_{\sigma }w<_{\sigma }u_{2}$, then again Lemma~%
\ref{lem:ordering.edges} implies that $u_{2}w\in E(G)$, since $u_{1}u_{2}\in
E(G) $. This completes the argument that $u_{2}w\in E(G)$. Since $u_{2}w\in
E(G)$, it follows that there exists the path $P^{\prime
}=(P_{1},u_{1},u_{2},w,P_{2}^{\prime })$ which has greater weight than $P$,
which is a contradiction to the maximality assumption on $P$.

Now suppose that $u_{1}\notin V(P)$ and $u_{2}\in V(P)$. Let then $%
P=(P_{1},u_{2},P_{2})$. Notice that, if $P_{1}=\emptyset $, then the path $%
P^{\prime }=(u_{1},u_{2},P_{2})=(u_{1},P)$ is a path of~$G$ with greater
weight than $P$, which is a contradiction. Thus $P_{1}\neq \emptyset $. Let $%
z\in V(P_{1})$ be the last vertex of~$P_{1}$, i.e.,~$P=(P_{1}^{\prime
},z,u_{2},P_{2})$. We show that $u_{1}z\in E(G)$. First let $u_{2}<_{\sigma
}z$. Then $I_{u_{2}}\subseteq I_{z}$ by Lemma~\ref{lem:turning.point}, and
thus $N(u_{2})\subseteq N(z)$. Therefore, since $u_{1}u_{2}\in E(G)$, it
follows that $u_{1}z\in E(G)$ in the case where $u_{2}<_{\sigma }z$. Now let 
$z<_{\sigma }u_{2}$. Suppose that $u_{1}z\notin E(G)$. Note that $%
l_{u_{2}}<r_{u_{1}}$ since $u_{1}<_{\sigma }u_{2}$ and $u_{1}u_{2}\in E(G)$.
Furthermore, since $S$ is a weakly reducible set, $\mathcal{I}[S]$ induces a
proper interval representation of~$G[S]$ by Definition~\ref%
{def:weakly.reducible.set}. Then, since $u_{1}<_{\sigma }u_{2}$ and $%
u_{1},u_{2}\in S$ are consecutive in~$\sigma $, Observation~\ref%
{obs:order.in.uirepresentation} implies that $l_{u_{1}}<l_{u_{2}}$. That is, 
$l_{u_{1}}<l_{u_{2}}<r_{u_{1}}$. Hence, since $u_{2}z\in E(G)$ and $%
u_{1}z\notin E(G)$, it follows that $l_{u_{1}}<l_{u_{2}}<r_{u_{1}}<l_{z}$.
Finally $r_{z}<r_{u_{2}}$, since $z<_{\sigma }u_{2}$ by assumption, and thus 
$I_{z}\subseteq I_{u_{2}}$. Then, since $u_{1},u_{2}$ are vertices of the
weakly reducible set $S$ and $I_{z}\subseteq I_{u_{2}}$, it follows by
Definition~\ref{def:weakly.reducible.set} that $u_{1}z\in E(G)$, a
contradiction. Therefore $u_{1}z\in E(G)$ in the case where $z<_{\sigma
}u_{2}$. That is, always $zu_{1}\in E(G)$. Hence, there exists the path $%
P^{\prime }=(P_{1},u_{1},u_{2},w,P_{2^{\prime }})$ which has greater weight
than~$P$, which is a contradiction to the maximality assumption on~$P$.
\end{proof}

\medskip

We are now ready to present our second data reduction. As we prove in
Theorem~\ref{second-data-reduction-thm}, this data reduction maintains the
maximum weight of a path.

\begin{redrule}[second data reduction]
\label{data-reduction-2-redrule}Let $G$ be a weighted interval graph with
weight function ${w:V(G)\rightarrow \mathbb{N}}$ and $\mathcal{I}$ be an
interval representation of~$G$. Let $D$ be a proper interval deletion set of~%
$G$. Finally, let $\mathcal{S}=\{S_{1},S_{2},\dots ,S_{|\mathcal{S}|}\}$ be
a family of pairwise disjoint weakly reducible sets, where ${S_{i}\cap
D=\emptyset}$ for every ${i\in \lbrack |\mathcal{S}|]}$. We recursively
define the graphs $G_{0},G_{1},\dots ,G_{|\mathcal{S}|}$ with the interval
representations $\mathcal{I}_{0},\mathcal{I}_{1},\dots ,\mathcal{I}_{|%
\mathcal{S}|}$ as follows:\vspace{-0,1cm}

\begin{itemize}
\item $G_{0}=G$ and $\mathcal{I}_{0}=\mathcal{I}$,

\item for $1\leq i \leq |\mathcal{S}|$,~$\mathcal{I}_{i}$ is obtained by
replacing in~$\mathcal{I}_{i-1}$ the intervals ${\{I_{v}:v\in S_{i}\}}$ with 
$\min \{|S_{i}|,|D|+4\}$ copies of the interval $I_{S_{i}}=\mathbf{span}%
(S_{i})$, each having equal weight $\frac{1}{\min \{|S_{i}|,|D|+4\}}%
\sum_{u\in S_{i}}w(u)$; all other intervals remain unchanged, and

\item finally $\widehat{G}=G_{|\mathcal{S}|}$ and $\widehat{\mathcal{I}}=%
\mathcal{I}_{|\mathcal{S}|}$.
\end{itemize}
\end{redrule}

Note that in the construction of the interval representation $\mathcal{I}%
_{i} $ by Reduction Rule~\ref{data-reduction-2-redrule}, where $i\in \lbrack
|\mathcal{S}|]$, we can always slightly perturb the endpoints of the $\min
\{|S_{i}|,|D|+4\}$ copies of the interval $I_{S_{i}}=\mathbf{span}(S_{i})$
such that all endpoints remain distinct in $\mathcal{I}_{i}$, and such that
these $\min \{|S_{i}|,|D|+4\}$ newly introduced intervals induce a proper
interval representation in $\mathcal{I}_{i}$.

Next we prove in Lemma~\ref{lem:weakly.consecutive.vertices} that, for every
weakly reducible set $S$, every maximum-weight path $P$ which contains the
whole set $S$ can be rewritten as a path $P^{\prime }$, where the appearance
of the vertices of $S$ in $P$ is \emph{interrupted at most} $|D|+3$ times by
vertices outside $S$. That is, such a path~$P^{\prime }$ has at most $\min
\{|S|,|D|+4\}$ vertex-maximal subpaths with vertices from $S$. Before we
prove Lemma~\ref{lem:weakly.consecutive.vertices}, we first provide the
following auxiliary lemma.

\begin{lemma}
\label{lem:invariant.weakly.reducible}Let $G$ be an interval graph and $%
\mathcal{I}$ be an interval representation of~$G$. Let also $S,S^{\prime
}\subseteq V(G)$ such that $S$ is a weakly reducible set of~$G$ and $S\cap
S^{\prime }=\emptyset $. Let $G^{\prime }$ be the graph obtained from $G$ by
replacing $\mathcal{I}[S^{\prime }]$ by $\mathbf{span}(S^{\prime })$. Then $%
S $ remains a weakly reducible set of~$G^{\prime }$.
\end{lemma}

\begin{proof}
Let $u$ be the vertex of~$V(G^{\prime })\setminus V(G)$, i.e.,~$I_{u}=%
\mathbf{span}(S^{\prime })$. Denote by $\mathcal{I}^{\prime }$ the interval
representation obtained from $\mathcal{I}$ after replacing $\mathcal{I}%
[S^{\prime }]$ by $\mathbf{span}(S^{\prime })$. First note that $\mathcal{I}%
^{\prime }[S]=\mathcal{I}[S]$, and thus $\mathcal{I}^{\prime }[S]$ induces a
connected proper interval representation as $S$ is a weakly reducible set by
assumption. This proves Condition~1 of Definition~\ref%
{def:weakly.reducible.set}.

Let now $x\in V(G^{\prime })$ and $v\in S$. Assume that $I_{x}\subseteq
I_{v} $. If $x\neq u$, then $x\in V(G)$, and thus $I_{x}\subseteq I_{v}$ in $%
\mathcal{I}$. Therefore, since $S$ is a weakly reducible set of $G$ by
assumption, it follows that $S\subseteq N(x)$. If $x=u$, then $%
I_{x}=I_{u}\subseteq I_{v}$ in $\mathcal{I}^{\prime }$. Thus, since $I_{u}=%
\mathbf{span}(S^{\prime })$, it follows that for every vertex $x^{\prime
}\in S^{\prime }$, we have $I_{x^{\prime }}\subseteq I_{u}\subseteq I_{v}$.
Thus, since $S$ is a weakly reducible set of $G$, it follows that $%
S\subseteq N(x^{\prime })$, i.e.,~$I_{x^{\prime }}\cap I_{w}\neq \emptyset $
for every $w\in S$. Therefore, since $I_{x^{\prime }}\subseteq I_{u}$, it
follows that also $I_{u}\cap I_{w}\neq \emptyset $ for every $w\in S$, i.e.,~%
$S\subseteq N(u)=N(x)$. This proves Condition~2 of Definition~\ref%
{def:weakly.reducible.set}.
\end{proof}

\medskip

The next observation follows directly by the definition of the graphs $%
G_{0},G_{1},\ldots ,G_{|\mathcal{S}|}$ (see the Reduction Rule~\ref%
{data-reduction-2-redrule}) and by Lemma~\ref{lem:invariant.weakly.reducible}%
.

\begin{observation}
\label{obs:weakly.reducible.invariant2} For every $i\in [|\mathcal{S}|]$,
the sets $S_{i}$, $\dots$, $S_{|\mathcal{S}|}$ are weakly reducible sets of~$%
G_{i-1}$.
\end{observation}

\begin{lemma}
\label{lem:weakly.consecutive.vertices} Let $G$ be a weighted interval graph
with weight function $w:V(G)\rightarrow \mathbb{N}$ and $\mathcal{I}$ be an
interval representation of~$G$. Let $D$ be a proper interval deletion set of 
$G$ and $\mathcal{S}=\{S_{1},S_{2},\dots ,S_{|\mathcal{S}|}\}$ be a family
of pairwise disjoint weakly reducible sets of~$G$ such that $S_{i}\cap
D=\emptyset $, $i\in \lbrack |\mathcal{S}|]$. Also, let $i\in \lbrack |%
\mathcal{S}|]$ and $G_{i-1}$ be one of the graphs obtained by Reduction Rule~%
\ref{data-reduction-2-redrule}. Finally, let $P$ be a path of maximum weight
in~$G_{i-1}$ and let $S_{i}\subseteq V(P)$. Then there exists a path $%
P^{\prime }$ in~$G_{i-1}$ on the same vertices as $P$, which has at most $%
\min \{|S_{i}|,|D|+4\}$ vertex-maximal subpaths consisting only of vertices
from $S_{i}$.
\end{lemma}

\begin{proof}
Let $P$ be a path of maximum weight in~$G_{i-1}$ such that $S_{i}\subseteq
V(P)$. Denote $S_{i}=\{u_{1},u_{2},\dots ,u_{|S_{i}|}\}$, where $%
u_{1}<_{\sigma }u_{2}<_{\sigma }\dots <_{\sigma }u_{|S_{i}|}$. Without loss
of generality we may assume by Lemma~\ref{lem:normal.path.existence} that $P$
is normal. Observation~\ref{obs:weakly.reducible.invariant2} yields that $%
S_{i}$ is a weakly reducible set of~$G_{i-1}$. Thus, since $\mathcal{I}%
[S_{i}]$ induces a proper interval representation of~$G_{i-1}[S_{i}]$
(cf.~Definition~\ref{def:weakly.reducible.set}), Lemma~\ref%
{lem:propers.consecutive.in.path} implies that $u_{1}<_{P}u_{2}<_{P}\dots
<_{P}u_{|S_{i}|}$, i.e.,~the vertices of~$S_{i}$ appear in the same order
both in the vertex ordering $\sigma $ and in the path $P$. Furthermore,
Observation~\ref{obs:weakly.reducible.clique} implies that $u_{w}u_{w+1}\in
E(G_{i-1})$, for every $w\in \lbrack |S_{i}|-1]$. Let $%
P=(P_{0},u_{1},P_{1},u_{2},\dots ,P_{|S_{i}|-1},u_{|S_{i}|},P_{|S_{i}|})$.
For every $w\in \lbrack |S_{i}|-1]$, we denote 
\begin{equation*}
z_{w}=%
\begin{cases}
\text{the first vertex of }P_{w} & \text{if }P_{w}\neq \emptyset  \\ 
u_{w+1} & \text{otherwise}%
\end{cases}%
.
\end{equation*}

Let 
\begin{eqnarray*}
Z_{1} &=&\{z_{w}:~z_{w}<_{\sigma }u_{w},w\in \lbrack |S_{i}|-1]\}, \\
Z_{2} &=&\{z_{w}:~u_{w}<_{\sigma }z_{w},w\in \lbrack |S_{i}|-1]\}.
\end{eqnarray*}%
Note that, as $u_{w}<_{\sigma }u_{w+1}$, for every $w\in \lbrack |S_{i}|-1]$
it follows that $Z_{1}\cap S_{i}=\emptyset $.

In the first part of the proof we show that $|Z_{1}|\leq |D|+2$. First we
show that for every $z\in Z_{1}$, $S_{i}\subseteq N(z)$. Let $w\in \lbrack
|S_{i}|-1]$, such that $z_{w}\in Z_{1}$, i.e.,~$z_{w}<_{\sigma }u_{w}$. Then
from Lemma~\ref{lem:turning.point}, $I_{z_{w}}\subseteq I_{u_{w}}$.
Therefore, since $u_{w}$ is a vertex of the weakly reducible set of~$S_{i}$
in~$G_{i-1}$, it follows by Definition~\ref{def:weakly.reducible.set} that $%
S_{i}\subseteq N(z_{w})$ in the graph~$G_{i-1}$. Thus $S_{i}\subseteq N(z)$
for every $z\in Z_{1}$.

We now prove that $Z_{1}$ is an independent set. Towards a contradiction we
assume that there exist $z_{w},z_{j}\in Z_{1}$ with $z_{w}<_{P}z_{j}$ and $%
z_{w}z_{j}\in E(G_{i-1})$. Then, since $u_{w}z_{w},u_{w}z_{j}\in E(G_{i-1})$
and $z_{w}$ is the successor of~$u_{w}$ in the normal path $P$, it follows
that $z_{w}<_{\sigma }z_{j}$, i.e.,~$z_{w}$ and $z_{j}$ appear in the same
order in the vertex ordering $\sigma $ and in the path $P$. Let $v$ be the
successor of~$z_{w}$ in~$P$. Then since $z_{w}<_{\sigma }u_{w}$, Lemma~\ref%
{lem:not.turning.point} implies that $z_{w}<_{\sigma }v$. Furthermore, as $%
z_{w}<_{\sigma }z_{j}$, from Lemma~\ref{lem:endiamesa} we obtain that $%
z_{w}<_{\sigma }u<_{\sigma }z_{j}$, for every $u\in V(P)$ such that $%
z_{w}<_{P}u<_{P}z_{j}$. Since $u_{j}$ is the predecessor of~$z_{j}$ in~$P$,
it follows that $z_{w}<_{P}u_{j}<_{P}z_{j}$. Therefore $u_{j}<_{\sigma
}z_{j} $, a contradiction as $z_{j}\in Z_{1}$. Therefore $Z_{1}$ is an
independent set.

Assume now (towards a contradiction) that $|Z_{1}|\geq |D|+3$. Then $%
|Z_{1}\setminus D|\geq |Z_{1}|-|D|\geq 3$. Hence, there exist at least 3
vertices $z^{\prime },z^{\prime \prime },z^{\prime \prime \prime }\in
Z_{1}\setminus D$. Let $u\in S_{i}$. Note that $u\notin D$, since $S_{i}\cap
D=\emptyset $. Then $u,z^{\prime },z^{\prime \prime },z^{\prime \prime
\prime }\in V(G_{i-1}\setminus D)$. Moreover, recall that $S_{i}\subseteq
N(z)$, for every $z\in Z_{1}$. Thus $uz^{\prime },uz^{\prime \prime
},uz^{\prime \prime \prime }\in E(G_{i-1}\setminus D)$. Recall also that $%
Z_{1}$ is an independent set and thus the vertices $z^{\prime },z^{\prime
\prime },z^{\prime \prime \prime }$ form an independent set in~$%
G_{i-1}\setminus D$. Thus, $\{u,z^{\prime },z^{\prime \prime },z^{\prime
\prime \prime }\}$ induce a $K_{1,3}$ in~$G_{i-1}\setminus D$. Notice that
by construction of~$G_{i-1}$ (cf.~Reduction Rule~\ref%
{data-reduction-2-redrule}) each one of the vertices $z^{\prime },z^{\prime
\prime },z^{\prime \prime \prime }$ is either a vertex of~$G$, or it appears
in~$G_{i-1}$ as the replacement of some weakly reducible set $S$. If $%
z^{\prime }\notin V(G)$, then denote by $S^{\prime }$ the weakly reducible
set corresponding to vertex $z^{\prime }$. Otherwise, if $z^{\prime }\in
V(G) $, we define $S^{\prime }=\{z^{\prime }\}$. Similarly, if $z^{\prime
\prime }\notin V(G)$ (resp.~if $z^{\prime \prime \prime }\notin V(G)$) then
denote by $S^{\prime \prime }$ (resp.~by $S^{\prime \prime \prime }$) the
weakly reducible set corresponding to vertex $z^{\prime \prime }$ (resp.~$%
z^{\prime \prime \prime }$). Otherwise, if $z^{\prime \prime }\in V(G)$
(resp.~if $z^{\prime \prime \prime }\in V(G)$), we define $S^{\prime \prime
}=\{z^{\prime \prime }\}$ (resp.~$S^{\prime \prime \prime }=\{z^{\prime
\prime \prime }\}$). Since $z^{\prime }u\in E(G_{i-1})$, observe that there
always exists at least one vertex $v^{\prime }\in S^{\prime }$ such that $%
I_{v^{\prime }}\cap I_{u}\neq \emptyset $, and thus $uv^{\prime }\in E(G)$.
Similarly, since $z^{\prime \prime }u,z^{\prime \prime \prime }u\in
E(G_{i-1})$, there always exist vertices $v^{\prime \prime }\in S^{\prime
\prime }$ and $v^{\prime \prime \prime }\in S^{\prime \prime \prime }$ such
that $I_{v^{\prime \prime }}\cap I_{u}\neq \emptyset $ and $I_{v^{\prime
\prime \prime }}\cap I_{u}\neq \emptyset $, and thus $uv^{\prime \prime
},uv^{\prime \prime \prime }\in E(G)$. Note that $v^{\prime },v^{\prime
\prime },v^{\prime \prime \prime }\in V(G)\setminus D$. Furthermore, by the
definition of~$v^{\prime },v^{\prime \prime },v^{\prime \prime \prime }$, it
follows that $I_{v^{\prime }}\subseteq I_{z^{\prime }}$, $I_{v^{\prime
\prime }}\subseteq I_{z^{\prime \prime }}$, and $I_{v^{\prime \prime \prime
}}\subseteq I_{z^{\prime \prime \prime }}$. Therefore, since $z^{\prime
},z^{\prime \prime },z^{\prime \prime \prime }$ form an independent set, the
vertices $v^{\prime },v^{\prime \prime },v^{\prime \prime \prime }$ also
form an independent set. Furthermore, as $uv^{\prime },uv^{\prime \prime
},uv^{\prime \prime \prime }\in E(G)$, it follows that $\{u,v^{\prime
},v^{\prime \prime },v^{\prime \prime \prime }\}$ induce a $K_{1,3}$ in~$%
G\setminus D$, which is a contradiction to the assumption that $D$ is a
proper interval deletion set of~$G$. Thus $|Z_{1}|\leq |D|+2$. This
completes the first part of our proof.

Let now $z_{i_{0}}$ be the vertex of~$Z_{2}$ that appears first in~$P$. In
the second part of the proof we show that the set $\{x\in
V(G_{i-1}):~u_{i_{0}}<_{P}x<_{P}u_{|S_{i}|}\}$ induces a clique in~$G_{i-1}$%
. Note first that, since $u_{i_{0}}<_{\sigma }z_{i_{0}}\leq _{\sigma
}u_{|S_{i}|}$ and $u_{i_{0}}u_{|S_{i}|}\in E(G_{i-1})$, Lemma~\ref%
{lem:endiamesa} implies that $u_{i_{0}}<_{\sigma }x<_{\sigma }u_{|S_{i}|}$,
for every $x\in V(G_{i-1})$ such that $u_{i_{0}}<_{P}x<_{P}u_{|S_{i}|}$.
Therefore, since $u_{i_{0}}u_{|S_{i}|}\in E(G_{i-1})$, it follows by Lemma~%
\ref{lem:ordering.edges} that $xu_{|S_{i}|}\in E(G_{i-1})$ for every $x\in
V(G_{i-1})$ such that $u_{i_{0}}<_{\sigma }x<_{\sigma }u_{|S_{i}|}$.

We now prove that $u_{i_{0}}x\in E(G_{i-1})$, for every $x\in V(G_{i-1})$
such that $u_{i_{0}}<_{\sigma }x<_{\sigma }u_{|S_{i}|}$. Suppose otherwise
that there exists such an $x$ where $u_{i_{0}}x\notin E(G_{i-1})$. Note that 
$l_{u_{|S_{i}|}}<r_{u_{i_{0}}}$, since $u_{i_{0}}<_{\sigma }u_{|S_{i}|}$ and 
$u_{i_{0}}u_{|S_{i}|}\in E(G_{i-1})$. Furthermore, since $S_{i}$ is a weakly
reducible set, $\mathcal{I}[S_{i}]$ induces a proper interval representation
of~$G_{i-1}[S_{i}]$ by Definition~\ref{def:weakly.reducible.set}. Then,
since $u_{i_{0}}<_{\sigma }u_{|S_{i}|}$, it follows by Observation~\ref%
{obs:order.in.uirepresentation} that $l_{u_{i_{0}}}<l_{u_{|S_{i}|}}$. That
is, $l_{u_{i_{0}}}<l_{u_{|S_{i}|}}<r_{u_{i_{0}}}$. Hence, since $%
u_{|S_{i}|}x\in E(G_{i-1})$ and $u_{i_{0}}x\notin E(G_{i-1})$ by our
assumption, it follows that $r_{u_{i_{0}}}<l_{x}$, i.e.,~$%
l_{u_{i_{0}}}<l_{u_{|S_{i}|}}<r_{u_{i_{0}}}<l_{x}$. Finally $%
r_{x}<r_{u_{|S_{i}|}}$ from the choice of~$x$, and thus $I_{x}\subseteq
I_{u_{|S_{i}|}}$. Therefore, since $S_{i}$ is a weakly reducible set, $ux\in
E(G_{i-1})$ for every $u\in S_{i}$ by Definition~\ref%
{def:weakly.reducible.set}. This is a contradiction to our assumption that $%
u_{i_{0}}x\notin E(G_{i-1})$. Therefore $xu_{i_{0}}\in E(G_{i-1})$ for every 
$x\in V(G_{i-1})$ such that $u_{i_{0}}<_{P}x<_{P}u_{|S_{i}|}$.

We finally show that the set $\{x\in
V(G_{i-1}):~u_{i_{0}}<_{P}x<_{P}u_{|S_{i}|}\}$ induces a clique in~$G_{i-1}$%
. Since $u_{i_{0}}x\in E(G_{i-1})$ for every vertex $x$ in this set, we
obtain that $l_{x}<r_{u_{i_{0}}}<r_{x}$. Hence, the set $\{x\in
V(G_{i-1}):~u_{i_{0}}<_{P}x<_{P}u_{|S_{i}|}\}$ induces a clique in~$G_{i-1}$%
, as all such intervals $I_{x}$ contain the point $r_{u_{i_{0}}}$. This
completes the second part of our proof.

In the third part of our proof we show that there exists a path $P^{\prime }$
which has at most $|D|+4$ vertex-maximal subpaths that consist only of
vertices from $S_{i}$. Define $P^{\prime
}=(P_{0},u_{1},P_{1},u_{2},P_{2},\dots
,P_{i_{0}-1},u_{i_{0}},P_{i_{0}},P_{i_{0}+1},\dots
,P_{|S_{i}|-1},u_{i_{0}+1},\dots ,u_{|S_{i}|},P_{|S_{i}|})$. Notice that $%
V(P^{\prime })=V(P)$. From the second part of our proof the set $%
\bigcup_{j=i_{0}}^{|S_{i}|-1}V(P_{j})\bigcup_{j=i_{0}+1}^{|S_{i}|-1}\{u_{j}%
\} $ induces a clique. Thus, $P^{\prime }$ is a path of~$G_{i-1}$ with $%
V(P^{\prime })=V(P)$. From the choice of~$i_{0}$, for every $w\in \lbrack
i_{0}-1]$, if $P_{w}\neq \emptyset $, then $z_{w}\in Z_{1}$. Since $%
|Z_{1}|\leq |D|+2$, from the first part of our proof, there exist at most $%
|D|+2$ paths~$P_{w}$, $w\in \lbrack i_{0}-1]$, that are not empty. Thus, the
subpath $P_{1}^{\prime }=(P_{0},u_{1},P_{1},u_{2},P_{2},\dots ,P_{i_{0}-1})$
contains at most $|D|+2$ vertex-maximal subpaths consisting only of vertices
of~$S_{i}$. Furthermore, the subpath $P_{2}^{\prime
}=(u_{i_{0}},P_{i_{0}},P_{i_{0}+1},\dots ,P_{|S_{i}|-1},u_{i_{0}+1},\dots
,u_{|S_{i}|},P_{|S_{i}|})$ of~$P^{\prime }$ clearly contains two
vertex-maximal subpaths consisting only of vertices of~$S_{i}$, namely the
subpaths $(u_{i_{0}})$ and $(u_{i_{0}},\dots ,u_{|S_{i}|})$. Since, $%
P^{\prime }=(P_{1}^{\prime },P_{2}^{\prime })$, it follows that $P^{\prime } 
$ has at most $|D|+4$ vertex-maximal subpaths that consist only of vertices
from $S_{i}$. Moreover, $P^{\prime }$ has clearly at most $|S_{i}|$ such
vertex-maximal subpaths from vertices of~$S_{i}$. This completes the proof
of the lemma.
\end{proof}

\medskip

We now present the next auxiliary technical lemma that will be used to prove
the correctness of Reduction Rule~\ref{data-reduction-2-redrule} in Theorem~%
\ref{second-data-reduction-thm}.

\begin{lemma}
\label{lem:equal.weight.second.replacement.prelim}Let $\ell $ and $k$ be
positive integers. Let $G$ be a weighted interval graph with weight function 
$w:V(G)\rightarrow \mathbb{N}$ and $\mathcal{I}$ be an interval
representation of~$G$. Also, let $S$ be a weakly reducible set of~$G$.
Finally, let $G^{\prime }$ be the graph obtained from $G$ by replacing $%
\mathcal{I}[S]$ with $\min \{|S|,k+4\}$ copies of the interval $I_{v_{j}}=%
\mathbf{span}(S)$, each having weight $\frac{1}{\min \{|S|,k+4\}}\sum_{u\in
S}w(u)$. If the maximum weight of a path in~$G^{\prime }$ is $\ell $, then
the maximum weight of a path in~$G$ is at least $\ell $.
\end{lemma}

\begin{proof}
Denote $S=\{u_{1},u_{2},\dots ,u_{|S|}\}$, where $u_{1}<_{\sigma
}u_{2}<_{\sigma }\dots <_{\sigma }u_{|S|}$. That is, $r_{u_{1}}<r_{u_{2}}<%
\ldots <r_{u_{|S|}}$. Furthermore, since $\mathcal{I}[S]$ induces a proper
interval representation, it follows by Observation~\ref%
{obs:order.in.uirepresentation} that also $l_{u_{1}}<l_{u_{2}}<\ldots
<l_{u_{|S|}}$. That is, $l_{u_{|S|}}=\max \{l_{u}:~u\in S\}$ and $%
r_{u_{1}}=\min \{r_{u}:~u\in S\}$.

Assume that the maximum weight of a path $P^{\prime }$ in~$G^{\prime }$ is $%
\ell $. Define $V_{0}=V(G^{\prime })\setminus V(G)$ and denote $%
V_{0}=\{v_{j}:~j\in \left[ \min \{|S|,k+4\}\right] \}$. If $V(P^{\prime
})\cap V_{0}=\emptyset $, then $P^{\prime }$ is also a path in~$G$. Assume
now that $V(P^{\prime })\cap V_{0}\neq \emptyset $. Then we may assume
without loss of generality that $V_{0}\subseteq V(P^{\prime })$. Indeed,
otherwise we can augment $P^{\prime }$ to a path of~$G^{\prime }$ with
greater weight by adding the missing copies of~$\mathbf{span}(S)$ right
after the last copy of~$\mathbf{span}(S)$ in~$P^{\prime }$, which is a
contradiction to the maximality assumption on $P^{\prime }$. Furthermore, by
Lemma~\ref{lem:normal.path.existence} we may assume without loss of
generality that $P^{\prime }$ is normal.

Let $P^{\prime }=\left( P_{0}^{\prime },v_{1},P_{1}^{\prime },v_{2},\dots
,v_{|V_{0}|-1},P_{|V_{0}|-1}^{\prime },v_{|V_{0}|},P_{|V_{0}|}^{\prime
}\right) $. Denote by $Z$ the set of all predecessors and successors of the
vertices of~$V_{0}$ in the path $P^{\prime }$. Consider a vertex $z\in Z\cap
\left( V(P_{1}^{\prime })\cup V(P_{2}^{\prime })\cup \ldots \cup
V(P_{|V_{0}|-1}^{\prime })\right) $. Since $z$ is a predecessor or a
successor of a vertex $v\in V_{0}$ in~$P^{\prime }$, where $I_{v}=\mathbf{%
span}(S)$, it follows that $N(z)\cap S\neq \emptyset $. Assume that $%
S\nsubseteq N(z)$ in the initial graph $G$. Recall that $%
[l_{u_{|S|}},r_{u_{1}}]=[\max \{l_{v}:~v\in S\},\min \{r_{v}:~v\in S\}]$.
Therefore, since we assumed that $S\nsubseteq N(z)$, Observation~\ref%
{obs:common.intersection.weakly.reducible} implies that $I_{z}\cap \lbrack
l_{u_{|S|}},r_{u_{1}}]=\emptyset $, and thus either $r_{z}<l_{u_{|S|}}$ or $%
r_{u_{1}}<l_{z}$.

Suppose first that $r_{z}<l_{u_{|S|}}$, i.e.,~$r_{z}<l_{u_{|S|}}<r_{u_{1}}$.
Let $l_{u_{1}}<l_{z}$, i.e.,~$l_{u_{1}}<l_{z}<r_{z}<r_{u_{1}}$. Then $%
I_{z}\subseteq I_{u_{1}}$, and thus $S\subseteq N(z)$, since $S$ is a weakly
reducible set (cf.~Definition~\ref{def:weakly.reducible.set}). This is a
contradiction to our assumption that $S\nsubseteq N(z)$. Let $%
l_{z}<l_{u_{1}} $, i.e.,~$l_{z}<l_{u_{1}}<r_{z}<r_{u_{1}}$. Then $\mathcal{I}%
[\{z,\mathbf{span}(S)\}]$ induces a proper interval representation.
Therefore, since the path $P^{\prime }$ of~$G^{\prime }$ is normal, Lemma %
\ref{lem:propers.consecutive.in.path} implies that vertex~$z$ appears in~$%
P^{\prime }$ before the first vertex $v_{1}$ of~$V_{0}$, i.e.,~$z\in
P_{0}^{\prime }$. This is a contradiction to our assumption that $z\in Z\cap
\left( V(P_{1}^{\prime })\cup V(P_{2}^{\prime })\cup \ldots \cup
V(P_{|V_{0}|-1}^{\prime })\right) $.

Suppose now that $r_{u_{1}}<l_{z}$, i.e.,~$l_{u_{|S|}}<r_{u_{1}}<l_{z}$. Let 
$r_{z}<r_{u_{|S|}}$, i.e.,~$l_{u_{|S|}}<l_{z}<r_{z}<r_{u_{|S|}}$. Then $%
I_{z}\subseteq I_{u_{|S|}}$, and thus $S\subseteq N(z)$, since $S$ is a
weakly reducible set (cf.~Definition~\ref{def:weakly.reducible.set}). This
is a contradiction to our assumption that $S\nsubseteq N(z)$. Let $%
r_{u_{|S|}}<r_{z}$, i.e.,~$l_{u_{|S|}}<l_{z}<r_{u_{|S|}}<r_{z}$. Then $%
\mathcal{I}[\{z,\mathbf{span}(S)\}]$ induces a proper interval
representation. Therefore, since the path $P^{\prime }$ of~$G^{\prime }$ is
normal, Lemma~\ref{lem:propers.consecutive.in.path} implies that vertex $z$
appears in~$P^{\prime }$ after the last vertex $v_{|V_{0}|}$ of~$V_{0}$,
i.e.,~$z\in P_{|V_{0}|}^{\prime }$. This is a contradiction to our
assumption that $z\in Z\cap \left( V(P_{1}^{\prime })\cup V(P_{2}^{\prime
})\cup \ldots \cup V(P_{|V_{0}|-1}^{\prime })\right) $.

Summarizing, $S\subseteq N(z)$ in the initial graph $G$, for every vertex $%
z\in V(P_{1}^{\prime })\cup V(P_{2}^{\prime })\cup \ldots \cup
V(P_{|V_{0}|-1}^{\prime })$ which is a predecessor or a successor of a
vertex of~$V_{0}$ in~$P^{\prime }$. Therefore $P=\left( P_{0}^{\prime
},u_{1},P_{1}^{\prime },u_{2},\dots ,u_{|V_{0}|-1},P_{|V_{0}|-1}^{\prime
},u_{|V_{0}|},u_{|V_{0}|+1},\ldots ,u_{|S|},P_{|V_{0}|}^{\prime }\right) $
is a path in~$G$, where $V(P)=V(P^{\prime })$. Finally, since $\sum_{i\in
\lbrack |S|]}w(u_{i})=\sum_{j\in \left[ \min \{|S|,k+4\}\right] }w(v_{j})$,
it follows that $\sum_{v\in V(P)}w(v)=\sum_{v\in V(P^{\prime })}w(v)$. Thus $%
G$ has a path $P$ of weight at least~$\ell $.
\end{proof}

\medskip

Now we are ready to prove the correctness of our Reduction Rule~\ref%
{data-reduction-2-redrule}.

\begin{theorem}
\label{second-data-reduction-thm}Let $\ell $ be a positive integer. Let $G$
be a weighted interval graph and $\widehat{G}$ be the weighted interval
graph obtained by Reduction Rule~\ref{data-reduction-2-redrule}. Then the
maximum weight of a path in~$G$ is $\ell $ if and only if the maximum weight
of a path in~$\widehat{G}$ is $\ell $.
\end{theorem}

\begin{proof}
First assume that the maximum weight of a path in~$\widehat{G}=G_{|\mathcal{S%
}|}$ is $\ell $. Then, by iteratively applying Lemma~\ref%
{lem:equal.weight.second.replacement.prelim} it follows that the maximum
weight of a path in~$G$ is at least $\ell $.

Conversely, assume that the maximum weight of a path in~$G$ is $\ell $. We
show by induction on $i\in \{0,1,\ldots ,|\mathcal{S}|\}$ that the maximum
weight of a path in~$G_{i}$ is at least $\ell $. For $i=0$ we have $G_{0}=G$
and the argument follows by our assumption. This proves the induction basis.

For the induction step, let $i\geq 1$ and assume that the weight of a
maximum path in~$G_{i-1}$ is at least $\ell $. Let $P$ be a path of maximum
weight in~$G_{i-1}$, i.e.,~the weight of~$P$ is at least $\ell $. Without
loss of generality, from Lemma~\ref{lem:normal.path.existence}, we may also
assume that $P$ is normal. Furthermore, Lemma~\ref%
{lem:full.weakly.reducible.or.none} implies that either $S_{i}\subseteq V(P)$
or $S_{i}\cap V(P)=\emptyset $. Notice that if $S_{i}\cap V(P)=\emptyset $,
then $P$ is also a path of~$G_{i}$. Suppose that $S_{i}\subseteq V(P)$. Then
from Lemma~\ref{lem:weakly.consecutive.vertices}, we can obtain a path $%
\widehat{P}$ such that $V(\widehat{P})=V(P)$ and such that $\widehat{P}$ has 
$q\leq \min \{|S_{i}|,|D|+4\}$ vertex-maximal subpaths consisting only of
vertices of~$S_{i}$. Let $\widehat{P}_{1},\widehat{P}_{2},\dots ,\widehat{P}%
_{q}$ be those subpaths. Consider now the path $P^{\prime }$ that is
obtained by replacing in the interval representation of~$\widehat{P}$ each
of the subpaths $\widehat{P}_{1},\widehat{P}_{2},\dots ,\widehat{P}_{q-1}$
with a copy of~$\mathbf{span}(S_{i})$, and by replacing the subpath $%
\widehat{P}_{q}$ with $\min \{|S_{i}|,|D|+4\}-q$ copies of~$\mathbf{span}%
(S_{i})$. Note that $P^{\prime }$ is a path in the graph~$G_{i}$. Recall by
the definition of~$G_{i}$ (cf.~Reduction Rule~\ref{data-reduction-2-redrule}%
) that each of the $\min \{|S_{i}|,|D|+4\}$ copies of the interval $\mathbf{%
span}(S_{i})$ has weight $w(v_{j})=\frac{1}{\min \{|S_{i}|,|D|+4\}}%
\sum_{u\in S_{i}}w(u)$. Since the total weight of all these copies of~$%
\mathbf{span}(S_{i})$ is equal to $\sum_{u\in S_{i}}w(u)$, it follows that $%
\sum_{v\in V(P^{\prime })}w(v)=\sum_{v\in V(P)}w(v)$. That is, the weight of
the path $P^{\prime }$ of~$G_{i}$ is at least $\ell $. Therefore the maximum
weight of a path in~$G_{i}$ is at least $\ell $. This completes the
induction step and the proof of the theorem.
\end{proof}

\section{Special weighted interval graphs\label{special-interval-sec}}

In this section we sequentially apply the two data reductions of Sections~%
\ref{first-data-reduction-subsec} and~\ref{second-data-reduction-subsec} to
a given interval graph~$G$ with a proper interval deletion set $D$. To do
so, we first define a specific family $\mathcal{S}_{1}$ of \emph{reducible}
sets in~$G\setminus D$ and we apply Reduction Rule~\ref%
{data-reduction-1-redrule} to $G$ with respect to the family $\mathcal{S}%
_{1} $, resulting in the weighted interval graph~$G^{\#}$. Then we define a
specific family $\mathcal{S}_{2}$ of \emph{weakly reducible} sets in~$%
G^{\#}\setminus D$ and we apply Reduction Rule~\ref{data-reduction-2-redrule}
to $G^{\#}$ with respect to the family $\mathcal{S}_{2}$, resulting in the
weighted interval graph $\widehat{G}$. As it turns out, the vertex sets of $%
\mathcal{S}_{1}\cup \mathcal{S}_{2}$ are a partition of the graph $%
G\setminus D$. The final graph $\widehat{G}$ is then given as input to our
fixed-parameter algorithm of Section~\ref{algorithm-sec}.

We now introduce the notion of a \emph{special weighted interval graph} with 
\emph{parameter} $\kappa $. As we will prove at the end of this section, the
constructed graph $\widehat{G}$ is a special weighted interval graph with
parameter $\kappa $, where $\kappa $ depends only on the size of~$D$
(cf.~Theorem~\ref{special-weighted-interval-graph-thm}). Furthermore $%
\widehat{G}$ can be computed in $O(k^{2}n)$ time (cf.~Theorem~\ref%
{second-reduction-computation-thm}).

\begin{definition}[special weighted interval graph with parameter $\protect%
\kappa $]
\label{special-interval-def} Let $G=(V,E)$ be a weighted interval graph, $%
\mathcal{I}=\{I_{v}:v\in V\}$ be an interval representation of~$G$, and $%
\kappa \in \mathbb{N}$, where the vertex set $V$ can be partitioned into two
sets $A$ and $B$ such that:

\begin{enumerate}
\item $A$ is an independent set in~$G$,

\item for every $v\in A$ and every $u\in V\setminus \{v\}$, we have $%
I_{u}\nsubseteq I_{v}$, and 

\item $|B|\leq \kappa $.
\end{enumerate}

Then $G$ (resp.~$\mathcal{I}$) is a \emph{special weighted interval graph}
(resp.~\emph{special interval representation}) with \emph{parameter} $\kappa 
$. The partition $V=A\cup B$ is a \emph{special vertex partition} of~$G$.
\end{definition}

\subsection{The graph $G^{\#}$\label{graph-G-diesi-subsec}}

Let $D$ be a proper interval deletion set of~$G$ with $k$ vertices $%
d_{1},d_{2},\ldots ,d_{k}$, i.e.,~$G\setminus D$ is a proper interval graph.
First we add two isolated dummy vertices $d_{0},d_{k+1}$ to the set $D$
(together with the corresponding intervals $I_{d_{0}},I_{d_{k+1}}$), such
that $d_{0}<_{\sigma }v<_{\sigma }d_{k+1}$ for every vertex $v\in V(G)$.
Note that $\mathcal{I}^{\prime }=\mathcal{I}\cup \{I_{d_{0}},I_{d_{k+1}}\}$
is the interval representation of an interval graph $G^{\prime }$, where $%
V(G^{\prime })=V(G)\cup \{d_{0},d_{k+1}\}$ and $E(G^{\prime })=E(G)$. For
simplicity of the presentation we refer in the following to $G^{\prime }$
and $\mathcal{I}^{\prime }$ by $G$ and $\mathcal{I}$, respectively.
Furthermore we denote $D=\{d_{0},d_{1},\ldots ,d_{k},d_{k+1}\}$, where $%
d_{0}<_{\sigma }d_{1}<_{\sigma }\ldots <_{\sigma }d_{k}<_{\sigma }d_{k+1}$
and $d_{0},d_{k+1}$ are the two dummy isolated vertices. Now we define the
following four sets: 
\begin{eqnarray*}
L &=&\{l_{v}:~v\in D\}, \\
R &=&\{r_{v}:~v\in D\}, \\
U &=&V(G)\setminus D, \\
U^{\ast } &=&\{u\in U:~I_{u}\cap R=\emptyset \}.
\end{eqnarray*}%
Furthermore, for every $i\in \lbrack k+1]$ we define the following set: 
\begin{equation}
L_{i}=\{l\in L:~r_{d_{i-1}}<l<r_{d_{i}}\}\cup \{r_{d_{i-1}},r_{d_{i}}\}.
\label{Li-eq}
\end{equation}

For every $i\in \lbrack k+1]$ we denote $L_{i}=\{l_{i,0},l_{i,1},\dots
,l_{i,p_{i}}\}$, where $l_{i,0}<l_{i,1}<\dots <l_{i,p_{i}}$. Note that $%
l_{i,0}=r_{d_{i-1}}$ and $l_{i,p_{i}}=r_{d_{i}}$, and thus $|L_{i}|\geq 2$
for every $i\in \lbrack k+1]$. Now, for every $i\in \lbrack k+1]$ and $x\in
\lbrack p_{i}]$ define%
\begin{equation*}
U_{i,x}^{\ast }=\{u\in U^{\ast }:~l_{i,x-1}<r_{u}<l_{i,x}\}
\end{equation*}%
and 
\begin{equation*}
U_{i,x}^{\ast \ast }=\left\{ u\in U_{i,x}^{\ast }:~N(u)\cap \left(
\bigcup\nolimits_{j\in \lbrack x-1]}U_{i,j}^{\ast }\right) =\emptyset
\right\} .
\end{equation*}%
Note that the set $U^{\ast }$ is partitioned by the sets $\{U_{i,x}^{\ast
}:i\in \lbrack k+1],x\in \lbrack p_{i}]\}$.

\begin{observation}
\label{U-star-star-between-consecutive-R-or-L-obs}Let $u\in U$ such that $%
I_{u}$ is strictly contained between two consecutive points of~$R\cup L$.
Then $u\in U_{i,x}^{\ast \ast }$, for some $i\in \lbrack k+1]$ and $x\in
\lbrack p_{i}]$.
\end{observation}

\begin{lemma}
\label{sum-pi-lem}%
\begin{equation*}
\displaystyle\sum_{i=1}^{k+1}p_{i}=2(k+1).
\end{equation*}
\end{lemma}

\begin{proof}
First, note by Eq.~(\ref{Li-eq}) that $|L_{i}|=2+|\{l\in
L:~r_{d_{i-1}}<l<r_{d_{i}}\}|$. Furthermore, since $L_{i}=\{l_{i,0},l_{i,1},%
\dots ,l_{i,p_{i}}\}$, it follows that $p_{i}=|L_{i}|-1$, i.e.,~$%
p_{i}=1+|\{l\in L:~r_{d_{i-1}}<l<r_{d_{i}}\}|$. Therefore, 
\begin{equation*}
\sum_{i=1}^{k+1}p_{i}=(k+1)+\sum_{i=1}^{k+1}|\{l\in
L:~r_{d_{i-1}}<l<r_{d_{i}}\}|=2(k+1).
\end{equation*}
\end{proof}

\begin{lemma}
\label{lem:proper.repres} Let $i\in [k+1]$ and $x\in [p_{i}]$. The interval
representation $\mathcal{I}[U_{i,x}^{*}]$ of~$G[U_{i,x}^{*}]$ is a proper
interval representation.
\end{lemma}

\begin{proof}
The proof is done by contradiction. Assume that, for some $i\in \lbrack k+1]$
and $x\in \lbrack p_{i}]$, $\mathcal{I}[U_{i,x}^{\ast }]$ is not a proper
interval representation of~$G[U_{i,x}^{\ast }]$. That is, there exist two
vertices $u,v\in U_{i,x}^{\ast }$ such that $I_{u}\subseteq I_{v}$. From the
preprocessing of Theorem~\ref{interval-representation-preprocessing-thm}
there exist two vertices $z,z^{\prime }\in V(G)$ such that $z<_{\sigma
}u<_{\sigma }z^{\prime }$, where $z,z^{\prime }\in N(v)\setminus N(u)$.

First suppose that $z,z^{\prime }\notin D$. Then the vertices $\{z,z^{\prime
},v,u\}$ induce a $K_{1,3}$ in~$G\setminus D$, which is a contradiction to
the assumption that $D$ is a proper interval deletion set of~$G$~\cite%
{Roberts69}.

Now suppose that $z\in D$. Since $zv\in E(G)$ and $r_{z}<r_{u}<r_{v}$, it
follows that $l_{v}<r_{z}$. Thus $r_{z}\in I_{v}$. Therefore, since we
assumed that $z\in D$, it follows that $I_{v}\cap R\neq \emptyset $ and thus 
$v\notin U^{\ast }$ (cf.~the definition of the set $U^{\ast }$). This is a
contradiction to the assumption that $v\in U_{i,x}^{\ast }\subseteq U^{\ast
} $.

Finally suppose that $z^{\prime }\in D$. Then, since $u<_{\sigma }z^{\prime
} $ and $z^{\prime }\notin N(u)$, it follows that $r_{u}<l_{z^{\prime }}$.
If $r_{v}<l_{z^{\prime }}$, then $z^{\prime }\notin N(v)$, which is a
contradiction. Thus $l_{z^{\prime }}<r_{v}$, i.e.,~$r_{u}<l_{z^{\prime
}}<r_{v}$. Therefore, since we assumed that $z^{\prime }\in D$, it follows
that $u,v$ do not belong to the same set $U_{i,x}^{\ast }$ (cf.~the
definition of~$U_{i,x}^{\ast }$), which is a contradiction to our assumption.

Thus, for every ${i\in \lbrack k+1]}$ and ${x\in \lbrack p_{i}]}$, ${%
\mathcal{I}[U_{i,x}^{\ast}]}$ is a proper interval representation of~${%
G[U_{i,x}^{\ast}]}$.
\end{proof}

\medskip

For every $i\in \lbrack k+1]$ and $x\in \lbrack p_{i}]$ we denote the
connected components of~$U_{i,x}^{\ast \ast }$ by $C_{i,x}^{1},C_{i,x}^{2},%
\dots ,C_{i,x}^{q(i,x)}$, such that $V(C_{i,x}^{1})<_{\sigma
}V(C_{i,x}^{2})<_{\sigma }\dots <_{\sigma }V(C_{i,x}^{q(i,x)})$. Note that
any two distinct components $C_{i,x}^{t}$ and $C_{i^{\prime },x^{\prime
}}^{t^{\prime }}$ are disjoint, i.e.,~$V(C_{i,x}^{t})\cap V(C_{i^{\prime
},x^{\prime }}^{t^{\prime }})=\emptyset $. Furthermore we define the family $%
\mathcal{S}_{1}$ of vertex subsets of~$V(G)\setminus D$ as follows:%
\begin{equation}
\mathcal{S}_{1}=\{V(C_{i,x}^{t}):i\in \lbrack k+1],\ x\in \lbrack p_{i}],\
t\in \lbrack q(i,x)]\}.  \label{S-1-family-eq}
\end{equation}

\begin{lemma}
\label{all-components-U**-reducible-lem}Every set $S\in \mathcal{S}_{1}$ is
reducible.
\end{lemma}

\begin{proof}
Consider a set $S\in \mathcal{S}_{1}$. Then $S=V(C_{i,x}^{t})$, for some $%
i\in \lbrack k+1]$, $x\in \lbrack p_{i}]$, and $t\in \lbrack q(i,x)]$. We
need to prove that the two conditions of Definition~\ref{def:reducible.set}
are satisfied for the connected component $C_{i,x}^{t}$ of~$U_{i,x}^{\ast
\ast }$. For Condition 1, recall that $V(C_{i,x}^{t})\subseteq U_{i,x}^{\ast
\ast }\subseteq U_{i,x}^{\ast }$. Thus, since the interval representation $%
\mathcal{I}[U_{i,x}^{\ast }]$ of~$G[U_{i,x}^{\ast }]$ is proper by Lemma~\ref%
{lem:proper.repres}, the interval representation $\mathcal{I}[V(C_{i,x}^{t})]
$ of~$G[V(C_{i,x}^{t})]$ is a connected proper interval representation. This
proves Condition~1 of Definition~\ref{def:reducible.set}.

Now let $v\in V(G)$ such that $I_{v}\subseteq \mathbf{span}(V(C_{i,x}^{t}))$%
. Recall by the definition of~$V(C_{i,x}^{t})$ that $G[V(C_{i,x}^{t})]$ is
connected. Therefore, since $I_{u}\cap R=\emptyset $ for every $u\in ~$ $%
V(C_{i,x}^{t})\subseteq U^{\ast }$, it follows that $\mathbf{span}%
(V(C_{i,x}^{t}))\cap R=\emptyset $, and thus also $I_{v}\cap R=\emptyset $.
That is, $v\in U^{\ast }$ (cf.~the definition of~$U^{\ast }$). Let $%
V(C_{i,x}^{t})=\{u_{1},u_{2},\ldots ,u_{a}\}$, where $u_{1}<_{\sigma
}u_{2}<_{\sigma }\ldots <_{\sigma }u_{a}$. Note that $r_{v}\leq r_{u_{a}}$,
since $I_{v}\subseteq \mathbf{span}(V(C_{i,x}^{t}))$ by assumption.
Furthermore, since $\mathcal{I}[V(C_{i,x}^{t})]$ is a proper interval
representation as we proved above, Observation~\ref%
{obs:order.in.uirepresentation} implies that $l_{u_{1}}<l_{u_{2}}<\ldots
<l_{u_{a}}$. Suppose that $v<_{\sigma }u_{1}$. Then, since $I_{v}\subseteq 
\mathbf{span}(V(C_{i,x}^{t}))$, it follows that $I_{v}\subseteq I_{u_{1}}$.
This is a contradiction to Condition~1 of Definition~\ref{def:reducible.set}
that we proved above. Thus $u_{1}\leq _{\sigma }v$. That is, $%
l_{i,x-1}<r_{u_{1}}\leq r_{v}\leq r_{u_{a}}<l_{i,x}$. Therefore $v\in
U_{i,x}^{\ast }$ (cf.~the definition of~$U_{i,x}^{\ast }$). Moreover, since $%
I_{v}\subseteq \mathbf{span}(V(C_{i,x}^{t}))$ by assumption, it follows that 
$N(v)\subseteq \ \bigcup_{u\in V(C_{i,x}^{t})}N(u)$, and thus $N(v)\cap
\bigcup_{j\in \lbrack x-1]}U_{i,j}^{\ast }=\emptyset $ (cf.~the definition
of~$U_{i,x}^{\ast \ast }$). It follows that $v\in U_{i,x}^{\ast \ast }$, and
thus $v\in V(C_{i,x}^{t})$. This proves Condition~2 of Definition~\ref%
{def:reducible.set}.
\end{proof}

\begin{lemma}
\label{span-S-1-disjoint-lem}If $S$ and $S^{\prime }$ are two distinct
elements of~$\mathcal{S}_{1}$, then $\mathbf{span}(S)\cap \mathbf{span}%
(S^{\prime })=\emptyset $.
\end{lemma}

\begin{proof}
Towards a contradiction let $S$ and $S^{\prime }$ be two distinct elements
of~$\mathcal{S}_{1}$. That is, $S=V(C_{i,x}^{t})$ and $S^{\prime}=V(C_{i^{%
\prime },x^{\prime }}^{t^{\prime }})$, for some $i,i^{\prime}\in \lbrack
k+1] $, $x\in \lbrack p_{i}]$, $x^{\prime}\in \lbrack p_{i^{\prime}}]$, and $%
t\in \lbrack q(i,x)]$, $t^{\prime}\in \lbrack q(i^{\prime},x^{\prime})]$.
Assume that $\mathbf{span}(S)\cap \mathbf{span}(S^{\prime })\neq \emptyset $%
. Then there exist $u\in S$ and $u^{\prime }\in S^{\prime } $ such that $%
I_{u}\cap I_{u^{\prime }}\neq \emptyset $.

First suppose that $i\neq i^{\prime }$. Without loss of generality we assume
that $i^{\prime }<i$. Note that $l_{i,x-1}<r_{u}<l_{i,x}$, since $u\in
V(C_{i,x}^{t})\subseteq U_{i,x}^{\ast }$. Furthermore, since by the
definition of the set $L_{i}$ we have $r_{d_{i-1}}\leq l_{i,x-1}$ and $%
l_{i,x}\leq r_{d_{i}}$, it follows that $r_{d_{i-1}}<r_{u}<r_{d_{i}}$.
Similarly it follows that $r_{d_{i^{\prime }-1}}<r_{u^{\prime
}}<r_{d_{i^{\prime }}}$. Therefore, $r_{u^{\prime }}<r_{d_{i^{\prime }}}\leq
r_{d_{i-1}}<r_{u}$. Furthermore, since $I_{u}\cap I_{u^{\prime }}\neq
\emptyset $, it follows that $l_{u}<r_{u^{\prime }}$. That is, $%
l_{u}<r_{u^{\prime }}<r_{d_{i-1}}<r_{u}$, and thus $r_{d_{i-1}}\in I_{u}$.
This is a contradiction, since $u\in U^{\ast }$ (cf.~the definition of~$%
U^{\ast }$). Therefore, $i^{\prime }=i$.

Now suppose that $x^{\prime }\neq x$. Without loss of generality we assume
that $x^{\prime }<x$. Note that $u\in V(C_{i,x}^{t})\subseteq U_{i,x}^{\ast
\ast }$ and that $u^{\prime }\in V(C_{i,x^{\prime }}^{t^{\prime }})\subseteq
U_{i,x^{\prime }}^{\ast }\subseteq \bigcup\nolimits_{j\in \lbrack
x-1]}U_{i,j}^{\ast }$. Therefore, since $~N(u)\cap \left(
\bigcup\nolimits_{j\in \lbrack x-1]}U_{i,j}^{\ast }\right) =\emptyset $ (cf.
the definition of~$U_{i,x}^{\ast \ast }$), it follows that $u^{\prime
}\notin N(u)$. This is a contradiction, since $I_{u}\cap I_{u^{\prime }}\neq
\emptyset $. Therefore, $x^{\prime }=x$.

Summarizing, the sets $S$ and $S^{\prime }$ are two different connected
components of~$U_{i,x}^{\ast \ast }$. Therefore, there are no vertices $u\in
S$ and $u^{\prime }\in S^{\prime }$ such that $I_{u}\cap I_{u^{\prime }}\neq
\emptyset $, which is a contradiction to our assumption. It follows that $%
\mathbf{span}(S)\cap \mathbf{span}(S^{\prime })\neq \emptyset $, for any two
distinct sets $S,S^{\prime }\in \mathcal{S}_{1}$.
\end{proof}

\medskip

Note that, for every $i,x,t$, the connected component $C_{i,x}^{t}$ of~$%
U_{i,x}^{\ast \ast }$ contains no vertices of~$D$, since by definition $%
U_{i,x}^{\ast \ast }\subseteq U=V(G)\setminus D$. Therefore, since all sets
of~$\mathcal{S}_{1}$ are reducible (by Lemma~\ref%
{all-components-U**-reducible-lem}) and disjoint, we can apply Reduction
Rule~\ref{data-reduction-1-redrule} to the graph~$G$ with respect to the
sets of~$\mathcal{S}_{1}$, by replacing in the interval representation $%
\mathcal{I}$ the intervals $\{I_{v}:v\in S\}$ with the interval $I_{S}=%
\mathbf{span}(S)$ with weight $|S|$, for every $S\in \mathcal{S}_{1}$.
Denote the resulting weighted graph by $G^{\#}=(V^{\#},E^{\#})$ and its
interval representation by $\mathcal{I}^{\#}$. Furthermore denote by $\sigma
^{\#}$ the right-endpoint ordering of~$\mathcal{I}^{\#}$. Then, the next
corollary follows immediately by Theorem~\ref{first-data-reduction-thm}.

\begin{corollary}
\label{G-diesi-correctness-cor}The maximum number of vertices of a path in~$%
G $ is equal to the maximum weight of a path in~$G^{\#}$.
\end{corollary}

Define $A=V(G^{\#})\setminus V(G)$, i.e.,~each vertex $v\in A$ corresponds
to an interval $I_{v}=\mathbf{span}(S)$ in the interval representation $%
\mathcal{I}^{\#}$ , where $S\in \mathcal{S}_{1}$. Recall that for every $%
S\in \mathcal{S}_{1}$, we have that $S=V(C_{i,x}^{t})$, where $i\in \lbrack
k+1]$, $x\in \lbrack p_{i}]$, and $t\in \lbrack q(i,x)]$. Thus, in the
remainder of this section we denote $A=\{v_{i,x}^{t}:~i\in \lbrack k+1],\
x\in \lbrack p_{i}],\ t\in \lbrack q(i,x)]\}$. Furthermore, for every $%
v_{i,x}^{t}\in A$ we denote for simplicity the corresponding interval in the
representation $\mathcal{I}^{\#}$ by $I_{i,x}^{t}=I_{v_{i,x}^{t}}$. Recall
that $V(C_{i,x}^{1})<_{\sigma }V(C_{i,x}^{2})<_{\sigma }\dots <_{\sigma
}V(C_{i,x}^{q(i,x)})$ in the graph $G$, and thus $v_{i,x}^{1}<_{\sigma
^{\#}}v_{i,x}^{2}<_{\sigma ^{\#}}\ldots <_{\sigma ^{\#}}v_{i,x}^{q(i,x)}$ in
the graph $G^{\#}$. Furthermore, since $\mathbf{span}(S)\cap \mathbf{span}%
(S^{\prime })=\emptyset $ for any two distinct sets $S,S^{\prime }\in 
\mathcal{S}_{1}$ by Lemma~\ref{span-S-1-disjoint-lem}, the next corollary
follows immediately by Lemma~\ref{lem:first:data:reduction}.

\begin{corollary}
\label{first-reduction-G-diesi-properties-cor}The set $D$ remains a proper
interval deletion set of the weighted interval $G^{\#}=(V^{\#},E^{\#})$.
Furthermore the set $V^{\#}\setminus D$ can be partitioned into the sets $%
A=V(G^{\#})\setminus V(G)$ and $U^{\#}=V^{\#}\setminus (D\cup A)$, such that:

\begin{enumerate}
\item $A$ is an independent set of~$G^{\#}$, and

\item for every $v_{i,x}^{t}\in A$ and every $u\in V^{\#}\setminus
\{v_{i,x}^{t}\}$, we have $I_{u}\nsubseteq I_{i,x}^{t}$.
\end{enumerate}
\end{corollary}

In the next theorem we prove that the weighted interval graph $G^{\#}$ and
the vertex subset $A$ (cf.~Corollary~\ref%
{first-reduction-G-diesi-properties-cor}) can be computed in $O(n)$ time.

\begin{theorem}
\label{first-reduction-computation-thm}Let $G=(V,E)$ be an interval graph,
where $|V|=n$. Let $D$ be a proper interval deletion set of~$G$, where $%
D=\{d_{0},d_{1},\ldots ,d_{k},d_{k+1}\}$. Then the graph $%
G^{\#}=(V^{\#},E^{\#})$ and the independent set $A\subseteq V^{\#}$ can be
computed in~$O(n)$ time.
\end{theorem}

\begin{proof}
Denote by $\mathcal{I}$ the given interval representation of~$G$. Recall
that the endpoints of the intervals in~$\mathcal{I}$ are given sorted
increasingly, e.g.,~in a linked list $M$. The sets $L$ and $R$ can be
computed in~$O(k)$ time, since they store the left and the right endpoints
of the intervals of the set $D$, where $|D|=k$. Furthermore, the set $%
U=V\setminus D$ can be clearly computed in~$O(n)$ time.

The set $U^{\ast }$ and the sets $L_{i}$, $i\in \lbrack k+1]$, can be
efficiently computed as follows. First we visit all endpoints of the
intervals in~$\mathcal{I}$ (in increasing order). For every endpoint $p\in
\{l_{u},r_{u}:u\in U\}\cup \{l_{d_{1}},l_{d_{2}},\ldots ,l_{d_{k+1}}\}$ that
we visit, such that $r_{d_{i}}<p<r_{d_{i+1}}$, where $r_{d_{i}},r_{d_{i+1}}%
\in D$, we add to $p$ the label $label(p)=d_{i}$. Initially, we set $%
L_{i}=\{r_{d_{i-1}},r_{d_{i}}\}$, for every $i\in \lbrack k+1]$. Then we
iterate for every $p\in \{l_{d_{1}},l_{d_{2}},\ldots ,l_{d_{k+1}}\}$. If $%
label(p)=d_{i}$ then we add $p$ to the set $L_{i}$. Note that, during this
computation, we can store the elements of each set $L_{i}$ in increasing
order, using a linked list. Furthermore, in the same time we add to every
element $p$ of the set $L_{i}$ a pointer to the position of~$p$ in the
linked list $M$, which keeps the endpoints of the intervals in~$\mathcal{I}$
in increasing order. To compute the set $U^{\ast }$ we iterate for every $%
u\in U$ and we compare $label(l_{u})$ with $label(r_{u})$. If $%
label(l_{u})=label(r_{u})$ then we set $u\in U^{\ast }$, otherwise we set $%
u\notin U^{\ast }$. Similarly, during this computation we can store the
elements of the set $U^{\ast }$ in increasing order (according to the order
of their right endpoints in the linked list $M$). Since the endpoints in~$%
\mathcal{I}$ are assumed to be sorted in increasing order, the computation
of~$U^{\ast }$ and of all sets~$L_{i}$, $i\in \lbrack k+1]$, can be done in
total in~$O(n)$ time.

Note that there are in total $\sum_{i=1}^{k+1}p_{i}=2(k+1)$ different sets $%
U_{i,x}^{\ast }$, where $i\in \lbrack k+1]$ and $x\in \lbrack p_{i}]$ (cf.
Lemma~\ref{sum-pi-lem}). All these sets $U_{i,x}^{\ast }$ can be efficiently
computed as follows, using the sets $U^{\ast }$ and $L_{1},L_{2},\ldots
,L_{k+1}$, which we have already computed, as follows. For every $i\in
\lbrack k+1]$ we iterate over the points $\{l_{i,0},l_{i,1},\dots
,l_{i,p_{i}}\}$ of the set~$L_{i}$ in increasing order. For every $x\in
\lbrack p_{i}]$ we visit sequentially in the linked list $M$ (which keeps
the endpoints of~$\mathcal{I}$ in increasing order) from the point $%
l_{i,x-1} $ until the point $l_{i,x}$. For every point $p$ between $%
l_{i,x-1} $ and $l_{i,x}$, we check whether $p=r_{u}$ and $u\in U^{\ast }$,
and if this is the case then we add vertex $u$ to the set~$U_{i,x}^{\ast }$.
This check on point $p$ can be done on $O(1)$ time, e.g.,~by checking
whether $label(l_{u})=label(r_{u})$, as we described above. Thus, as we scan
once through the linked list $M$, the computation of all sets~$U_{i,x}^{\ast
}$ can be done in~$O(n)$ time in total. Note that, during this computation
we can store the elements of each of the sets $U_{i,x}^{\ast }$ in
increasing order (according to the order of their right endpoints in the
linked list $M$).

The computation of the sets $U_{i,x}^{\ast \ast }$, where $i\in \lbrack k+1]$
and $x\in \lbrack p_{i}]$, can be done efficiently as follows. First, we
scan once through each of the $O(k)$ sets $U_{i,x}^{\ast }\neq \emptyset $
and, for every such set, we compute its rightmost right endpoint $p_{i,x}$
in the interval representation $\mathcal{I}$, i.e.,~$p_{i,x}=\max
\{r_{u}:u\in U_{i,x}^{\ast }\}$. Furthermore define $p_{i,0}=r_{d_{d-1}}$,
for every $i\in \lbrack k+1]$. The computation of all such points $p_{i,x}$
can be done in~$O(n)$ time in total. Now, for every fixed $i\in \lbrack k+1]$%
, the computation of the sets $U_{i,1}^{\ast \ast },U_{i,2}^{\ast \ast
},\ldots ,U_{i,p_{i}}^{\ast \ast }$ can be done as follows. Initially, set $%
p=p_{i,0}$. We iterate for every $x=1,2,\ldots ,p_{i}$ (in this order). For
every such $x$, we first define $p=\max \{p,p_{i,x-1}\}$. Then, for every
vertex $u\in U_{i,x}^{\ast }$ we check whether $p<l_{u}$, and if it is the
case then we add vertex $u$ to the set $U_{i,x}^{\ast \ast }$. It can be
easily checked that this process eventually correctly computes the sets $%
U_{i,x}^{\ast \ast }$ (cf.~the definition of~$U_{i,x}^{\ast \ast }$). Since
every two sets $U_{i,x}^{\ast },U_{i^{\prime },x^{\prime }}^{\ast }$ are
disjoint, all the sets $U_{i,x}^{\ast \ast }$, where $i\in \lbrack k+1]$ and 
$x\in \lbrack p_{i}]$, can be done by these computations in~$O(n)$ time in
total. Moreover, in the same time we can store the elements of each of the
sets $U_{i,x}^{\ast \ast }$ in increasing order (according to the order of
their right endpoints in the linked list $M$).

Let now $i\in \lbrack k+1]$ and $x\in \lbrack p_{i}]$. The computation of
the connected components $C_{i,x}^{1}$, $C_{i,x}^{2},\dots ,C_{i,x}^{q(i,x)}$
of~$U_{i,x}^{\ast \ast }$ can be done efficiently as follows. We visit all
vertices of~$U_{i,x}^{\ast \ast }$ in increasing order. For every such
vertex $u$, we compare its left endpoint $l_{u}$ with the right endpoint $%
r_{u^{\prime }}$ of the previous vertex $u^{\prime }$ of~$U_{i,x}^{\ast \ast
}$. If $l_{u}<r_{u^{\prime }}$ then $u$ belongs to the same connected
component $C_{i,x}^{t}$ where $u^{\prime }$ belongs, otherwise $u$ belongs
to the next connected component $C_{i,x}^{t+1}$. The correctness of this
computation of~$C_{i,x}^{1}$, $C_{i,x}^{2},\dots ,C_{i,x}^{q(i,x)}$ follows
from the fact that the interval representation $\mathcal{I}[U_{i,x}^{\ast }]$
of~$G[U_{i,x}^{\ast }]$ is proper by Lemma~\ref{lem:proper.repres} and by
Observation~\ref{obs:consecutive.joined.by.edge}. Since we visit every
vertex of each $U_{i,x}^{\ast \ast }$ once, where $i\in \lbrack k+1]$ and $%
x\in \lbrack p_{i}]$, the computation of all connected components $%
C_{i,x}^{t}$, where $i\in \lbrack k+1]$, $x\in \lbrack p_{i}]$, and $t\in
\lbrack q(i,x)]$, can be done in~$O(n)$ time in total. Note that, in the
same time we can also compute the interval $\mathbf{span}(C_{i,x}^{t})$ for
every such component $C_{i,x}^{t}$, by just keeping track of the left
endpoint $l_{u}$ (resp.~right endpoint $r_{u^{\prime }}$) of the leftmost
vertex $u$ (resp.~of the rightmost vertex~$u^{\prime }$) in~$C_{i,x}^{t}$.

Summarizing, we can compute in~$O(n)$ time the family $\mathcal{S}_{1}$\
that contains all vertex sets $S=V(C_{i,x}^{t})$, where $i\in \lbrack k+1]$, 
$x\in \lbrack p_{i}]$, and $t\in \lbrack q(i,x)]$, cf.~Eq.~(\ref%
{S-1-family-eq}). Moreover, in the same time we can also compute the
intervals $I_{S}=\mathbf{span}(S)$, where $S\in \mathcal{S}_{1}$. Then the
interval representation $\mathcal{I}^{\#}$ can be computed from $\mathcal{I}$
in~$O(n)$ time, by replacing for every $S\in \mathcal{S}_{1}$ the intervals $%
\{I_{v}:v\in S\}$ with the interval $I_{S}=\mathbf{span}(S)$. Note that
these intervals $\{\mathbf{span}(S):S\in \mathcal{S}_{1}\}$ are exactly the
intervals of the vertices in the independent set $A$ of~$G^{\#}$. Therefore,
the sets $A\subseteq V^{\#}$ and $U^{\#}=V^{\#}\setminus (D\cup A)$ can be
computed in~$O(n)$ time.
\end{proof}

\subsection{The graph $\widehat{G}$\label{graph-G-hat-subsec}}

Consider the weighted interval graph $G^{\#}=(V^{\#},E^{\#})$ with the
interval representation $\mathcal{I}^{\#}$ and the right-endpoint ordering $%
\sigma ^{\#}$ that we constructed in Section~\ref{graph-G-diesi-subsec}.
Recall that, by Corollary~\ref{first-reduction-G-diesi-properties-cor}, $%
D=\{d_{0},d_{1},\ldots ,d_{k},d_{k+1}\}\subseteq V^{\#}$ is a proper
interval deletion set of~$G^{\#}$ and that the vertices of~$V^{\#}\setminus D
$ are partitioned into the independent set $A=\{v_{i,x}^{t}:~i\in \lbrack
k+1],\ x\in \lbrack p_{i}],\ t\in \lbrack q(i,x)]\}$ and the set $%
U^{\#}=V^{\#}\setminus (D\cup A)$. Recall that $v_{i,x}^{1}<_{\sigma
^{\#}}v_{i,x}^{2}<_{\sigma ^{\#}}\ldots <_{\sigma ^{\#}}v_{i,x}^{q(i,x)}$.
For every vertex $v_{i,x}^{t}\in A$, the interval of~$v_{i,x}^{t}$ in the
representation $\mathcal{I}^{\#}$ is denoted by $I_{i,x}^{t}$. 
We define the set $T$ of endpoints in the representation $\mathcal{I}^{\#}$
as

\begin{equation*}
T=R\cup L\cup \bigcup_{i,x}\left\{
l_{I_{i,x}^{1}},r_{I_{i,x}^{1}},l_{I_{i,x}^{2}},r_{I_{i,x}^{2}},l_{I_{i,x}^{q(i,x)-1}},r_{I_{i,x}^{q(i,x)-1}},l_{I_{i,x}^{q(i,x)}},r_{I_{i,x}^{q(i,x)}}\right\} .
\end{equation*}%
Note that $|T|\leq |R|+|L|+8\sum_{i=1}^{k+1}p_{i}$, and thus Lemma~\ref%
{sum-pi-lem} implies that $|T|\leq 18k+16$. We denote $T=\{t_{1},t_{2},\dots
,t_{|T|}\}$, where $t_{1}<t_{2}<\dots <t_{|T|}$. For every $1\leq j\leq
i\leq |T|$ we define 
\begin{equation}
U_{ji}=\{u\in U^{\#}:t_{j-1}<l_{u}<t_{j}\text{ and }t_{i-1}<r_{u}<t_{i}\}.
\label{Uji-eq}
\end{equation}%
Note that $\{U_{ji}:1\leq j\leq i\leq |T|\}$ provides a partition of~$U^{\#}$%
. As the next lemma shows, it suffices to consider in the following only the
sets $U_{ji}$ such that $j\neq i$.

\begin{lemma}
\label{lem:uii.is.empty} For every $i\in [k+1]$, $U_{ii}=\emptyset$.
\end{lemma}

\begin{proof}
Let $u\in U_{ii}$, for some $i\in |T|$. Since $U_{ii}\subseteq
U^{\#}=V^{\#}\setminus (D\cup A)$, note that vertex $u$ exists also in the
original (unweighted) interval graph $G$. Furthermore, since $I_{u}$ is
strictly contained between two consecutive points of~$T$, it is also
strictly contained between two consecutive points of~$R\cup L\subseteq T$.
Therefore, for some $i\in \lbrack k+1]$ and $x\in \lbrack p_{i}]$, $u\in
U_{i,x}^{\ast \ast }$ by Observation~\ref%
{U-star-star-between-consecutive-R-or-L-obs}. However, all vertices of~$%
\bigcup_{i,x}U_{i,x}^{\ast \ast }$ in the initial interval graph $G$ have
been replaced by the vertex set $A$ in the weighted interval graph $G^{\#}$.
This is a contradiction, since $u\in U^{\#}=V^{\#}\setminus (D\cup A)$. Thus 
$U_{ii}=\emptyset $.
\end{proof}

\medskip

We are now ready to define the family $\mathcal{S}_{2}$ of vertex subsets of 
$U^{\#}$ as follows:%
\begin{equation}
\mathcal{S}_{2}=\{U_{ji}:1\leq j<i\leq |T|\}.  \label{S-2-family-eq}
\end{equation}

\begin{lemma}
\label{Uji-weakly-reducible-lem}Every set $S\in \mathcal{S}_{2}$ is weakly
reducible in the graph $G^{\#}=(V^{\#},E^{\#})$.
\end{lemma}

\begin{proof}
Consider a set $S\in \mathcal{S}_{2}$. Then $S=U_{ji}$, for some $1\leq
j<i\leq |T|$. In the first part of the proof we show by contradiction that $%
\mathcal{I}^{\#}[U_{ji}]$ is a proper interval representation of~$%
G^{\#}[U_{ji}]$. Assume otherwise that there exist two vertices $v,u\in
U_{ji}$ such that $I_{v}\subseteq I_{u}$. Since the intervals for the
vertices of~$U_{ji}$ are the same in both interval representations $\mathcal{%
I}$ and $\mathcal{I}^{\#}$, it follows that $I_{v}\subseteq I_{u}$ in the
representation $\mathcal{I}$ of the initial (unweighted) interval graph $G$.
Then, from the preprocessing of Theorem~\ref%
{interval-representation-preprocessing-thm} there exist two vertices $%
z,z^{\prime }\in V(G)$ such that $z<_{\sigma }v<_{\sigma }z^{\prime }$,
where $z,z^{\prime }\in N(u)\setminus N(v)$ in the graph $G$.

First suppose that $z,z^{\prime }\notin D$. Then the vertices $\{z,z^{\prime
},u,v\}$ induce a $K_{1,3}$ in~$G\setminus D$, which is a contradiction to
the assumption that $D$ is a proper interval deletion set of~$G$~\cite%
{Roberts69}.

Now suppose that $z\in D$. Then $r_{z}\in R\subseteq T$. If $r_{z}<l_{u}$
then $zu\notin E(G)$, which is a contradiction. Thus $l_{u}<r_{z}$.
Moreover, since $zv\notin E(G)$ and $z<_{\sigma }v$, it follows that $%
r_{z}<l_{v}$. That is, $l_{u}<r_{z}<l_{v}$, where $r_{z}\in T$. This is a
contradiction to the assumption that both $v$ and $u$ belong to the same set 
$U_{ji}$.

Finally suppose that $z^{\prime }\in D$. Then $l_{z^{\prime }}\in L\subseteq
T$. If $r_{u}<l_{z^{\prime }}$ then $z^{\prime }u\notin E(G)$, which is a
contradiction. Thus $l_{z^{\prime }}<r_{u}$. Moreover, since $z^{\prime
}v\notin E(G)$ and $v<_{\sigma }z^{\prime }$, it follows that $%
r_{v}<l_{z^{\prime }}$. That is, $r_{v}<l_{z^{\prime }}<r_{u}$, where $%
l_{z^{\prime }}\in T$. This is a contradiction to the assumption that both $%
v $ and $u$ belong to the same set $U_{ji}$. This proves Condition~1 of
Definition~\ref{def:weakly.reducible.set}.

In the second part of the proof we show by contradiction that for every $%
u\in U_{ji}$ and every $v\in V^{\#}$, if $I_{v}\subseteq I_{u}$, then $%
U_{ji}\subseteq N(v)$ in the graph $G^{\#}$. Let $u\in U_{ji}$ and $v\in
V^{\#}$ such that $I_{v}\subseteq I_{u}$. Assume that there exists a vertex $%
u^{\prime }\in U_{ji}$ such that $u^{\prime }v\notin E^{\#}$. First suppose
that $I_{v}\cap T\neq \emptyset $, i.e.,~$l_{v}\leq t_{0}\leq r_{v}$ for
some $t_{0}\in T$. Let $v<_{\sigma ^{\#}}u^{\prime }$. Then, since $%
u^{\prime }v\notin E^{\#}$ by assumption, it follows that $%
r_{v}<l_{u^{\prime }}$. Furthermore $l_{u}<l_{v}$, since $I_{v}\subseteq
I_{u}$. Therefore, $l_{u}<l_{v}\leq t_{0}\leq r_{v}<l_{u^{\prime }}$, i.e.,~$%
l_{u}<t_{0}<l_{u^{\prime }}$, where $t_{0}\in T$. This is a contradiction to
the assumption that both $u$ and $u^{\prime }$ belong to the same set $%
U_{ji} $. Let $u^{\prime }<_{\sigma ^{\#}}v$. Then, since $u^{\prime
}v\notin E^{\#} $ by assumption, it follows that $r_{u^{\prime }}<l_{v}$.
Furthermore $r_{v}<r_{u} $, since $I_{v}\subseteq I_{u}$. Therefore, $%
r_{u^{\prime }}<l_{v}\leq t_{0}\leq r_{v}<r_{u}$, i.e.,~$r_{u^{\prime
}}<t_{0}<r_{u}$, where $t_{0}\in T $. This is a contradiction to the
assumption that both $u$ and $u^{\prime }$ belong to the same set $U_{ji}$.

Now suppose that $I_{v}\cap T=\emptyset $. If $v\in D$ then both its
endpoints belong to $R\cup T\subseteq T$, and thus $I_{v}\cap T\neq
\emptyset $, which is a contradiction. Thus $v\in V^{\#}\setminus D=A\cup
U^{\#}$. Let $v\in U^{\#}$, i.e.,~$v\in U_{j^{\prime }i^{\prime }}$, for
some $1\leq j^{\prime }\leq i^{\prime }\leq |T|$. Note that $j^{\prime }\neq
i^{\prime }$ by Lemma~\ref{lem:uii.is.empty}, and thus $j^{\prime
}<i^{\prime }$. Thus, it follows by equation (\ref{Uji-eq}) that $%
l_{v}<t_{j^{\prime }}\leq t_{i^{\prime }-1}<r_{v}$, i.e.,~$I_{v}\cap T\neq
\emptyset $, which is a contradiction. Therefore, $v\in U^{\#}$, and thus $%
v\in A$. If $v\in
\bigcup_{i,x}\{v_{i,x}^{1},v_{i,x}^{2},v_{i,x}^{q(i,x)-1},v_{i,x}^{q(i,x)}\}$%
, then both its endpoints belong to $T$ (by the definition of the set $T$),
i.e.,~$I_{v}\cap T\neq \emptyset $, which is a contradiction.

Therefore $v=v_{i,x}^{h}$, for some $i\in \lbrack k+1]$, $x\in \lbrack
p_{i}] $, and $3\leq h\leq q(i,x)-2$. Recall that $v_{i,x}^{1}<_{\sigma
^{\#}}v_{i,x}^{2}<_{\sigma ^{\#}}v_{i,x}^{h}<_{\sigma
^{\#}}v_{i,x}^{q(i,x)-1}<_{\sigma ^{\#}}v_{i,x}^{q(i,x)}$. Furthermore
recall that the intervals $\{I_{z}:z\in U_{i,x}^{\ast \ast }\}$ in the
interval representation $\mathcal{I}$ of the initial graph $G$ have been
replaced by the intervals $\{I_{i,x}^{h}:1\leq h\leq q(i,x)\}$ in the
interval representation $\mathcal{I}^{\#}$ of the graph $G^{\#}$. Suppose
that $l_{u}<l_{v_{i,x}^{1}}$. Then, since $I_{v}=I_{v_{i,x}^{h}}\subseteq
I_{u}$ by assumption, it follows that $I_{u}$ properly contains in the
representation $\mathcal{I}^{\#}$ all three intervals of the vertices $%
v_{i,x}^{1}$, $v_{i,x}^{2}$, and $v=v_{i,x}^{h}$. Therefore the interval $%
I_{u}$ properly contains in the initial representation $\mathcal{I}$ all
triples of intervals $\{I_{z_{1}},I_{z_{2}},I_{z_{h}}\}$, where $z_{1}\in
C_{i,x}^{1}$, $z_{2}\in C_{i,x}^{2}$, and $z_{h}\in C_{i,x}^{h}$. Thus,
since $z_{i},z_{2},z_{3}$ induce an independent set (they belong to
different connected components $C_{i,x}^{1},C_{i,x}^{2},C_{i,x}^{h}$ of~$%
U_{i,x}^{\ast \ast }$ in the initial graph $G$), it follows that the
vertices $\{u,z_{1},z_{2},z_{h}\}$ induce a $K_{1,3}$ in~$G\setminus D$.
This is a contradiction to the assumption that $D$ is a proper interval
deletion set of~$G$~\cite{Roberts69}. Thus $l_{v_{i,x}^{1}}<l_{u}$. Suppose
that $r_{v_{i,x}^{q(i,x)}}<r_{u}$. Then it follows similarly that $I_{u}$
properly contains in the initial representation $\mathcal{I}$ all triples of
intervals $\{I_{z_{h}},I_{z_{q(i,x)-1}},I_{z_{q(i,x)}}\}$, where $z_{h}\in
C_{i,x}^{h}$, $z_{q(i,x)-1}\in C_{i,x}^{q(i,x)-1}$, and $z_{q(i,x)}\in
C_{i,x}^{q(i,x)}$, which is again a contradiction. Therefore, $%
l_{v_{i,x}^{1}}<l_{u}<r_{u}<r_{v_{i,x}^{q(i,x)}}$. That is, the interval $%
I_{u}$ is properly contained in the interval $\mathbf{span}(U_{i,x}^{\ast
\ast })$, and thus $u\in U_{i,x}^{\ast \ast }$ (cf.~the definition of the
sets $U_{i,x}^{\ast }$ and $U_{i,x}^{\ast \ast }$ in Section~\ref%
{graph-G-diesi-subsec}). Therefore, vertex~$u$ has been replaced in the
weighted graph $G^{\#}$ by a vertex of~$A$. This is a contradiction, since $%
u\in U_{ji}\subseteq U^{\#}$ by assumption.

Thus for every $u\in U_{ji}$ and every $v\in V^{\#}$, if $I_{v}\subseteq
I_{u}$, then $U_{ji}\subseteq N(v)$ in the graph $G^{\#}$. This proves
Condition~2 of Definition~\ref{def:weakly.reducible.set}.
\end{proof}

\medskip

Note that for every $1\leq j<i\leq |T|$, the set $U_{ji}$ contains no
vertices of~$D$, since by definition $U_{ji}\subseteq U^{\#}=V^{\#}\setminus
(D\cup A)$. Therefore, since all sets of~$\mathcal{S}_{2}$ are disjoint, we
can apply Reduction Rule~\ref{data-reduction-2-redrule} to the graph $G^{\#}$
with respect to the sets of~$\mathcal{S}_{2}$, by replacing in the interval
representation $\mathcal{I}^{\#}$ the intervals $\{I_{v}:v\in S\}$ with $%
\min \{|S|,|D|+4\}$ copies of the interval $I_{S}=\mathbf{span}(S)$, for
every $S\in \mathcal{S}_{2}$. Denote the resulting weighted graph by $%
\widehat{G}=(\widehat{V},\widehat{E})$ and its interval representation by $%
\widehat{\mathcal{I}}$. Then the next corollary follows immediately by
Theorem~\ref{second-data-reduction-thm}.

\begin{corollary}
\label{G-tilde-correctness-cor}The maximum weight of a path in~$G^{\#}$ is
equal to the maximum weight of a path in~$\widehat{G}$.
\end{corollary}

In the next two theorems we provide the main results of this section. In
particular, in Theorem~\ref{special-weighted-interval-graph-thm} we prove
that the constructed weighted interval graph $\widehat{G}$ is a special
weighted interval graph with a parameter~$\kappa$ that is upper bounded by $%
O(k^{3})$ and in Theorem~\ref{second-reduction-computation-thm} we provide a
time bound of $O(k^{2}n)$ for computing $\widehat{G}=(\widehat{V},\widehat{E}%
)$ and a special vertex partition $\widehat{V}=A\cup B$.

\begin{theorem}
\label{special-weighted-interval-graph-thm}The weighted interval graph $%
\widehat{G}=(\widehat{V},\widehat{E})$ is a \emph{special weighted interval
graph} with \emph{parameter} $\kappa =O(k^{3})$.
\end{theorem}

\begin{proof}
Define $A=V(G^{\#})\setminus V(G)$, i.e.,~$A$ is the set of vertices that
have been introduced in the weighted interval graph $G^{\#}$ by applying
Reduction Rule~\ref{data-reduction-1-redrule} to the initial (unweighted)
interval graph $G$ (cf.~Section~\ref{graph-G-diesi-subsec}). Note that the
vertices of~$A$ also belong to the weighted graph $\widehat{G}$, since they
are not affected by the application of Reduction Rule~\ref%
{data-reduction-2-redrule} to the graph $G^{\#}$. Furthermore, we define the
vertex set $B=\widehat{V}\setminus A$, i.e.,~$\widehat{V}$ is partitioned
into the sets $A$ and $B$.

We will prove that $A$ and $B$ satisfy the three conditions of Definition %
\ref{special-interval-def}. Since the vertices of~$A$ are not affected by
the application of Reduction Rule~\ref{data-reduction-2-redrule}, Corollary~%
\ref{first-reduction-G-diesi-properties-cor} implies that $A$ induces an
independent set in~$\widehat{G}$. This proves Condition~1 of Definition~\ref%
{special-interval-def}.

Let $v_{i,x}^{t}\in A$ and $u\in \widehat{V}\setminus \{v_{i,x}^{t}\}$.
Assume that $I_{u}\subseteq I_{i,x}^{t}$ in the interval representation~$%
\widehat{\mathcal{I}}$. If $u$ is also a vertex of the weighted graph $%
G^{\#} $, then Corollary~\ref{first-reduction-G-diesi-properties-cor}
implies that $I_{u}\nsubseteq I_{i,x}^{t}$ in the interval representation $%
\mathcal{I}^{\#} $ (and thus also in the representation $\widehat{\mathcal{I}%
}$). This is a contradiction to the assumption that $I_{u}\subseteq
I_{i,x}^{t}$ in~$\widehat{\mathcal{I}}$. Otherwise, if $u$ is a vertex of~$%
\widehat{G}$ but not a vertex of~$G^{\#}$, then $I_{u}=\mathbf{span}(S)$,
for some $S\in \mathcal{S}_{2}$. Therefore, for every vertex $u^{\prime }\in
S$, we have that $I_{u^{\prime }}\subseteq \mathbf{span}(S)=I_{u}\subseteq
I_{i,x}^{t}$ in the interval representation $\mathcal{I}^{\#}$ of the graph~$%
G^{\#}$. This is a contradiction by Corollary~\ref%
{first-reduction-G-diesi-properties-cor}. Therefore, for every $%
v_{i,x}^{t}\in A$ and every $u\in \widehat{V}\setminus \{v_{i,x}^{t}\}$, we
have $I_{u}\subseteq I_{i,x}^{t}$ in the interval representation $\widehat{%
\mathcal{I}}$. This proves Condition~2 of Definition~\ref%
{special-interval-def}.

Recall by Corollary~\ref{first-reduction-G-diesi-properties-cor} that the
set $V^{\#}\setminus D$ is partitioned into the sets $A$ and $U^{\#}$.
Furthermore, recall that $\{U_{ji}:1\leq j<i\leq |T|\}$ provides a partition
of~$U^{\#}$, and that each of these vertex sets $U_{ji}$ is replaced in the
graph $\widehat{G}$ by at most $\min \{|U_{ji}|,|D|+4\}$ vertices. Thus the
vertex set $B$ contains all vertices of~$D$ and at most $|D|+4$ vertices for
each of the vertex subsets $\{U_{ji}:1\leq j<i\leq |T|\}$ of~$G^{\#}$.
Recall that $|T|\leq 18k+16$, and thus there exist at most $\binom{18k+16}{2}
$ different sets $U_{ji}$. Furthermore, recall that $D=\{d_{0},d_{1},\ldots
,d_{k},d_{k+1}\}$, i.e.,~$|D|=k+2$. Therefore, $|B|\leq |D|+\binom{18k+16}{2}%
\cdot (|D|+4)=(k+2)+\binom{18k+16}{2}\cdot (k+6)=O(k^{3})$. This proves
Condition 3 of Definition~\ref{special-interval-def} and completes the proof
of the theorem.
\end{proof}

\begin{theorem}
\label{second-reduction-computation-thm}Let $G=(V,E)$ be an interval graph,
where $|V|=n$. Let $D=\{d_{0},d_{1},\ldots ,d_{k},d_{k+1}\}$ be a proper
interval deletion set of~$G$. Then the special weighted interval graph $%
\widehat{G}=(\widehat{V},\widehat{E})$ and a special vertex partition $%
\widehat{V}=A\cup B$ can be computed in~$O(k^{2}n)$ time.
\end{theorem}

\begin{proof}
First recall that the graph $G^{\#}=(V^{\#},E^{\#})$ and the independent set 
$A\subseteq V^{\#}$ can be computed in~$O(n)$ time by Theorem~\ref%
{first-reduction-computation-thm}. Furthermore, recall by the proof of
Theorem~\ref{first-reduction-computation-thm} that, during the computation
of the interval representation $\mathcal{I}^{\#}$ of the graph $G^{\#}$, we
also compute the points of~$R\cup L$ and the intervals $I_{i,x}^{t}=\mathbf{%
span}(C_{i,x}^{t})$ for the connected components $C_{i,x}^{t}$ of~$%
U_{i,x}^{\ast \ast }$, where $i\in \lbrack k+1]$, $p_{i}\in \lbrack p_{i}]$,
and $t\in \lbrack q(i,x)]$. Thus, since all these computations can be done
in~$O(n)$ time, we can also compute the set $T$ of endpoints in the interval
representation $\mathcal{I}^{\#}$ in~$O(n)$ time in total.

Now, for every pair $\{j,i\}$ such that $1\leq j<i\leq |T|$, we can compute
the set $U_{ji}$ in~$O(n)$ time by visiting each vertex $u\in U^{\#}$ once
and by checking whether $t_{j-1}<l_{u}<t_{j}$ and $t_{i-1}<r_{u}<t_{i}$ (cf.
the definition of the sets $U_{ji}$). Furthermore, we can compute in the
same time the interval $\mathbf{span}(U_{ji})$ by keeping the leftmost left
endpoint and the rightmost right endpoint of~$U_{ji}$, respectively. Thus,
since there are $\binom{|T|}{2}\leq \binom{18k+16}{2}=O(k^{2})$ such pairs
of indices $\{j,i\}$, all sets $U_{j,i}$ and all intervals $\mathbf{span}%
(U_{ji})$ can be computed in~$O(k^{2}n)$ time in total.

Once we have computed all intervals $\mathbf{span}(U_{ji})$, we can
iteratively remove from the representation $\mathcal{I}^{\#}$ the intervals
of the vertices of~$U_{ji}$ and replace them with $\min
\{|U_{ji}|,|D|+2\}\leq |U_{ji}|$ copies of the interval $\mathbf{span}%
(U_{ji})$, resulting thus at the interval representation $\widehat{\mathcal{I%
}}$ of~$\widehat{G}$. Since the number of vertices in all these sets $U_{ji}$
is at most $n$, all these replacements can be done in~$O(n)$ time in total.
Finally, since the set~$A$ can be computed in~$O(n)$ time by Theorem~\ref%
{first-reduction-computation-thm}, the set $B=\widehat{V}\setminus A$ can be
also computed in~$O(n)$ time. Summarizing, the interval representation $%
\widehat{\mathcal{I}}$ of~$\widehat{G}$ and the special vertex partition $%
\widehat{V}=A\cup B$ can be computed in~$O(k^{2}n)+O(n)=O(k^{2}n)$ time in
total.
\end{proof}

\medskip

Note here that, although $\widehat{G}=(\widehat{V},\widehat{E})$ is a
special weighted interval graph with a parameter $\kappa$ that depends only
on the size of~$D$ by Theorem~\ref{special-weighted-interval-graph-thm}, $%
\widehat{G}$ may still have $O(n)$ vertices, as the independent set $A$ in
its special vertex partition $\widehat{V}=A\cup B$ may be arbitrarily large.

\section{Parameterized longest path on interval graphs\label{algorithm-sec}}

In this section, we first present Algorithm~\ref%
{weighted-path-special-interval-graph-alg} (cf.~Section~\ref%
{special-algorithm-subsec}) which computes in~$O(\kappa ^{3}n)$ time the
maximum weight of a path in a \emph{special weighted interval graph} with
parameter $\kappa $ (cf.~Definition~\ref{special-interval-def}). Then, using
Algorithm~\ref{weighted-path-special-interval-graph-alg} and the results of
Sections~\ref{data-reductions-sec} and~\ref{special-interval-sec}, we
conclude in Section~\ref{general-algorithm-subsec} with our fixed-parameter
algorithm for \textsc{Longest Path on Interval Graphs}, where the parameter $%
k$ is the size of a minimum proper interval deletion set $D$. Since
Algorithm~\ref{weighted-path-special-interval-graph-alg} can be implemented
to run in~$O(\kappa ^{3}n)$ time and $\kappa =O(k^{3})$ by Theorem~\ref%
{special-weighted-interval-graph-thm}, the algorithm of Section~\ref%
{general-algorithm-subsec} runs in~$O(k^{9}n)$ time.

\subsection{The algorithm for special weighted interval graphs\label%
{special-algorithm-subsec}}

Consider a \emph{special weighted interval graph} $G=(V,E)$ with parameter $%
\kappa \in \mathbb{N}$, which is given along with a special interval
representation $\mathcal{I}$ and a special vertex partition $V=A\cup B$.
Recall by Definition~\ref{special-interval-def} that $A$ is an independent
set and that $|B|\leq \kappa $. Let $w:V\rightarrow \mathbb{N}$ be the
vertex weight function of~$G$. Now we add to the set $B$ an isolated dummy
vertex $v_{0}$ such that $v_{0}<_{\sigma }v_{1}$ and $w(v_{0})=0$. Thus,
after the addition of~$v_{0}$ to $G$, we have $|B|\leq \kappa +1$. Note that 
$v_{0}$ is not contained in any maximum-weight path of this augmented graph.
Thus, every maximum-weight path in the augmented graph is also a maximum
weight path in~$G$, and vice versa. In the following we denote this
augmented graph by $G$. Furthermore, denote by $\mathcal{I}$ the augmented
interval representation and by $\sigma =(v_{0},v_{1},v_{2},\dots ,v_{n})$
its right-endpoint ordering. For every vertex $v\in B$, we define%
\begin{equation*}
\xi _{v}=%
\begin{cases}
l_{u} & \text{if }l_{v}\in I_{u}\text{ for some }u\in A, \\ 
l_{v} & \text{otherwise.}%
\end{cases}%
\end{equation*}

\begin{lemma}
\label{lem:xi.is.well.defined} For every vertex $v\in B$, $\xi_{v}$ is
well-defined.
\end{lemma}

\begin{proof}
It is enough to prove that if there exists a vertex $u\in A$ such that $%
l_{v}\in I_{u}$ then $u$ is unique. Let us assume to the contrary that there
exist two distinct vertices $u$ and $u^{\prime }$ in~$A$ such that $l_{v}\in
I_{u}$ and $l_{v}\in I_{u^{\prime }}$. Then $I_{u}\cap I_{u^{\prime }}\neq
\emptyset $ and $uu^{\prime }\in E(G)$. This is a contradiction, since $A$
is an independent set.
\end{proof}

\medskip

Now we define the set $\Xi $ as%
\begin{equation}
\Xi =\{\xi _{v},l_{v}:v\in B\}.  \label{Ksi-eq}
\end{equation}%
Note that $|\Xi |\leq 2|B|\leq 2(\kappa +1)$. Furthermore, let $u,v\in V$,
where $u\in N(v)$ and $u<_{\sigma}v$. We define the vertex 
\begin{equation}
\pi_{u,v}=\max_{\sigma }\{\{u\}\cup \{w\in B\cap N(u):u<_{\sigma }w<_{\sigma
}v\}\}  \label{pi-eq}
\end{equation}%
Note that, by definition, if $\pi _{u,v}\neq u$ then $\pi_{u,v}\in B$.
Furthermore, due to the condition that $u\in N(v)$ in the definition of the
vertex $\pi _{u,v}$, it follows that $u\in B$ or $v\in B$, since $A$ is an
independent set. That is, vertex~$\pi_{u,v}$ is defined for at most $%
2(\kappa +1)(n+1)=O(\kappa n)$ pairs of vertices $u,v$.

\begin{definition}
\label{subgraphs-interval-definition} Let $\xi \in \Xi $ and $i\in \lbrack
n] $ such that $\xi <r_{v_{i}}$. We define the induced subgraph ${%
G_{\xi}(v_{i})=G[\{v\in V:\xi \leq l_{v}<r_{v}\leq r_{v_{i}}\}]}$ of~$G$
which contains all vertices whose intervals (in the representation $\mathcal{%
I}$ of~$G$) are entirely contained between the points $\xi $ and $r_{v_{i}}$.
\end{definition}

Note by Definition~\ref{subgraphs-interval-definition} that, if $%
l_{v_{i}}<\xi $, then the vertex $v_{i}$ does not belong to the subgraph $%
G_{\xi }(v_{i})$.

\begin{notation}
\label{notation-interval-paths} Let $\xi \in \Xi $ and $i\in \lbrack n]$
such that $\xi <r_{v_{i}}$. Furthermore, let $y\in V(G_{\xi }(v_{i}))$ such
that ${y\in N(v_{i})}$. We denote by $P_{\xi }(v_{i},y)$ a \emph{maximum
weight normal path} of~$G_{\xi }(v_{i})$, among those normal paths whose
last vertex is $y$. For every path $P_{\xi }(v_{i},y)$, we denote its weight 
$w(P_{\xi }(v_{i},y))$ by $W_{\xi }(v_{i},y)$.
\end{notation}

Before we present Algorithm~\ref{weighted-path-special-interval-graph-alg},
we first present some auxiliary technical lemmas (cf.~Lemmas~\ref%
{lem:vi.not.in.p}-\ref{lem:algorithm.second-direction-1}) that will be
useful in the proof of correctness and the running time analysis of the
algorithm (cf.~Theorems~\ref{special-algorithm-correcntess-thm} and~\ref%
{special-algorithm-running-time-thm}, respectively).

\begin{lemma}
\label{lem:vi.not.in.p} Let $\xi \in \Xi $ and $i\in \lbrack n]$, where $\xi
<r_{v_{i}}$, and let $y\in V(G_{\xi }(v_{i}))$ and $y\in N(v_{i})$. Let also 
$P$ be a normal path of~$G_{\xi }(v_{i})$ that has $y$ as its last vertex.
If $v_{i}\notin V(P)$, then $P$ is a path of~$G_{\xi }(\pi _{y,v_{i}})$.
\end{lemma}

\begin{proof}
Let $y^{\prime }$ denote the rightmost vertex of~$P$ (in the ordering $%
\sigma $). Since $v_{i}\notin V(P)$, we have that $y^{\prime }\neq v_{i}$.
Note also that either $y=y^{\prime }$ or $y<_{\sigma }y^{\prime }$. We claim
that $y^{\prime }\leq _{\sigma }\pi _{y,v_{i}}$. Notice that the statement
trivially holds if $y=y^{\prime }$. Thus, it is enough to prove that $%
y^{\prime }\leq _{\sigma }\pi _{y,v_{i}}$ when $y<_{\sigma }y^{\prime }$.
First, as $y<_{\sigma }y^{\prime }$ and $y^{\prime }<_{P}y$, Lemma~\ref%
{lem:oxi.anapoda} implies that $y^{\prime }y\in E(G)$. Furthermore, since $P$
is normal, Lemma~\ref{lem:propers.consecutive.in.path} implies that $%
\mathcal{I}[\{y,y^{\prime }\}]$ does not induce a proper interval
representation. Therefore, since $y<_{\sigma }y^{\prime }$, it follows that $%
I_{y}\subseteq I_{y^{\prime }}$, and thus $y^{\prime }\in B$. Hence, since
also $y^{\prime }<_{\sigma }v_{i}$ and $y^{\prime }y\in E(G)$, it follows
that $y^{\prime }\leq _{\sigma }\pi _{x,v_{i}}$. Therefore $P$ is a path of~$%
G_{\xi }(\pi _{y,v_{i}})$.
\end{proof}

\begin{lemma}
\label{extra-computation-lem}Let $\xi \in \Xi $ and $i\in \lbrack n]$, where 
$\xi <r_{v_{i}}$, and let $y\in V(G_{\xi }(v_{i}))$ and $y\in N(v_{i})$. If $%
l_{y}<l_{v_{i}}$ or $v_{i}\notin V(G_{\xi }(v_{i}))$, then $W_{\xi
}(v_{i},y)=W_{\xi }(\pi _{y,v_{i}},y)$.
\end{lemma}

\begin{proof}
Let $l_{y}<l_{v_{i}}$ or $v_{i}\notin V(G_{\xi }(v_{i}))$. First we prove
that $v_{i}\notin V(P_{\xi }(v_{i},y))$. If $v_{i}\notin V(G_{\xi }(v_{i}))$%
, then clearly $v_{i}\notin V(P_{\xi }(v_{i},y))$. Suppose now that $%
v_{i}\in V(G_{\xi }(v_{i}))$ and that $l_{y}<l_{v_{i}}$. For the sake of
contradiction, assume that $v_{i}\in V(P_{\xi }(v_{i},y))$. Since $%
l_{y}<l_{v_{i}}$ and $r_{y}<r_{v_{i}}$, note that $I_{y}\nsubseteq I_{v_{i}}$
and $I_{v_{i}}\nsubseteq I_{y}$. Thus $\mathcal{I}[\{y,v_{i}\}]$ induces a
proper interval representation. Therefore, since $y<_{\sigma }v_{i}$ and $%
y,v_{i}\in V(P_{\xi }(v_{i},y))$ by assumption, Lemma~\ref%
{lem:propers.consecutive.in.path} implies that $y<_{P}v_{i}$. This is a
contradiction to the assumption that $y$ is the last vertex of~$P_{\xi
}(v_{i},y)$. Therefore $v_{i}\notin V(P_{\xi }(v_{i},y))$.

Thus Lemma~\ref{lem:vi.not.in.p} implies that $V(P_{\xi }(v_{i},y))$ is a
path of $G_{\xi }(\pi _{y,v_{i}})$. Therefore, since $G_{\xi }(\pi
_{y,v_{i}})$ is a subgraph of $G_{\xi }(v_{i}) $ (cf.~Eq.~(\ref{pi-eq})), it
follows that $W_{\xi }(v_{i},y)=W_{\xi }(\pi _{y,v_{i}},y)$.
\end{proof}

\begin{lemma}
\label{lem:algorithm.first} Let $\xi \in \Xi $ and $i\in \lbrack n]$, where $%
\xi <r_{v_{i}}$, and let $v_{i}\in V(G_{\xi }(v_{i}))$. Then $P_{\xi
}(v_{i},v_{i})=(P_{1},v_{i})$, where 
\begin{equation}
w(P_{1})=\max \{W_{\xi }(\pi _{x,v_{i}},x):x\in
V(G_{\xi}(v_{i})),l_{v_{i}}<r_{x}<r_{v_{i}}\}  \label{w-P1-eq-1}
\end{equation}
\end{lemma}

\begin{proof}
Let $P=(P_{1},v_{i})$ be a normal path of~$G_{\xi }(v_{i})$ such that $%
w(P)=W_{\xi }(v_{i},v_{i})$. Denote by $x$ the last vertex of~$P_{1}$. Then,
since $P$ is a normal path of~$G_{\xi }(v_{i})$, $P_{1}$ is a normal path of~%
$G_{\xi }(v_{i})$ that does not contain~$v_{i}$. Lemma~\ref{lem:vi.not.in.p}
implies that $P_{1}$ is a normal path of~$G_{\xi }(\pi _{x,v_{i}})$ that has 
$x$ as its last vertex and hence 
\begin{equation*}
w(P_{1})\leq W_{\xi }(\pi _{x,v_{i}},x).
\end{equation*}

We will now prove that $w(P_{1})=W_{\xi }(\pi _{x,v_{i}},x)$. For this,
assume towards a contradiction that $w(P_{1})<W_{\xi }(\pi _{x,v_{i}},x)$.
Recall that $P_{\xi }(\pi _{x,v_{i}},x)$ is a normal path of~$G_{\xi }(\pi
_{x,v_{i}})\subseteq G_{\xi }(v_{i})$ that has $x$ as its last vertex. Since 
$xv_{i}\in E(G)$, this implies that $(P_{\xi }(\pi _{x,v_{i}},x),v_{i})$ is
a path of~$G_{\xi }(v_{i})$ that has $v_{i}$ as its last vertex.
Furthermore, since $v_{i}$ is the rightmost vertex of the path, it follows
that $(P_{\xi }(\pi _{x,v_{i}},x),v_{i})$ is normal (Observation~\ref%
{obs:appending.a.vertex}). Moreover, 
\begin{equation*}
w(P)=w(P_{1})+w(v_{i})<W_{\xi }(\pi _{x,v_{i}},x)+w(v_{i})=w((P_{\xi }(\pi
_{x,v_{i}},x),v_{i})),
\end{equation*}%
a contradiction to the assumption that $w(P)=W_{\xi }(v_{i},v_{i})$. Hence, 
\begin{equation*}
w(P_{1})=W_{\xi }(\pi _{x,v_{i}},x).
\end{equation*}

To conclude, $P_{\xi }(v_{i},v_{i})=(P_{1},v_{i})$, where 
\begin{equation*}
w(P_{1})=\max \{W_{\xi }(\pi _{x,v_{i}},x) : x\in V(G_{\xi
}(v_{i})),l_{v_{i}}\leq r_{x}\leq r_{v_{i}}\}
\end{equation*}%
and this completes the proof of the lemma.
\end{proof}

\begin{lemma}
\label{lem:algorithm.second-direction-2} Let $\xi \in \Xi $ and $i\in
\lbrack n]$, where $\xi <r_{v_{i}}$, and let $v_{i},y\in V(G_{\xi }(v_{i}))$
and $y\in N(v_{i})$. Let $\zeta \in \{l_{y}\}\cup \{\xi \in \Xi
:l_{v_{i}}<\xi <l_{y}\}$ and $x\in V(G_{\xi }(v_{i}))$ be such that $%
l_{v_{i}}<r_{x}<\zeta $. Furthermore, let $P_{1}$ be a normal path of~$%
G_{\xi }(\pi _{x,v_{i}})$ with $x$ as its last vertex and $P_{2}$ be a
normal path of~$G_{\zeta }(\pi _{y,v_{i}})$ with $y$ as its last vertex.
Then $P=(P_{1},v_{i},P_{2})$ is a normal path of~$G_{\xi }(v_{i})$ with $y$
as its last vertex.
\end{lemma}

\begin{proof}
Since $V(P_{2})\subseteq V(G_{\zeta }(\pi _{y,v_{i}}))=\{v\in V(G) : \zeta
\leq l_{v}\leq r_{v}\leq r_{\pi _{y,v_{i}}}\}$ it follows that $\zeta \leq
l_{v}$ for every vertex $v\in V(P_{2})$. Therefore, since $r_{x}<\zeta $, it
follows that 
\begin{equation}
x<_{\sigma }v,\text{ for every }v\in V(P_{2}).  \label{eq:eqinequal}
\end{equation}

Therefore, since $P_{1}$ is normal, $x$ is the last vertex of~$P_{1}$, and $%
x<_{\sigma }v$, $xv\notin E(G)$ for every $v\in V(P_{2})$, it follows that $%
V(P_{1})\cap V(P_{2})=\emptyset $ (Lemma~\ref{lem:oxi.anapoda}). Moreover,
since $l_{v_{i}}<r_{x}$, $xv_{i}\in E(G)$ and since $\zeta \leq l_{v}\leq
r_{v}\leq r_{\pi _{y,v_{i}}}\leq r_{v_{i}}$ it follows that $v_{i}v\in E(G)$
for every vertex in~$V(P_{2})$. Therefore, $(P_{1},v_{i},P_{2})$ is a path
that has $y$ as its last vertex (as $y$ is the last vertex of~$P_{2}$).
Moreover, since $V(P_{1})\subseteq V(G_{\xi }(\pi _{x,v_{i}}))\subseteq
V(G_{\xi }(v_{i}))$, $V(P_{2})\subseteq V(G_{\zeta }(\pi
_{y,v_{i}}))\subseteq V(G_{\xi }(\pi _{x,v_{i}}))$ and $v_{i}\in V(G_{\xi
}(v_{i}))$, $P$ is a path of~$G_{\xi }(v_{i})$. It remains to show that $P$
is normal.

We first show that if $v_{1}$ is the first vertex of~$P_{1}$, then $%
v_{1}<_{\sigma }v$ for every vertex $v\in V(P)\setminus \{v_{1}\}$. Notice
that $v_{1}<_{\sigma }v$, for every vertex $v\in V(P_{1})\setminus \{v\}$,
since $P_{1}$ is a normal path and $v_{1}$ is its first vertex. Recall also
that $x<_{\sigma }v$, for every vertex $v\in P_{2}\cup \{v_{i}\}$ (equation~(%
\ref{eq:eqinequal})). As $v_{1}<_{\sigma }x<_{\sigma }v$ for every vertex $%
v\in P_{2}\cup \{v_{i}\}$, it indeed follows that $v_{1}<_{\sigma }v$, for
every vertex in~$V(P)\setminus \{v_{1}\}$.

We now show that for every vertex $v\in V(P)$, with successor $v^{\prime
}\in V(P)$ and every vertex $u\in V(P)$ such that $v^{\prime }<_{P}u$, and $%
vu\in E(G)$ it holds that $v^{\prime }<_{\sigma }u$. Let us assume to the
contrary that for some $v\in V(P)$, with successor $v^{\prime }\in V(P)$
there exists a vertex $u\in V(P)$ such that $v^{\prime }<_{P}u$, $vu\in E(G)$%
, and $u<_{\sigma }v^{\prime }$. Notice that if $\{v,v^{\prime
},u\}\subseteq V(P_{1})$ or $\{v,v^{\prime },u\}\subseteq V(P_{2})$, then we
obtain a contradiction to the assumptions that $P_{1}$ and $P_{2}$ are
normal paths. Similarly, if $v=v_{i}$ we obtain a contradiction to the fact
that $P_{2}$ is a normal path since the successor of~$v_{i}$ in~$P$ is the
first vertex of~$P_{2}$. Moreover, as the only neighbor of~$x$ in~$%
V(P)\setminus V(P_{1})$ is $v_{i}$ we obtain that $v\in V(P_{1})\setminus
\{x\}$ and $u\in V(P_{2})$. Notice then that since $vu\in E(G)$, it holds
that $r_{v}>l_{u}\geq \zeta >r_{x}$, and since $x<_{\sigma }v$ and $v<_{P}x$%
, as $P_{1}$ is normal from Lemma~\ref{lem:oxi.anapoda}, $xv\in E(G)$.
However, then $x<_{\sigma }u<_{\sigma }v^{\prime }$, a contradiction to the
assumption that $P_{1}$ is normal. Therefore, we conclude that for every
vertex $v\in V(P)$, with successor $v^{\prime }\in V(P)$ and every vertex $%
u\in V(P)$ such that $v^{\prime }<_{P}u$, and $vu\in E(G)$ it holds that $%
v^{\prime }<_{\sigma }u$. Thus, we completed the proof that $P$ is a normal
path of~$G_{\xi }(v_{i})$ that has $y$ as its last vertex.
\end{proof}

\begin{lemma}
\label{lem:algorithm.second-direction-1} Let $\xi \in \Xi $ and $i\in
\lbrack n]$, where $\xi <r_{v_{i}}$, and let $v_{i},y\in V(G_{\xi }(v_{i}))$
and $y\in N(v_{i})$. Let $P_{\xi }(v_{i},y)=(P_{1},v_{i},P_{2})$. If $%
P_{2}\neq (y)$, then there exists some $\zeta \in \Xi $, where $%
l_{v_{i}}<\zeta \leq l_{y}$, such that%
\begin{eqnarray}
w(P_{1}) &=&\max \{W_{\xi }(\pi _{x,v_{i}},x):x\in
V(G_{\xi}(v_{i})),l_{v_{i}}<r_{x}<\zeta \},  \label{w-P1-eq-2} \\
w(P_{2}) &=&W_{\zeta }(\pi _{y,v_{i}},y).  \label{w-P2-eq}
\end{eqnarray}%
Otherwise, if $P_{2}=(y)$ then $l_{v_{i}}<l_{y}$ and 
\begin{equation}
w(P_{1})=\max \{W_{\xi }(\pi _{x,v_{i}},x):x\in
V(G_{\xi}(v_{i})),l_{v_{i}}<r_{x}<l_{y}\}.  \label{w-P1-eq-3}
\end{equation}
\end{lemma}

\begin{proof}
Denote $P=P_{\xi }(v_{i},y)$. Notice first that, since $v_{i}$ is the
rightmost vertex of~$P$, the path $P_{1}$ is not empty by Observation~\ref%
{obs:first.vertex}. Denote by $x$ the last vertex of~$P_{1}$. Then, since $%
P_{1}$ is the prefix of the normal path $P$, observe that $P_{1}$ is a
normal path of~$G_{\xi }(v_{i})$ that has $x$ as its last vertex and does
not contain~$v_{i}$. Furthermore, $P_{1}$ is a path of~$G_{\xi }(\pi
_{x,v_{i}})$ by Lemma~\ref{lem:vi.not.in.p}. Therefore $P_{1}$ is a normal
path of~$G_{\xi }(\pi _{x,v_{i}})$, and thus, 
\begin{equation*}
w(P_{1})\leq W_{\xi }(\pi _{x,v_{i}},x).
\end{equation*}

Let $v\in V(P_{2})$. Since $v_{i}<_{P}v$ and $v<_{\sigma }v_{i}$, Lemma~\ref%
{lem:oxi.anapoda} implies that $v_{i}v\in E(G)$. Therefore, again since $%
v_{i}<_{P}v$ and $v<_{\sigma }v_{i}$, Lemma~\ref%
{lem:propers.consecutive.in.path} implies that $\mathcal{I}[v,v_{i}]$ does
not induce a proper interval representation, i.e.,~either $I_{v}\subseteq
I_{v_{i}}$ or $I_{v_{i}}\subseteq I_{v}$. Thus, since $v<_{\sigma }v_{i}$,
it follows that $I_{v}\subseteq I_{v_{i}}$. That is, $I_{v}\subseteq
I_{v_{i}}$ for every $v\in V(P_{2})$.

Therefore, since $y\in V(P_{2})$, it follows that $I_{y}\subseteq I_{v_{i}}$%
, and thus in particular $l_{v_{i}}<l_{y}$. Now let 
\begin{equation*}
\zeta =%
\begin{cases}
\min \{l_{v}\in \Xi :~v\in V(P_{2})\} & \text{if }P_{2}\neq (y), \\ 
l_{y} & \text{otherwise.}%
\end{cases}%
\end{equation*}%
We show that, if $P_{2}\neq (y)$, then $\{l_{v}\in \Xi :~v\in V(P_{2})\}\neq
\emptyset $, and thus $\zeta $ is well-defined. Notice first that, if $y\in
B $, then $l_{y}\in \Xi $. Let $y\in A$, then let $y^{\prime }$ be the
neighbor of~$y$ in~$P_{2}$ (note that $y^{\prime }$ always exists since $%
P_{2}\neq (y) $). Then, as $A$ is an independent set, it follows that $%
y^{\prime }\in B$, and thus $l_{y^{\prime }}\in B$. Hence, if $P_{2}\neq (y)$%
, then in any case $\{l_{v}\in \Xi :~v\in V(P_{2})\}\neq \emptyset $.

Suppose that $\zeta <r_{x}$. Then, by definition of~$\zeta $, there exists
some $v\in V(P_{2})$ such that $\zeta =l_{v}<r_{x}$. Let $r_{x}<r_{v}$,
i.e.,~ $l_{v}<r_{x}<r_{v}$. Then $xv\in E(G)$. Therefore, since $v<_{\sigma
}v_{i}$ for every $v\in V(P_{2})$, it follows by the normality of~$P$ that $%
v_{i}$ is not the next vertex of~$x$ in~$P$, which is a contradiction. Let $%
r_{v}<r_{x}$, i.e.,~$v<_{\sigma }x$. Then, since $x<_{P}v$, Lemma~\ref%
{lem:oxi.anapoda} implies that $xv\in E(G)$, which is again a contradiction
by the normality of~$P$. Therefore $r_{x}<\zeta $. Now note that $%
l_{v_{i}}<r_{x}$, since $xv_{i}\in E(G)$. That is, $l_{v_{i}}<r_{x}<\zeta $.
Therefore, since $\zeta \leq l_{y}$ by the definition of~$\zeta $, it
follows that 
\begin{equation*}
l_{v_{i}}<r_{x}<\zeta \leq l_{y}.
\end{equation*}

Let now $P_{2}\neq (y)$. We prove that $\zeta \leq l_{v}$, for every $v\in
V(P_{2})$. Assume otherwise that there exists a vertex $v\in V(P_{2})$ such
that $l_{v}<\zeta $. Then, by the definition of~$\zeta $, it follows that $%
l_{v}\notin \Xi $. Therefore $v\notin B$, and thus $v\in A$ (cf.~the
definition of the set $\Xi $). Since $P_{2}\neq (y)$, it follows that $v$
has at least one neighbor~$u$ in~$P_{2}$. Then $u\in B$, since $A$ is an
independent set. Furthermore, $l_{v}<\zeta \leq l_{u}$. Therefore, since $%
uv\in E(G)$ by assumption, it follows that $l_{u}\in I_{v}$. That is, $\xi
_{u}=l_{v}\in \Xi $ (cf.~the definition of~$\xi _{u}$ for a vertex $u\in B$%
). Therefore $\zeta \leq \xi _{u}=l_{v}$, which is a contradiction to our
assumption. Thus $\zeta \leq l_{v}$, for every $v\in V(P_{2})$.

Now recall that $I_{v}\subseteq I_{v_{i}}$ for every $v\in V(P_{2})$, as we
proved above, and thus $v_{i}v\in E(G)$ for every $v\in V(P_{2})$.
Furthermore, recall that all vertices of~$P_{2}$ appear in~$P$ after vertex $%
v_{i}$. Therefore, since $P$ is a normal path by assumption, it follows that 
$P_{2}$ is also a normal path.

Thus, since $\zeta \leq l_{v}$ for every $v\in V(P_{2})$, as we proved
above, it follows that $P_{2}$ is a normal path of~$G_{\zeta }(v_{i})$ that
does not contain~$v_{i}$. Therefore Lemma~\ref{lem:vi.not.in.p} implies that 
$P_{2}$ is a path of~$G_{\zeta }(\pi _{y,v_{i}})$. Thus, since $y$ is the
last vertex of~$P_{2}$, it follows that%
\begin{equation*}
w(P_{2})\leq W_{\zeta }(\pi _{y,v_{i}},y).
\end{equation*}

In the remainder of the proof we show that $w(P_{1})=W_{\xi }(\pi
_{x,v_{i}},x)$ and $w(P_{2})=W_{\zeta }(\pi _{y,v_{i}},y)$. Towards a
contradiction assume that at least one of the equalities does not hold.
Notice first that from Lemma~\ref{lem:algorithm.second-direction-2}, $%
P=(P_{\xi }(\pi _{x,v_{i}},x),v_{i},P_{\zeta }(\pi _{y,v_{i}},y))$ is a
normal path of~$G_{\xi }(v_{i})$ that has $y$ as its last vertex. Notice now
that 
\begin{equation*}
W_{\xi }(v_{i},y)=w(P_{1})+w(v_{i})+w(P_{2})<W_{\xi }(\pi
_{x,v_{i}},x)+w(v_{i})+W_{\zeta }(\pi _{y,v_{i}},y),
\end{equation*}%
a contradiction. Therefore, 
\begin{equation*}
w(P_{1})=W_{\xi }(\pi _{x,v_{i}},x)
\end{equation*}%
and 
\begin{equation*}
w(P_{2})=W_{\zeta }(\pi _{y,v_{i}},y).
\end{equation*}%
Summarizing, if $P_{2}\neq (y)$ we obtain that 
\begin{eqnarray*}
w(P_{1}) &=&\max \{W_{\xi }(\pi _{x,v_{i}},x) : x\in V(G_{\xi
}(v_{i})),l_{v_{i}}<r_{x}<\zeta \} , \\
w(P_{2}) &=&W_{\zeta }(\pi _{y,v_{i}},y),
\end{eqnarray*}%
and if $P_{2}=(y)$ we obtain that 
\begin{equation*}
w(P_{1})=\max \{W_{\xi }(\pi _{x,v_{i}},x) : x\in V(G_{\xi
}(v_{i})),l_{v_{i}}<r_{x}<l_{y}\}.
\end{equation*}%
This completes the proof of the lemma.
\end{proof}

\medskip

We are now ready to present Algorithm~\ref%
{weighted-path-special-interval-graph-alg}, which computes the maximum
weight of a path in a given \emph{special} \emph{weighted} interval graph $G$%
. It is easy to check that Algorithm~\ref%
{weighted-path-special-interval-graph-alg} can be slightly modified such
that it returns the actual path $P$ instead of its weight only.

First we give a brief overview of the algorithm. Using dynamic programming,
it computes a 3-dimensional table. In this table, for every point $\xi \in
\Xi $, every index $i\in \lbrack n]$, and every vertex $y\in V(G_{\xi
}(v_{i}))$, where $\xi <r_{v_{i}}$ and $y\in N(v_{i})$, the entry $W_{\xi
}(v_{i},y)$ (resp.~the entry~$W_{\xi }(v_{i},v_{i})$) keeps the weight of a
normal path in the subgraph $G_{\xi }(v_{i})$ which is the largest among
those normal paths whose last vertex is $y$ (resp.~$v_{i}$). Thus, since $%
w(v_{0})=0$ for the dummy isolated vertex $v_{0}$ (cf.~line~\ref{alg-line-1}
of the algorithm), the maximum weight of a path in $G$ will be eventually
stored in one of the entries $\left\{ W_{l_{v_{0}}}(v_{i},v_{i}):1\leq i\leq
n\right\} $ or in one of the entries $\left\{ W_{l_{v_{0}}}(v_{i},y):1\leq
i\leq n, \ y<_{\sigma}v_{i}, \ y\in N(v_{i})\right\} $, depending on whether
the last vertex $y$ of the desired maximum-weight path coincides with the
rightmost vertex $v_{i}$ of this path in the ordering $\sigma $ (cf.~line~%
\ref{alg-line-18} of the algorithm).

Note that for every computed entry $W_{\xi }(v_{i},y)$ the vertices $v_{i}$
and $y$ are adjacent, and thus $v_{i}\in B$ or $y\in B$, since $A$ is an
independent set. Thus, since $|B|=O(\kappa )$, there are at most $O(\kappa
n) $ such eligible pairs of vertices $v_{i},y$. Furthermore, since also $%
|\Xi |=O(\kappa )$, the computed 3-dimensional table stores at most $%
O(\kappa ^{2}n)$ entries $W_{\xi }(v_{i},v_{i})$ and $W_{\xi }(v_{i},y)$.
From the \emph{for}-loops of lines~\ref{alg-line-2}-\ref{alg-line-3} of the
algorithm and from the obvious inductive hypothesis we may assume that
during the $\{i,\xi \}$th iteration all previous values $W_{\xi ^{\prime
}}(v_{i^{\prime}},v_{i^{\prime}})$ and $W_{\xi ^{\prime
}}(v_{i^{\prime}},y^{\prime })$, where $i^{\prime }<i$ or $\xi ^{\prime
}<\xi $, have been correctly computed at a previous iteration.

In the initialization phase for a particular pair $\{i,\xi \}$ (cf.~lines~%
\ref{alg-line-4}-\ref{alg-line-6}) the algorithm computes some initial
values for $W_{\xi }(v_{i},v_{i})$ and $W_{\xi }(v_{i},y)$. For a path with $%
v_{i}$ as its last vertex, we are only interested in the case where $%
v_{i}\in V(G_{\xi }(v_{i}))$; in this case we initialize $W_{\xi
}(v_{i},v_{i})=w(v_{i})$, cf.~line~\ref{alg-line-4}. For a path with $y\neq
v_{i}$ as its last vertex (cf.~lines~\ref{alg-line-5}-\ref{alg-line-6}), we
initialize $W_{\xi }(v_{i},y)=W_{\xi }(\pi _{y,v_{i}},y)$, since the path $%
P_{\xi }(\pi _{y,v_{i}},y)$ is indeed a normal path of the graph $G_{\xi
}(\pi _{y,v_{i}})$, which is an induced subgraph of $G_{\xi }(v_{i})$.

For the induction step phase (cf.~lines~\ref{alg-line-7}-\ref{alg-line-17})
the algorithm updates the initialized entries $W_{\xi }(v_{i},v_{i})$ and $%
W_{\xi }(v_{i},y)$ according to Lemmas~\ref{extra-computation-lem}-\ref%
{lem:algorithm.second-direction-1}. To update the value $W_{\xi
}(v_{i},v_{i})$ we only need to consider the case where $v_{i}\in V(G_{\xi
}(v_{i}))$; in this case $W_{\xi }(v_{i},v_{i})$ is updated in lines~\ref%
{alg-line-7}-\ref{alg-line-9} according to Lemma~\ref{lem:algorithm.first}.
The values of $W_{\xi }(v_{i},y)$, where $y \neq v_{i}$, are updated in
lines~\ref{alg-line-10}-\ref{alg-line-17}. In particular, in the case where $%
l_{y}<l_{v_{i}}$ or $v_{i}\notin V(G_{\xi }(v_{i}))$, the value of $W_{\xi
}(v_{i},y)$ is updated in lines~\ref{alg-line-11}-\ref{alg-line-12}
according to Lemma~\ref{extra-computation-lem}. Otherwise, $W_{\xi
}(v_{i},y) $ is updated in lines~\ref{alg-line-14}-\ref{alg-line-17}
according to Lemma~\ref{lem:algorithm.second-direction-1}.

The correctness of the algorithm is proved in Theorem~\ref%
{special-algorithm-correcntess-thm} and its running time is proved in
Theorem~\ref{special-algorithm-running-time-thm}.

\begin{algorithm}[htb]
\caption{Computing a maximum-weight path of a special weighted interval graph} \label{weighted-path-special-interval-graph-alg}
\begin{algorithmic}[1]
\REQUIRE{A special weighted interval graph $G=(V,E)$ with parameter $\kappa\in \mathbb{N}$, along with the special interval representation $\mathcal{I}$ of~$G$ and the partition $V=A\cup B$, 
where $\sigma = (v_1, v_2, \ldots, v_n)$ is a right-endpoint ordering of~$V$.}
\ENSURE{The maximum weight of a path in~$G$}

\medskip

\STATE{Add an isolated dummy vertex $v_{0}$ with $w(v_{0})=0$ to set $B$, where $v_{0} <_{\sigma} v_{1}$; denote $\sigma=(v_{0},v_{1},v_{2},\ldots,v_{n})$}\label{alg-line-1} 

\medskip%

\FOR{$i=0$ to $n$} \label{alg-line-2}
     \FOR{every $\xi \in \Xi$ where $\xi < r_{v_{i}}$} \label{alg-line-3}

\medskip
          
          \STATE{\textbf{if} $v_{i}\in V(G_{\xi}(v_{i}))$ \textbf{then} $W_{\xi}(v_{i},v_{i}) \leftarrow w(v_{i})$} \label{alg-line-4}
          \COMMENT{initialization}
     
\vspace{0,1cm}
     
          \FOR{every $y\in V(G_{\xi}(v_{i}))$ where $y\in N(v_{i})$} \label{alg-line-5}
               \STATE{$W_{\xi}(v_{i},y) \leftarrow W_{\xi}(\pi_{y,v_{i}},y)$} \label{alg-line-6}
               \COMMENT{initialization}
          \ENDFOR

\medskip

          \IF{$v_{i} \in V(G_{\xi}(v_{i}))$} \label{alg-line-7}
               \STATE{$W_{1} \leftarrow \max \{W_{\xi }(\pi _{x,v_{i}},x) :  x\in V(G_{\xi}(v_{i})),l_{v_{i}}<r_{x}<r_{v_{i}}\}$} \label{alg-line-8}
               \STATE{$W_{\xi}(v_{i},v_{i}) \leftarrow \max\{W_{\xi}(v_{i},v_{i}), W_{1} + w(v_{i})\}$} \label{alg-line-9}
          \ENDIF

\vspace{0.1cm}
          
          \FOR{every $y\in V(G_{\xi}(v_{i}))$ where $y\in N(v_{i})$} \label{alg-line-10}
               \IF{$l_{y} < l_{v_{i}}$ or $v_{i} \notin V(G_{\xi}(v_{i}))$} \label{alg-line-11}
                    \STATE{$W_{\xi}(v_{i},y) \leftarrow W_{\xi}(\pi_{y,v_{i}},y)$} \label{alg-line-12}
               \ELSE \label{alg-line-13}
                    \STATE{$W^{\prime}_{1} \leftarrow \max \{W_{\xi }(\pi _{x,v_{i}},x) :  x\in V(G_{\xi}(v_{i})),l_{v_{i}}<r_{x}<l_{y} \}$} \label{alg-line-14}
                    \FOR{every $\zeta \in \Xi$ with $l_{v_{i}} < \zeta \leq l_{y}$} \label{alg-line-15}
                         \STATE{$W_{1} \leftarrow \max \{W_{\xi }(\pi _{x,v_{i}},x) : x\in V(G_{\xi}(v_{i})),l_{v_{i}}<r_{x}<\zeta \}$} \label{alg-line-16}
                         \STATE{$W_{\xi}(v_{i},y) \leftarrow \max\{W_{\xi}(v_{i},y), W^{\prime}_{1} + w(v_{i}) + w(y), W_{1} + w(v_{i}) + W_{\zeta }(\pi _{y,v_{i}},y)\}$} \label{alg-line-17}
                    \ENDFOR
               \ENDIF
          \ENDFOR
     \ENDFOR
\ENDFOR

\medskip

\RETURN{$\max\{W_{l_{v_{0}}}(v_{i},v_{i}), W_{l_{v_{0}}}(v_{i},y) : 1\leq i \leq n, \ y<_{\sigma} v_{i}, \ y\in N(v_{i}) \}$} \label{alg-line-18}
\end{algorithmic}
\end{algorithm}

\begin{theorem}
\label{special-algorithm-correcntess-thm}Let $G=(V,E)$ be a special weighted
interval graph, given along with a special interval representation $\mathcal{%
I}$ and a special vertex partition $V=A\cup B$. Then Algorithm~\ref%
{weighted-path-special-interval-graph-alg} computes the maximum weight of a
path $P$ in~$G$.
\end{theorem}

\begin{proof}
In lines~\ref{alg-line-2}-\ref{alg-line-17}, Algorithm~\ref%
{weighted-path-special-interval-graph-alg} iterates for every $i\in
\{0,1,2,\ldots ,n\}$ and for every $\xi \in \Xi $ such that $\xi <r_{v_{i}}$%
. For every such $i$ and $\xi $, the algorithm computes the values $W_{\xi
}(v_{i},v_{i})$ and the values $W_{\xi }(v_{i},y)$, for every vertex $y\in
V(G_{\xi }(v_{i}))$ such that $y\in N(v_{i})$. We will prove by induction on 
$i$ that these values are the weights of the maximum-weight normal paths $%
P_{\xi }(v_{i},v_{i})$ and the values $P_{\xi }(v_{i},y)$, respectively (cf.
Notation~\ref{notation-interval-paths}).

For the induction basis, let $i=0$. Then, since $v_{0}$ is an isolated
vertex (cf.~line~\ref{alg-line-1} of Algorithm~\ref%
{weighted-path-special-interval-graph-alg}), the only $\xi \in \Xi $, for
which $\xi <r_{v_{0}}$, is $\xi =l_{v_{0}}$. Then line~\ref{alg-line-4} of
the algorithm is executed and the algorithm correctly computes the value $%
W_{\xi }(v_{0},v_{0})=w(v_{0})=0$. Furthermore, since $v_{0}$ is a dummy
vertex by assumption, the lines~\ref{alg-line-5}-\ref{alg-line-6} and the
lines~\ref{alg-line-10}-\ref{alg-line-17} of the algorithm are not executed
at all for $i=0$. Finally, in lines~\ref{alg-line-7}-\ref{alg-line-9} the
algorithm recomputes the value $W_{\xi }(v_{0},v_{0})=w(v_{0})=0$, since
there exists no vertex $x$ such that $l_{v_{i}}<r_{x}<r_{v_{i}}$ (cf.~line %
\ref{alg-line-8} of the algorithm). This value of~$W_{\xi }(v_{0},v_{0})$ is
clearly correct. This completes the induction basis.

For the induction step, let $i\geq 1$. Consider the iteration of the
algorithm for any $\xi \in \Xi $, where $\xi <r_{v_{i}}$. First the
algorithm initializes in lines~\ref{alg-line-4}-\ref{alg-line-6} the values $%
W_{\xi }(v_{i},v_{i})$ and the values $W_{\xi }(v_{i},y)$, for every vertex $%
y\in V(G_{\xi }(v_{i}))$ such that $y\in N(v_{i})$. The initialization of
line~\ref{alg-line-4} is correct, since the single-vertex path $P=(v_{i})$
is clearly a normal path of the graph $G_{\xi }(v_{i})$ which has $v_{i}$ as
its last vertex. The initialization of lines~\ref{alg-line-5}-\ref%
{alg-line-6} is correct, since the path $P_{\xi }(\pi _{y,v_{i}},y)$ is
indeed a normal path of~$G_{\xi }(\pi _{y,v_{i}})$, which is an induced
subgraph of~$G_{\xi }(v_{i})$ (cf.~Definition~\ref%
{subgraphs-interval-definition}).

In lines~\ref{alg-line-7}-\ref{alg-line-9} the algorithm updates the current
(initialized) value of~$W_{\xi }(v_{i},v_{i})$. The correctness of this
update follows directly by Lemma~\ref{lem:algorithm.first}. Furthermore, in
lines~\ref{alg-line-10}-\ref{alg-line-17} the algorithm iterates for every
vertex $y\in V(G_{\xi }(v_{i}))$ such that $y\in N(v_{i})$. For every such
value of~$y$, the algorithm updates the current (initialized) value of~$%
W_{\xi }(v_{i},y)$.

The correctness of the update in line~\ref{alg-line-12} follows directly by
Lemma~\ref{extra-computation-lem}.

During the execution of lines~\ref{alg-line-14}-\ref{alg-line-17} the
algorithm deals with the case where $v_{i}\in V(G_{\xi }(v_{i},y))$ and $%
l_{v_{i}}<l_{y}$. If $v_{i}$ does not belong to the desired path $P_{\xi
}(v_{i},y)$, then by Lemma~\ref{lem:vi.not.in.p} $P_{\xi }(v_{i},y)$ is also
a normal path of~$G_{\xi }(\pi _{y,v_{i}})$, which is an induced subgraph of 
$G_{\xi }(v_{i})$. Therefore, in this case, $W_{\xi }(v_{i},y)=W_{\xi }(\pi
_{y,v_{i}},y)$. The algorithm does not update the current value of~$W_{\xi
}(v_{i},y)$, since $W_{\xi }(v_{i},y)$ has been initialized to $W_{\xi }(\pi
_{y,v_{i}},y)$ in line~\ref{alg-line-6}.

For the remainder of the proof, assume that $v_{i}$ belongs to the desired
path $P_{\xi }(v_{i},y)$, i.e.,~$P_{\xi }(v_{i},y)=(P_{1},v_{i},P_{2})$, for
some sub-paths $P_{1}$ and $P_{2}$ of~$P_{\xi }(v_{i},y)$. In lines~\ref%
{alg-line-14}-\ref{alg-line-16} the algorithm distinguishes between the
cases where $P_{2}=(y)$ and $P_{2}\neq (y)$. To deal with the case where $%
P_{2}=(y)$, i.e.,~with the case where $P_{\xi }(v_{i},y)=(P_{1},v_{i},y)$,
the algorithm computes in line~\ref{alg-line-14} the value $W_{1}^{\prime
}=\max \{W_{\xi }(\pi _{x,v_{i}},x):x\in V(G_{\xi
}(v_{i})),l_{v_{i}}<r_{x}<l_{y}\}$ of the desired path $P_{1}$ (cf.~Eq.\ (%
\ref{w-P1-eq-3}) of Lemma~\ref{lem:algorithm.second-direction-1}). Then it
compares in line~\ref{alg-line-17} the current value of~$W_{\xi }(v_{i},y)$
with the value $W_{1}^{\prime }+w(v_{i})+w(y)$, and it stores the greatest
value between them in~$W_{\xi }(v_{i},y)$. This update is correct by Eq.~(%
\ref{w-P1-eq-3}) of Lemma~\ref{lem:algorithm.second-direction-1}.

To deal with the case where $P_{2}\neq (y)$, the algorithm iterates in lines %
\ref{alg-line-15}-\ref{alg-line-16} for every $\zeta \in \Xi $ such that $%
l_{v_{i}}<\zeta \leq l_{y}$. For every such value of~$\zeta $ it computes
the value $W_{1}$ of the desired path $P_{1}$ (cf.~Eq.~(\ref{w-P1-eq-2}) of
Lemma~\ref{lem:algorithm.second-direction-1}). Then the algorithm compares
in line~\ref{alg-line-17} the current value of~$W_{\xi }(v_{i},y)$ with the
value $W_{1}+w(v_{i})+W_{\zeta }(\pi _{y,v_{i}},y)$ and it stores the
greatest between them in~$W_{\xi }(v_{i},y)$. For every $\zeta \in \Xi $,
where $l_{v_{i}}<\zeta \leq l_{y}$, Lemma~\ref%
{lem:algorithm.second-direction-2} implies that the path $(P_{\xi }(\pi
_{x,v_{i}},x),v_{i},P_{\zeta }(\pi _{y,v_{i}},y))$ is a normal path of~$%
G_{\xi }(v_{i})$ with $y$ as its last vertex. Therefore $W_{1}+w(v_{i})+W_{%
\zeta }(\pi _{y,v_{i}},y)\leq W_{\xi }(v_{i},y)$, for every such value of~$%
\zeta $. Furthermore Lemma~\ref{lem:algorithm.second-direction-1} implies
that there exists at least one such value $\zeta $, such that the values $%
W_{1}=\max \{W_{\xi }(\pi _{x,v_{i}},x):x\in V(G_{\xi
}(v_{i})),l_{v_{i}}<r_{x}<\zeta \}$ and $W_{\zeta }(\pi _{y,v_{i}},y)$ are
equal to the weights of the sub-paths $P_{1}$ and $P_{2}$ of~$P_{\xi
}(v_{i},y)$, respectively. Therefore, these updates of~$W_{\xi }(v_{i},y)$
for all values of~$\zeta $ are correct. This completes the induction step.

Therefore, after the execution of lines~\ref{alg-line-2}-\ref{alg-line-17},
Algorithm~\ref{weighted-path-special-interval-graph-alg} has correctly
computed all values $W_{\xi }(v_{i},v_{i})$ and $W_{\xi }(v_{i},y)$, where $%
i\in \{0,1,2,\ldots ,n\}$, $\xi \in \Xi $ such that $\xi <r_{v_{i}}$, and $%
y\in V(G_{\xi }(v_{i}))$ such that $y\in N(v_{i})$. Thus, since for every $i$
and every $\xi $ the graph $G_{\xi }(v_{i})$ is an induced subgraph of~$%
G_{l_{v_{0}}}(v_{i})$, it follows that the maximum weight of a path in~$G$
is one of the values $W_{l_{v_{0}}}(v_{i},v_{i})$ and $%
W_{l_{v_{0}}}(v_{i},y) $. Therefore the algorithm returns the correct value
in line~\ref{alg-line-18}.
\end{proof}

\begin{theorem}
\label{special-algorithm-running-time-thm}Let $G=(V,E)$ be a special
weighted interval graph with $n$ vertices and parameter $\kappa $. Then
Algorithm~\ref{weighted-path-special-interval-graph-alg} can be implemented
to run in~$O(\kappa ^{3}n)$ time.
\end{theorem}

\begin{proof}
Since $G$ is a special weighted interval graph with $V=A\cup B$ as its
special vertex partition (cf.~Definition~\ref{special-interval-def}), $A$ is
an independent set and $|B|\leq \kappa +1$ (after the addition of the dummy
vertex $v_{0}$ to the set~$B$). Recall that the endpoints of the intervals
in~$\mathcal{I}$ are given sorted increasingly, e.g.,~in a linked list $M$.
The points $\{\xi _{v}:v\in B\}$ can be efficiently as follows. First we
visit all endpoints of the intervals in~$\mathcal{I}$ (in increasing order).
For every endpoint $l_{v}$, where $v\in B$, which we visit between the
endpoints $l_{u}$ and $r_{u}$, for some $u\in A$, we define $\xi _{v}=l_{u}$%
. Thus the points $\{\xi _{v}:v\in B\}$ can be computed in~$O(n)$ time.
Furthermore the points $\{l_{v}:v\in B\}$ can be computed in~$O(\kappa )$
time by just enumerating all vertices of~$B$. Therefore the set $\Xi $ can
be computed in~$O(\kappa +n)=O(n)$ time in total. Furthermore recall that
there are $O(\kappa n)$ different vertices $\pi _{u,v}$ (cf.~Eq.~(\ref{pi-eq}%
)) and note that, given two adjacent vertices $u,v$, we can compute the
vertex $\pi _{u,v}$ in~$O(\kappa )$ time by enumerating in worst case all
vertices of~$B$. Thus, all vertices $\pi _{u,v}$ can be computed in total $%
O(\kappa ^{2}n)$ time.

Now we provide an upper bound on the number of values $W_{\xi }(v_{i},v_{i})$
and $W_{\xi }(v_{i},y)$ that are computed by Algorithm~\ref%
{weighted-path-special-interval-graph-alg}. Since $\xi \in B$ and $v_{i}\in
V $, there are in total at most $O(\kappa n)$ different values $W_{\xi
}(v_{i},v_{i})$. Furthermore, the values $W_{\xi }(v_{i},y)$, where $y\neq
v_{i}$, are computed for every $i\in \{0,1,2,\ldots ,n\}$, every $\xi \in
\Xi $ such that $\xi <r_{v_{i}}$, and every $y\in V(G_{\xi }(v_{i}))$ such
that $y\in N(v_{i})$. Thus, due to the condition that $y\in N(v_{i})$, it
follows that $v_{i}\in B$ or $y\in B$. That is, there are in total at most $%
2(\kappa +1)(n+1)=O(\kappa n)$ pairs of vertices $v_{i},y$ for which we
compute the values $W_{\xi }(v_{i},y)$. Therefore, since $\xi \in B$ and $%
|B|\leq \kappa +1$, there are at most $O(\kappa ^{2}n)$ different values $%
W_{\xi }(v_{i},y)$. Summarizing, Algorithm~\ref%
{weighted-path-special-interval-graph-alg} computes at most $O(\kappa ^{2}n)$
different values $W_{\xi }(v_{i},v_{i})$ and $W_{\xi }(v_{i},y)$.

We now show that all computations performed in the lines~\ref{alg-line-8},~%
\ref{alg-line-14}, and~\ref{alg-line-16} of the algorithm can be implemented
to run in total $O(\kappa ^{2}n)$ time. Denote by $Q=\{l_{v},r_{v}:v\in V\}$
the set of all endpoints of the intervals in~$\mathcal{I}$. Recall that the
points of~$Q$ are assumed to be already sorted increasingly. Note that, in
order to perform all computations of the lines~\ref{alg-line-8},~\ref%
{alg-line-14}, and~\ref{alg-line-16}, it suffices to store at each point $%
q\in Q$ the values 
\begin{equation}
\omega _{\xi }(q,v_{i})=\max \{W_{\xi }(\pi _{x,v_{i}},x):x\in V(G_{\xi
}(v_{i})),\ l_{v_{i}}<r_{x}<q\}  \label{omega-eq-1}
\end{equation}%
for every $\xi \in \Xi $ and every $i\in \lbrack n]$ such that $\xi \leq
l_{v_{i}}<q\leq r_{v_{i}}$. Indeed, once we have computed all possible
values $\omega _{\xi }(q,v_{i})$, lines~\ref{alg-line-8},~\ref{alg-line-14},
and~\ref{alg-line-16} of the algorithm can be executed in $O(1)$ time by
just accessing the stored values $\omega _{\xi }(r_{v_{i}},v_{i})$, $\omega
_{\xi }(l_{y},v_{i})$, and $\omega _{\xi }(\zeta ,v_{i})$, respectively.
Observe that for every point $q\in Q$ such that $l_{v_{i}}<q\leq r_{v_{i}}$,
the vertex which has $q$ as an endpoint is adjacent to vertex $v_{i}$. Thus,
since there are at most $O(\kappa n)$ pairs of adjacent vertices in~$G$, it
follows that there are $O(\kappa n)$ such pairs of a point $q\in Q$ and a
vertex $v_{i}\in V$.

Given a point $\xi \in \Xi $ and a vertex $v_{i}$, we can compute all values 
$\omega _{\xi }(q,v_{i})$ in~$O(|N(v_{i})|)$ time as follows. Let $%
q_{0}>l_{v_{i}}$ be the first endpoint after $l_{v_{i}}$ in the ordering of
the endpoints in~$Q$. As there does not exist any vertex $x$ such that $%
l_{v_{i}}<r_{x}<q_{0}$ (cf.~Eq.~(\ref{omega-eq-1})), we store at point $%
q_{0} $ the value $\omega _{\xi }(q_{0},v_{i})=0$. Then, we visit in
increasing order all points of~$q\in Q$ between $q_{0}$ and $r_{v_{i}}$.
Note that we can visit all these vertices in $O(|N(v_{i})|)$ time as the
points of~$Q$ are already sorted increasingly. Let $q\in Q$ be the currently
visited point between $q_{0}$ and $r_{v_{i}}$, and let $q^{\prime }$ be the
predecessor of~$q$ in the ordering of~$Q$. Then it follows by the definition
of $\omega _{\xi }(q,v_{i})$ in Eq.~(\ref{omega-eq-1}) that 
\begin{equation}
\omega _{\xi }(q,v_{i})=\max \{\omega _{\xi }(q^{\prime },v_{i}),\ W_{\xi
}(\pi _{x,v_{i}},x):x\in V(G_{\xi }(v_{i})),\ q^{\prime }\leq r_{x}<q\}.
\label{omega-eq-2}
\end{equation}%
Therefore, since $q^{\prime }$ and $q$ are two consecutive points of $Q$
between $q_{0}$ and $r_{v_{i}}$, the value $\omega _{\xi }(q,v_{i})$ can be
computed in $O(1)$ time using the value of $\omega _{\xi }(q^{\prime},v_{i})$%
, as follows:%
\begin{equation}
\omega _{\xi }(q,v_{i})=%
\begin{cases}
\max \{\omega _{\xi }(q^{\prime },v_{i}),\ W_{\xi }(\pi _{x,v_{i}},x_{0})\}
& \text{if }q^{\prime }=r_{x_{0}}\text{, for some }x_{0}\in V(G_{\xi
}(v_{i})) \\ 
\omega _{\xi }(q^{\prime },v_{i}) & \text{otherwise}%
\end{cases}%
.  \label{omega-eq-3}
\end{equation}%
Since the value of $\omega _{\xi }(q,v_{i})$ can be computed by Eq.~(\ref%
{omega-eq-3}) in $O(1)$ time, all these computations of the values $\{\omega
_{\xi }(q,v_{i}):q\in Q,\ l_{v_{i}}<q\leq r_{v_{i}}\}$ (for a fixed $\xi $
and a fixed $v_{i}$) can be performed in~$O(|N(v_{i})|)$ time in total.
Thus, since $|\Xi |=O(\kappa )$ and $\sum\nolimits_{i\in \lbrack
n]}|N(v_{i})|=O(\kappa n)$, we can compute all values $\omega _{\xi
}(q,v_{i})$ in total $O(\kappa ^{2}n)$ time. That is, all computations
performed in the lines~\ref{alg-line-8},~\ref{alg-line-14}, and~\ref%
{alg-line-16} of the algorithm can be implemented to run in total $O(\kappa
^{2}n)$ time.

In the remainder of the proof we assume that each of the lines~\ref%
{alg-line-8},~\ref{alg-line-14}, and~\ref{alg-line-16} is executed in~$O(1)$
time. Each of the lines~\ref{alg-line-4},~\ref{alg-line-8}, and~\ref%
{alg-line-9} is executed for every $i\in \{0,1,\ldots ,n\}$ and at most for
every $\xi \in \Xi $, i.e.,~$O(\kappa n)$ times in total. Furthermore, each
of the lines~\ref{alg-line-6},~\ref{alg-line-11},~\ref{alg-line-12}, and~\ref%
{alg-line-14} is executed at most for every $\xi \in \Xi $ and for every
pair $\{v_{i},y\}$ of adjacent vertices in~$G$, i.e.,~$O(\kappa ^{2}n)$
times in total. Each of the lines~\ref{alg-line-16}-\ref{alg-line-17} is
executed at most for every $\xi \in \Xi $, for every $\zeta \in \Xi $, and
for every pair $\{v_{i},y\}$ of adjacent vertices in~$G$, i.e.,~$O(\kappa
^{3}n)$ times in total.

Finally, once we have computed all values $W_{\xi }(v_{i},v_{i})$ and $%
W_{\xi }(v_{i},y)$ in lines~\ref{alg-line-2}-\ref{alg-line-17}, the output
of line~\ref{alg-line-18} can be computed in~$O(\kappa n)$ time by
considering the $O(\kappa n)$ computed values $W_{l_{v_{0}}}(v_{i},v_{i})$
and $W_{l_{v_{0}}}(v_{i},y)$, for every vertex $v_{i}$ and for at most each
pair of adjacent vertices $v_{i},y$. Summarizing, Algorithm~\ref%
{weighted-path-special-interval-graph-alg} can be implemented to run in
total $O(\kappa ^{2}n+\kappa ^{3}n+\kappa n)=O(\kappa ^{3}n)$ time.
\end{proof}

\subsection{The general algorithm\label{general-algorithm-subsec}}

Here we combine all our results of Sections~\ref{data-reductions-sec},~\ref%
{special-interval-sec}, and~\ref{special-algorithm-subsec} to present our 
\emph{parameterized linear-time} algorithm for \textsc{Longest Path on
Interval Graphs}. The parameter $k$ of this algorithm is the size of a \emph{%
minimum proper interval deletion set} $D$ of the input graph~$G$ and its
running time has a \emph{polynomial dependency} on $k$.

\begin{theorem}
\label{general-longest-path-algorithm-thm}Let $G=(V,E)$ be an interval
graph, where $|V|=n$ and $|E|=m$, and let $k$ be the minimum size of a
proper interval deletion set of~$G$. Let $\mathcal{I}$ be an interval
representation of~$G$ whose endpoints are sorted increasingly. Then:

\begin{enumerate}
\item a proper interval deletion set $D$, where $|D|\leq 4k$, can be
computed in ${O(n+m)}$ time,

\item a semi-proper interval representation $\mathcal{I}^{\prime }$ of~$G$
can be constructed in~$O(n+m)$ time, and

\item given $D$ and $\mathcal{I}^{\prime }$, a longest path of~$G$ can be
computed in~$O(k^{9}n)$ time.
\end{enumerate}
\end{theorem}

\begin{proof}
The first two statements of the theorem follow immediately by Theorems~\ref%
{proper-deletion-set-approximation-thm} and~\ref%
{interval-representation-preprocessing-thm}, respectively. For the remainder
of the proof we assume that the proper interval deletion set $D$ and the
semi-proper interval representation $\mathcal{I}^{\prime }$ of~$G$ have been
already computed.

For the third statement of the theorem, we first compute the weighted
interval graph $G^{\#}=(V^{\#},E^{\#})$ in~$O(n)$ time by Theorem~\ref%
{first-reduction-computation-thm}. Then, given the graph~$G^{\#}$, we
compute the weighted interval graph $\widehat{G}=(\widehat{V},\widehat{E})$
in~$O(k^{2}n)$ time by Theorem~\ref{special-weighted-interval-graph-thm}. By
Theorem~\ref{special-weighted-interval-graph-thm}, this graph $\widehat{G}$
is a special weighted interval graph with parameter $\kappa =O(k^{3})$, cf.
Definition~\ref{special-interval-def}. During the computation of the graph $%
\widehat{G}$, we can compute in the same time (i.e.,~in~$O(k^{2}n)$ time)
also a special vertex partition $\widehat{V}=A\cup B$ of its vertex set.
Furthermore, it follows by Corollaries~\ref{G-diesi-correctness-cor} and~\ref%
{G-tilde-correctness-cor} that the maximum number of vertices of a path in
the initial interval graph $G$ is equal to the maximum weight of a path in
the special weighted interval graph $\widehat{G}$. Therefore, in order to
compute a longest path in~$G$ it suffices to compute a path of maximum
weight in~$\widehat{G}$. Thus, since $\widehat{G}$ is a special weighted
interval graph with parameter $\kappa =O(k^{3})$, we compute the maximum
weight of a path in~$\widehat{G}$ by Algorithm~\ref%
{weighted-path-special-interval-graph-alg}. The running time of Algorithm %
\ref{weighted-path-special-interval-graph-alg} with input $\widehat{G}$ is $%
O(\kappa ^{3}n)=O(k^{9}n)$ by Theorem~\ref%
{special-algorithm-running-time-thm}.
\end{proof}

\section{Kernelization of Maximum Matching\label{matching-kernelization-sec}}

For the sake of completeness, in this section we present the details of the
algorithm for \textsc{Maximum Matching}\footnote{%
Given a graph~$G$, find a maximum-cardinality matching in~$G$; in its
decision version, additionally the desired matching size~$k$ is specified as
part of the input.} that we sketched in Section~\ref{introduction-sec}. The
parameter $k$ is the solution size; for this parameter we show that a \emph{%
kernel} with at most $O(k^{2})$ vertices and edges can be computed in $O(kn)$
time, thus leading to a total running time of $O(kn+k^{3})$. Hence, \textsc{%
Maximum Matching}, parameterized by the solution size, belongs to the class
PL-FPT. First we present two simple data reduction rules, very similar in
spirit to the data reduction rules of Buss for \textsc{Vertex Cover} (see
e.g.,~\cite{DowneyF13,Niedermeierbook06}).

\begin{redrule}
\label{red-rule-1}If $\deg(v)>2(k-1)$ for some vertex $v\in V(G)$, then
return the instance $(G\setminus \{v\},k-1)$. 
\end{redrule}

\begin{redrule}
\label{red-rule-2}If $\deg (v)=0$ for some vertex $v\in V(G)$, then return
the instance $(G\setminus \{v\},k)$. 
\end{redrule}

An instance of parameterized \textsc{Maximum Matching} is called \emph{%
reduced} if none of Reduction Rules~\ref{red-rule-1} and~\ref{red-rule-2}
can be applied to this instance. It can be easily checked that Reduction
Rule~\ref{red-rule-2} is safe. In the next lemma we show that Reduction Rule~%
\ref{red-rule-1} is also safe.

\begin{lemma}
\label{rule-1-safe-lem}Let $k$ be a positive integer and $G$ be a graph. If $%
v\in V(G)$ such that $\deg (v)>2(k-1)$, then $(G\setminus \{v\},k-1)$ is a {%
\textsc{yes}-instance} if and only if $(G,k)$ is a {\textsc{yes}-instance}.
\end{lemma}

\begin{proof}
We first show that if $(G,k)$ is a {\textsc{yes}-instance} then $(G\setminus
\{v\},k-1)$ is a {\textsc{yes}-instance}. For this, let $M $ be a matching
of~$G$ of size at least~$k$. If $v\notin V(M)$, then $V(M)\subseteq
V(G\setminus \{v\})$ and hence $M\subseteq E(G\setminus \{v\})$. Therefore, $%
M$ is a matching of~$G\setminus \{v\}$ of size at least $k$ and thus $%
(G\setminus \{v\},k-1)$ is a {\textsc{yes}-instance}. If $v\in V(M)$, then
there exists a unique edge $e\in M$ such that $e=uv$. This implies that $%
V(M\setminus \{e\})\subseteq V(G\setminus \{v\})$ and $M\setminus
\{e\}\subseteq E(G\setminus \{v\})$. Therefore, $M^{\prime }=M\setminus
\{e\} $ is a matching of~$G\setminus \{v\}$ and $|M^{\prime }|=|M|-1\geq k-1$%
. Thus, $(G\setminus \{v\},k-1)$ is again a {\textsc{yes}-instance}.

We now show that if $(G\setminus \{v\},k-1)$ is a {\textsc{yes}-instance},
then $(G,k)$ is also a {\textsc{yes}-instance}. Let $M^{\prime }$ be a
matching of~$G\setminus \{v\}$ of size at least $k-1$. Note that, since $%
G\setminus \{v\}\subseteq G$, any matching of~$G\setminus \{v\}$ is also a
matching of~$G$. If $|M^{\prime }|\geq k$, then $(G,k)$ is clearly a {%
\textsc{yes}-instance, since }$M^{\prime }$ is a matching of~$G$ of size at
least $k$. Suppose now that $|M^{\prime }|=k-1$, that is,~$|V(M^{\prime
})|=2(k-1)$. Then, since $\deg (v)>2(k-1)$ in the graph $G $ by assumption,
there exists at least one vertex $u\in N(v)\setminus V(M^{\prime })$. Thus,
since also $v\notin M^{\prime }$ (as $M^{\prime }$ is a matching of~$%
G\setminus \{v\}$), it follows that the edge set $M=M^{\prime }\cup \{uv\}$
is a matching of~$G$ and $|M|=|M^{\prime }\cup \{uv\}|=|M^{\prime }|+1=k$.
Thus, $(G,k)$ is a {\textsc{yes}-instance}.
\end{proof}

\medskip

In the following, $\mathbf{mm}(G)$ denotes the size of a maximum matching of
graph~$G$. Furthermore, for every subset $S\subseteq V$ we denote $%
N(S)=\bigcup_{v\in S}N(v)$.

\begin{lemma}
\label{matching-bound-kernel-lem} Let $G$ be a graph. If $1\leq \deg (v)\leq
2(k-1)$ for every $v\in V(G)$, then $|E(G)|\leq (4k-5)\cdot \mathbf{mm}(G)$
and $|V(G)|\leq (4k-4)\cdot \mathbf{mm}(G)$.
\end{lemma}

\begin{proof}
Let $m_{0}=\mathbf{mm}(G)$ and let $M$ be a maximum matching of~$G$, i.e.~$%
|M|=m_{0}$. Then $V(M)$ is a vertex cover of~$G$, and thus $v\in N(V(M))$
for every vertex  $v\notin V(M)$. Note that $|V(M)|=2m_{0}$. Therefore, as $%
\deg (v)\leq 2k-2$ for every $v\in V(G)$ by the assumption of the lemma, it
follows that both the number of edges and the number of vertices of $G$ that
do not belong to $M$ are at most $2m_{0}(2k-3)$. Thus $|E(G)|\leq
m_{0}+2m_{0}(2k-3)=(4k-5)m_{0}$. Similarly, $|V(G)|\leq
2m_{0}+2m_{0}(2k-3)=(4k-4)m_{0}$.
\end{proof}

\medskip

Now we are ready to provide our kernelization algorithm for \textsc{Maximum
Matching}, together with upper bounds on its running time and on the size of
the resulting kernel.

\begin{theorem}
\label{kernel-thm}\textsc{Maximum Matching}, when parameterized by the
solution size $k$, admits a kernel with at most $O(k^{2})$ vertices and at
most $O(k^{2})$ edges. For an $n$-vertex graph the kernel can be computed in~%
$O(kn)$ time.
\end{theorem}

\begin{proof}
Let $(G,k)$ be an instance of parameterized \textsc{Maximum Matching}. Our
kernelization algorithm either returns {\textsc{yes}}, or it computes an
equivalent reduced instance $(G^{\prime },k^{\prime })$.

First, we exhaustively apply Reduction Rule~\ref{red-rule-1} by visiting
every vertex once and removing every vertex of degree greater than $2(k-1)$
in the current graph. Notably, since vertex removals can only reduce the
degree of the remaining vertices, the algorithm does not need to visit any
vertex twice. If we construct an instance $(G^{\prime },0)$ during this
procedure, that is, if we remove $k$ vertices from $G$, then we stop and
return {\textsc{yes}}. The correctness of this decision follows immediately
by the facts that $(G^{\prime },0)$ is clearly a {\textsc{yes}-instance }and
Reduction Rule~\ref{red-rule-1} is safe by Lemma~\ref{rule-1-safe-lem}.

The exhaustive application of Reduction Rule~\ref{red-rule-1} can be
implemented to run in $O(nk)$ time, as follows. Every time we discover a new
vertex $v$ with $\deg (v)>2(k-1)$ in the current graph (and for the current
value of the parameter~$k$), then we do not actually remove $v$ from the
current graph but we mark it as ``removed'' and we proceed to the next
vertex. Furthermore we keep in a counter~$r$ the number of vertices that
have been marked so far as ``removed''. Note that, to check whether we need
to apply Reduction Rule~\ref{red-rule-1} on a vertex~$v$, we only need to
visit at most all \emph{marked} neighbors of $v$ and at most $2(k-r-1)+1<2k$ 
\emph{unmarked} neighbors in the initial graph $G$. Thus, since there exist
at every point at most $r<k$ marked vertices, we only need to check less
than $3k$ neighbors of~$v$ in the initial graph $G$ to decide whether we
mark $v$ as a new ``removed'' vertex. Thus, since there are $n$ vertices in
total, the whole procedure runs in $O(kn)$ time. Denote by $r_{0}$ the total
number of vertices that have been marked as ``removed'' at the end of this
process.

Next, we exhaustively apply Reduction Rule~\ref{red-rule-2} by removing
every unmarked vertex $v$ that has only marked neighbors in $G$. Since such
a vertex $v$ remained unmarked during the exhaustive application of
Reduction Rule~\ref{red-rule-1}, $v$~has less than~$k$ \emph{marked}
neighbors and less than $2k$ \emph{unmarked} neighbors in $G$, that is, at
most $3k$ neighbors in total. Thus we can check in $O(k)$ time whether a
currently unmarked vertex has only marked neighbors; in this case we mark $v$
as \textquotedblleft removed\textquotedblright . This process can be clearly
done in $O(nk)$ time. Let $G^{\prime }$ be the induced subgraph of $G$ on
the unmarked vertices and let $k^{\prime }=k-r_{0}$. Note that every vertex
of $G^{\prime }$ has at least one and at most $2(k^{\prime }-1)$ neighbors
in $G^{\prime }$.

Finally we count the number of vertices and edges of $G^{\prime }$ in $O(kn)$
time. This can be done by visiting again all unmarked vertices $v$ and their
unmarked neighbors in $G$. If $G^{\prime }$ has strictly more than $%
(k^{\prime }-1)(4k^{\prime }-5)$ edges or more than $(k^{\prime
}-1)(4k^{\prime }-4)$ vertices, then we stop and return {\textsc{yes}}.
Otherwise the kernelization algorithm returns the kernel $(G^{\prime
},k^{\prime })$, which has $O(k^{2})$ vertices and $O(k^{2})$ edges.
Consequently, the kernelization algorithm runs in~$O(kn)$ time in total. It
remains to prove that, if $|E(G^{\prime })|>(k^{\prime }-1)(4k^{\prime }-5)$
or $|V(G^{\prime })|>(k^{\prime }-1)(4k^{\prime }-4)$, then $(G^{\prime
},k^{\prime })$ is a \textsc{yes}-instance. Assume otherwise that $%
(G^{\prime },k^{\prime })$ is a \textsc{no}-instance, that is, $\mathbf{mm}%
(G^{\prime })\leq k^{\prime }-1$. Then it follows by Lemma~\ref%
{matching-bound-kernel-lem} that $|E(G^{\prime })|\leq (k^{\prime
}-1)(4k^{\prime }-5)$ and $|V(G^{\prime })|\leq (k^{\prime }-1)(4k^{\prime
}-4)$, which is a contradiction. This completes the proof of the theorem.
\end{proof}

\medskip

Applying the matching algorithm due to Micali and Vazirani~\cite{MicaliV80}
to the kernel we obtained by Theorem~\ref{kernel-thm}, we achieve the
following result.

\begin{corollary}
\label{matching-algorithm-from-kernel-cor}\textsc{Maximum Matching}, when
parameterized by solution size~$k$, can be solved in~$O(nk+k^{3})$ time.
\end{corollary}

\begin{proof}
Let $(G,k)$ be an instance of parameterized \textsc{Maximum Matching}, where 
$k$ is the solution size. First we apply to~$(G,k)$ the kernelization
algorithm of Theorem~\ref{kernel-thm}, which returns either {\textsc{yes} or
an equivalent instance }$(G^{\prime \prime },k^{\prime \prime })$ {with }$%
O(k^{2})$ vertices and $O(k^{2})$ edges. Then we compute a maximum matching $%
\mathbf{mm}(G^{\prime \prime })$ of the graph $G^{\prime \prime }$ using any
of the known algorithms, e.g.,~the algorithm of Micali and Vazirani~\cite%
{MicaliV80}. It computes $\mathbf{mm}(G^{\prime \prime })$ in~$O\left(
|E(G^{\prime \prime })|\cdot \sqrt{|V(G^{\prime \prime })|}\right) =O(k^{3})$
time. Finally, if $\mathbf{mm}(G^{\prime \prime })\geq k^{\prime \prime }$,
then return {\textsc{yes}}, otherwise return \textsc{no}. In total, \textsc{%
Maximum Matching} can be thus solved in~$O(nk+k^{3})$ time. 
%on the instance $(G,k)$.
\end{proof}

\section{Outlook and Discussion\label{outlook-sec}}

Our work heads at stimulating a general research program which
systematically exploits the concept of fixed-parameter tractability for
polynomially solvable problems. 
For several fundamental and widely known problems, the time complexities of
the currently fastest algorithms are upper-bounded by polynomials of large
degrees. One of the most prominent examples is arguably the celebrated
polynomial-time recognition algorithm for \emph{perfect graphs}, whose time
complexity still remains $O(n^{9})$~\cite{Chudnovsky05}. Apart from trying
to improve the \emph{worst-case} time complexity for such problems, which
may be a very difficult (if not impossible) task, the complementary approach
that we propose here is to try to spot a \emph{parameter} that causes these
high-degree polynomial-time algorithms and to separate the dependency of the
time complexity from this parameter such that the dependency on the input
size becomes as close to linear as possible. We believe that the
``FPT~inside~P'' field is very rich and offers plenty of research
possibilities.

We conclude with three related topics that may lead to further interactions.
First, we remark that in classical parameterized complexity analysis there
is a growing awareness concerning the polynomial-time factors that often
have been neglected~\cite{Bev14}. Notably, there are some prominent
fixed-parameter tractability results giving \emph{linear-time} factors in
the input size (but quite large exponential factors in the parameter); these
include Bodlaender's famous ``linear-time'' algorithm for computing
treewidth~\cite{Bodlaender96} and the more recent ``linear-time'' algorithm
for computing the crossing number of a graph~\cite{KawarabayashiR07}.
Interestingly, these papers emphasize ``linear time'' in their titles,
instead of ``fixed-parameter tractability''. In this spirit, our result for 
\textsc{Longest Path in Interval graphs} is a ``linear-time'' algorithm
where the dependency on the parameter is not exponential~\cite{Bodlaender96,KawarabayashiR07} but \emph{polynomial}. 
In this line of research, Fomin et al.~studied graph and matrix problems on instances with
small treewidth. In particular the authors presented, among other results,
an $O(k^3 n \log n)$ randomized algorithm for computing the cardinality of a
maximum matching and an $O(k^4 n \log^2 n)$ randomized algorithm for
actually constructing a maximum matching, where $k$ is an upper bound for
the treewidth of the given graph~\cite%
{FominLPSW_Matching-Treewidth_SODA17}.

Second, polynomial-time solvability and the corresponding \emph{lower bounds}
have been of long-standing interest, e.g., it is believed that the famous
3SUM problem is only solvable in quadratic time and this conjecture has been
employed for proving relative lower bounds for other problems~\cite%
{GajentaanO95}. Very recently, there was a significant push in this research
direction with many new relative lower bounds~\cite%
{AbboudW14,AbboudGW15,Bringmann14}. The ``FPT~inside~P'' approach might help
in ``breaking'' these nonlinear relative lower bounds by introducing useful
parameterizations and striving for PL-FPT~results. In this direction an
interesting negative result appeared very recently by Abboud et al.~\cite%
{Abboud-VW-W-SODA16} who proved that, unless some
plausible complexity assumptions fail, for any $\varepsilon>0$ there does
not exist any algorithm with running time $2^{o(k)}n^{2-\varepsilon}$ for $(%
\frac{3}{2}-\delta)$-approximating the diameter or the radius of a graph,
where $k$ is an upper bound for the treewidth. In contrast, the authors
proved that both the diameter and the radius can be computed in $2^{O(k \log
k)}n^{1+o(1)}$ time~\cite{Abboud-VW-W-SODA16}.

Finally, coming back to a practical motivation for ``FPT inside~P'', it has
been very recently observed that identifying various parameterizations for
the same problem may help in designing meta-algorithms that (dynamically)
select the most appropriate solution strategy (also specified by respective
parameters)---this approach is known as ``programming by optimization''~\cite%
{Hoos12}. Note that so far this line of research is still in its infancy
with only one known study~\cite{HartungH15} for NP-hard problems; following
this approach might also be promising within our ``FPT inside~P'' framework.

\paragraph*{Acknowledgments.}

The authors wish to thank Jasper Slusallek for helpful discussions and for suggesting an improvement for the bounds given in Lemma~31.

\providecommand{\noopsort}[1]{}

\end{document}